%% file: main.tex
\documentclass[12pt]{article}
\usepackage{geometry}
\usepackage[T1]{fontenc}
\usepackage[utf8]{inputenc}
\usepackage{enumerate,enumitem}
\usepackage{natbib}
\usepackage{url} 
\usepackage{bbm}
\usepackage{float}
\usepackage{xspace}
\usepackage{comment}
\usepackage{diagbox}
\usepackage{multirow, multicol}
\usepackage{boxedminipage}
\usepackage[dvipsnames]{xcolor}
\usepackage{authblk}

\usepackage{natbib}
\usepackage{wrapfig}
\usepackage{xr}
\usepackage{booktabs}
\usepackage{algorithm,algpseudocode}

\usepackage{adjustbox}
\usepackage{setspace}
\usepackage{graphicx,psfrag,epsf,subcaption}
\usepackage{amssymb,amsthm,amsmath,mathrsfs,bbm,nicefrac}

\usepackage{booktabs}
\usepackage[bookmarks,bookmarksnumbered,%
    allbordercolors={0.8 0.8 0.8},%
    linktocpage%
    ]{hyperref} 
\hypersetup{colorlinks,
  citecolor={RoyalBlue},linkcolor={ForestGreen},urlcolor={BrickRed}}
\usepackage{textcomp}
\usepackage{times}
\input{mathcommands}

\title{Byzantine-tolerant distributed learning of finite mixture models}
\date{}
\author[1]{Qiong Zhang\thanks{Co-first author. Correspondence to: qiong.zhang@ruc.edu.cn}}
\author[2]{Yan Shuo Tan$^*$}
\author[3]{Jiahua Chen}
\affil[1]{\small Institute of Statistics and Big Data, Renmin University of China, Beijing, China}
\affil[2]{\small Department of Statistics and Data Science, National University of Singapore, Singapore}
\affil[3]{\small Department of Statistics, The University of British Columbia, Vancouver, Canada}

\newcommand{\ie}{{\em i.e.,~}}

\begin{document}

\maketitle

\begin{abstract}
Traditional statistical methods need to be updated to work with modern distributed data storage paradigms. 
The split-and-conquer framework that learns models on local machines and averaging their parameter estimates is common. 
However, this does not work for the important problem of learning finite mixture models, because subpopulation indices on each local machine may be arbitrarily permuted (the “label switching problem”). 
\citet{zhang2022distributed} proposed Mixture Reduction (MR) to address this issue, offering an effective and efficient solution for aligning and aggregating local mixture components. 
Building upon this foundation, this paper considers the additional challenge of Byzantine failure, where a fraction of local machines may transmit arbitrarily erroneous information. 
We introduce Distance-Filtered Mixture Reduction (DFMR), a Byzantine-tolerant framework that enhances MR by adding a distance-based filtering mechanism to identify and exclude corrupted local estimates before the MR aggregation. 
This integration allows DFMR to maintain MR’s efficiency while achieving strong robustness against Byzantine failure. 
We provide theoretical justification for DFMR, proving its optimal convergence rate and asymptotic equivalence to the global maximum likelihood estimate under standard assumptions. 
Numerical experiments on simulated and real-world data validate the effectiveness of DFMR in achieving robust and accurate aggregation in the presence of Byzantine failure.
The code can be found at \url{https://github.com/SarahQiong/RobustSCGMM}.
\end{abstract}


\input{sections/intro}
\input{sections/literature}
\input{sections/preliminary}

\input{sections/method}
\input{sections/theory}
\input{sections/expt}
\input{sections/real_data}
\input{sections/conclusion}

\section*{Funding}
Qiong Zhang is supported by the National Natural Science Foundation of China Grant 12301391.
Yan Shuo Tan is supported by NUS Startup Grant A-8000448-00-00 and the Singapore Ministry of Education (MOE) AcRF Tier 1 Grants A-8002498-00-00 and A-8004458-00-00.
Jiahua Chen is supported by the Natural Sciences and Engineering Research Council of Canada
under Grant RGPIN-2025-03989.

\bibliographystyle{apalike}
{\small \bibliography{biblio}}

\clearpage
\newpage
\appendix
\input{sections/appendix_algorithms}

\input{sections/appendix_theory}
\input{sections/appendix_exp}

\end{document}

%% file: mathcommands.tex
\usepackage{amsmath,amsfonts,bm,mathtools}

\providecommand{\matrixnorm}[1]{\left|\!\left|\!\left|{#1}
  \right|\!\right|\!\right|} 

\def\vu{{\bm{u}}}
\def\vv{{\bm{v}}}
\def\vw{{\bm{w}}}
\def\vx{{\bm{x}}}

\def\mG{{\bm{G}}}
\def\mH{{\bm{H}}}
\def\mI{{\bm{I}}}

\def\mS{{\bm{S}}}

\DeclareMathAlphabet{\mathsfit}{\encodingdefault}{\sfdefault}{m}{sl}
\SetMathAlphabet{\mathsfit}{bold}{\encodingdefault}{\sfdefault}{bx}{n}

\def\gC{{\mathcal{C}}}

\def\gF{{\mathcal{F}}}
\def\gG{{\mathcal{G}}}

\def\gN{{\mathcal{N}}}

\def\gX{{\mathcal{X}}}

\def\sA{{\mathbb{A}}}
\def\sB{{\mathbb{B}}}

\def\sE{{\mathbb{E}}}

\def\sG{{\mathbb{G}}}

\def\sP{{\mathbb{P}}}

\def\sR{{\mathbb{R}}}
\def\sS{{\mathbb{S}}}

\newcommand{\R}{\mathbb{R}}

\newcommand{\KL}{D_{\mathrm{KL}}}

\newcommand{\Var}{\mathrm{Var}}

\DeclareMathOperator*{\argmax}{arg\,max}
\DeclareMathOperator*{\argmin}{arg\,min}

\DeclareMathOperator*{\argsup}{arg\,sup}

\theoremstyle{thmstyleone}%
\newtheorem{theorem}{Theorem}[section]
\newtheorem{proposition}[theorem]{Proposition}%
\newtheorem{lemma}[theorem]{Lemma}
\newtheorem{corollary}[theorem]{Corollary}
\theoremstyle{thmstyletwo}%
\newtheorem{remark}{Remark}[section]%
\theoremstyle{thmstylethree}%

\newtheorem{definition}{Definition}[section]
\newtheorem{assumption}{Assumption}[section]

\newcommand{\bea}{\begin{eqnarray*}}
\newcommand{\eea}{\end{eqnarray*}}
\newcommand{\ba}{\begin{eqnarray*}}
\newcommand{\ea}{\end{eqnarray*}}
\newcommand{\be}{\begin{equation}}
\newcommand{\ee}{\end{equation}}
\newcommand{\bi}{\begin{itemize}}
\newcommand{\ei}{\end{itemize}}

\makeatletter
\newcommand*\rel@kern[1]{\kern#1\dimexpr\macc@kerna}
\newcommand*\widebar[1]{%
  \begingroup
  \def\mathaccent##1##2{%
    \rel@kern{0.8}%
    \overline{\rel@kern{-0.8}\macc@nucleus\rel@kern{0.2}}%
    \rel@kern{-0.2}%
  }%
  \macc@depth\@ne
  \let\math@bgroup\@empty \let\math@egroup\macc@set@skewchar
  \mathsurround\z@ \frozen@everymath{\mathgroup\macc@group\relax}%
  \macc@set@skewchar\relax
  \let\mathaccentV\macc@nested@a
  \macc@nested@a\relax111{#1}%
  \endgroup
}
\makeatother

\DeclarePairedDelimiter{\braces}{\lbrace}{\rbrace}
\DeclarePairedDelimiter{\paren}{(}{)}
\DeclarePairedDelimiter{\norm}{\|}{\|}
\DeclarePairedDelimiter{\abs}{|}{|}
\DeclarePairedDelimiter{\sqbracket}{[}{]}

\newcommand{\selected}{\mathbb{S}}

\newcommand{\dfmr}{\operatorname{DFMR}}
\newcommand{\omr}{\operatorname{OMR}}
\newcommand{\mr}{\operatorname{MR}}
\newcommand{\coat}{\operatorname{COAT}}

%% file: sections/intro.tex

\section{Introduction}
Mixture models are powerful tools for representing data with latent subpopulations.
They trace back to the seminal work of~\citet{pearson1894contributions}, who applied them to infer the existence of subspecies in crabs.
These models have been applied in diverse fields such as finance~\citep{liesenfeld2001generalized}, astronomy~\citep{baldry2004color}, image analysis~\citep{salimans2017pixelcnn++}, biology, and medicine over the years.
In medical research, for example, they have been used to model patient lifespans with lupus nephritis~\citep{marin2005using}, delineate breast cancer tumour stages~\citep{fraley2002model}, and identify post-stem cell transplantation states~\citep{baudry2010combining}.

As the applications of mixture models have expanded, so too have the challenges associated with their implementation. 
In particular, methods for fitting mixture models must now account for the massive scale of modern datasets, evolving paradigms of data storage, and growing concerns around security and robustness.
When the size of a dataset exceeds the capacity of a single device, it has to be stored in a distributed fashion across multiple local machines.
To fit statistical and machine learning models on such data, several distributed learning paradigms have emerged, one of which is the \emph{Split-and-Conquer (SC)} framework.
Under this framework, a dataset of size $N$ may be regarded as being randomly divided into $m$ non-overlapping chunks, each of size $n$, and distributed across $m$ local machines, with each machine processing a single chunk.
Each machine independently fits a model on its own local chunk and communicates the resulting local estimate to a central server, which aggregates the local estimates into a global one.
In recent years, SC adaptations have been developed for a wide range of statistical procedures. 
These include generalised linear models~\citep{chen2014split}, kernel ridge regression~\citep{zhang2015divide}, local average regression~\citep{chang2017divide}, Wald and score tests~\citep{battey2018distributed}, and estimation of principal eigenspaces~\citep{fan2019distributed}. 
Additionally, SC methods have been applied to the minimisation of additive smooth and locally convex loss functions~\citep{jordan2019communication, fan2023communication, wu2023quasi}, and empirical likelihood estimation~\citep{wang2024distributed}.

Unfortunately, the SC framework, and distributed learning in general, is vulnerable to \emph{Byzantine failure}~\citep{lamport1982byzantine}. 
This occurs when hardware or software malfunction, data corruption, or communication disruptions introduce unpredictable inaccuracies during the transmission of information from local machines to the central server.\footnote{The cloud platform Cloudflare experienced a real-world Byzantine failure on November 2, 2020, where a single misbehaving router in a data centre caused widespread network issues by advertising incorrect routes. This led to significant packet loss and service disruptions~\citep{cloudflare_byzantine_failure}.}
Mathematically, Byzantine failure is usually modelled as an unknown $\alpha < 1/2$ fraction of machines transmitting arbitrary values instead of their authentic local estimates.
It is well-known that this failure mode can lead to catastrophic errors when local estimates are aggregated via simple averaging, but can be mitigated using tools from robust statistics.
Indeed, researchers have proposed aggregating local estimates using robust alternatives such as the trimmed mean, coordinate-wise median~\citep{yin2018byzantine,yin2019defending}, geometric median~\citep{feng2014distributed,chen2017distributed} and variance reduced median-of-means~\citep{tu2021variance}.

Despite these advances, Byzantine-tolerant distributed learning for finite mixture models pose unique challenges and remain largely unexplored. 
A key difficulty arises from the ``label switching problem''--the indices of mixture components can be permuted without affecting the overall distribution, making it unclear how to align components across local estimators before aggregation. 
The \emph{mixture reduction} (MR) framework in~\citet{zhang2022distributed} effectively resolves this issue and enables efficient aggregation of local mixture estimates through an optimal transport formulation. 
It therefore provides a principled and computationally efficient foundation for distributed learning of mixture models.

In this paper, we first strengthen the theoretical foundation of MR.
We prove, under standard assumptions, that when the local estimates are maximum likelihood estimates (MLEs), the MR estimator converges at the optimal $O_{P}(N^{-1/2})$ rate whenever $m = O(n)$ and, furthermore, is asymptotically equivalent to the global MLE whenever $m = o(n)$.
This result provides a rigorous theoretical guarantee even when the number of local machines grows with the total sample size $N$, a regime of increasing practical importance in large-scale distributed environments.

Building on this theoretical foundation, we further propose a Byzantine-tolerant extension of MR--our key methodological contribution.
While MR effectively resolves the label-switching issue, it remains susceptible to Byzantine failures, as the optimal transport formulation is not intrinsically robust to arbitrary corruptions in local estimates.
At the same time, existing robust aggregation techniques~\citep{yin2018byzantine,yin2019defending,tu2021variance} are primarily developed for Euclidean parameter spaces and cannot be directly applied to mixing distributions, which lie on a non-Euclidean space.
Other approaches, such as the centre of attention (COAT) estimator~\citep{nissim2007smooth}, originally introduced in the context of differential privacy, effectively select a single representative local estimator.
Although COAT can in principle be extended to non-Euclidean settings, such methods do not exploit the full collection of local information and, as we demonstrate empirically, can perform poorly for mixture models under contamination.
To overcome these challenges, we design a novel robust aggregation procedure that adapts MR to Byzantine failure settings.
In particular, when a subset of local estimates is corrupted by Byzantine failures, our method identifies such anomalies through a distance-based filtering mechanism prior to aggregation, thereby enabling accurate and reliable distributed learning for mixture models.
The proposed approach thus advances both the robustness and practical applicability of distributed finite mixture modelling.

Our approach makes use of the densities of local estimates and not just their parameter vectors.
We prove that the squared $L^2$ distance between the density of an authentic local estimate and that of the true model is asymptotically a quadratic form in the difference between their parameter vectors, and hence converges in distribution to a generalized chi-squared distribution.
As the number of local machines $m$ also tends to infinity, the empirical distribution of these distances also approaches the limiting distribution.
These facts allow us to use the observed pairwise distances between local estimates to construct a filter that simultaneously contains an overwhelming majority of the uncorrupted estimates while removing all estimates that have experienced severe corruption.
More precisely, we first identify a local estimate that is highly central and then keep only a set of local estimates within some $L^2$ distance from the initial estimate.
Mixture reduction can then be applied on the set of filtered local estimates.
We call this procedure \emph{Distance Filtered Mixture Reduction (DFMR)}.

Our proposed DFMR approach is both computationally feasible and statistically sound.
\emph{Computationally}, DFMR remains highly efficient.
Its main computational burden arises from the filtering stage, which identifies non-corrupted local estimators by comparing their density functions.
Specifically, DFMR retains local estimators whose densities lie within an $L^2$ distance threshold $\rho r^{\text{init}}$ from a initial estimator, where $r^{\text{init}}$ measures the typical dispersion among local estimates and $\rho \ge 1$ is a user-chosen inflation factor controlling the filter’s tolerance.
This step requires evaluating all $m(m-1)/2$ pairwise $L^2$ distances between local mixture densities.
For common mixture models such as Gaussian, Gamma, or Poisson mixtures, each $L^2$ distance admits a closed-form expression in terms of model parameters,\footnote{The $L^2$ distance is a weighted sum of integrals involving products of component densities.
For many well-known parametric families, these integrals can be computed analytically.
Some examples are provided in Appendix~\ref{app:L2_closed_form}.}
which allows efficient computation.
For example, computing the $L^2$ distance between two $K$-component mixtures of $d$-dimensional Gaussians costs $O(K^2d^3)$, leading to an overall complexity of $O(m^2K^2d^3)$ for the filtering step.
The subsequent MR step is faster, with computational complexity $O(mK^2)$.\footnote{For instance, under a $5$-component Gaussian mixture in $10$ dimensions with $100$ machines, computing the pairwise distances takes approximately $25$ seconds on a MacBook with an Apple M2 Pro chip, while the reduction step takes only $0.2$ seconds.}

\emph{Statistically}, under standard regularity conditions on the finite mixture model, we show that the DFMR estimator achieves the convergence rate $O_P(N^{-1/2} + \alpha\rho n^{-1/2})$, which matches the optimal rate for Byzantine-tolerant split-and-conquer learning with Euclidean parameters~\citep{yin2018byzantine}, up to the factor $\rho$.
The inflation factor $\rho$ can be chosen as a small power of $m$ (e.g., $m^{1/10}$), reflecting the heavier tails of local mixture estimates.
Furthermore, when the corrupted local estimates are not overly concentrated, we establish that the DFMR estimator is asymptotically equivalent to an oracle estimator that performs MR on the subset of failure-free estimates—an condition satisfied under many attack models in the Byzantine-failure literature.

The remainder of the paper is organized as follows. 
Section~\ref{sec:literature} discussed related work.
Section~\ref{sec:preliminary_sc_gmm} introduces finite mixture models and discusses their SC aggregation in a Byzantine failure-free setting. 
Section~\ref{sec:byzantine_definition} defines the Byzantine failure model, Section~\ref{subsec:l2_distance_densities} analyses the $L^2$ distances between mixture densities, Section~\ref{sec:proposed_estimator_machine} presents the proposed Byzantine-tolerant robust aggregation approaches, while Section~\ref{sec:theory} details the statistical properties of these approaches.
Numerical studies using simulated data and real-world data to demonstrate the effectiveness of the proposed approaches are provided in Sections~\ref{sec:exp} and~\ref{sec:real_data}, respectively.
Section~\ref{sec:conclusion} concludes with a discussion and final remarks.

%% file: sections/literature.tex
\section{Related work}
\label{sec:literature}
In SC distributed learning, robust aggregation methods such as the trimmed mean~\citep{yin2018byzantine} and the median~\citep{minsker2019distributed} are commonly employed to combine local estimates in the presence of Byzantine failure.
Note that these, and other papers discussed in this section, consider a more general setting, where local information, such as gradients, can be shared over several rounds of communication.
When the goal is to minimize a strongly convex objective which is the expectation of a sample-wise loss, and assuming subexponential noise for each sample's gradient,~\citet{yin2018byzantine} showed that no algorithm can achieve an error lower than $\widetilde{\Omega}\left(\alpha n^{-1/2} + N^{-1/2}\right)$ regardless of the communication cost,
and that this rate can be achieved by aggregation using the coordinate-wise median when $n \geq m$.
Despite having optimal rates, the coordinate-wise median can still be improved upon in terms of efficiency, as shown for instance by~\citet{tu2021variance}.

Several works have investigated robust aggregation via distance based filtering.
For example,~\citet{blanchard2017machine}'s Krum\footnote{It is named after the Bulgarian Khan at the end of the eighth century, who undertook offensive attacks against the Byzantine empire.} method aggregates local estimates by identifying and outputting the local estimate with the minimum sum of squared distances to its $m - \zeta$ closest vectors, where $\zeta$ is a tuning parameter. 
Building on this idea,~\citet{xie2018generalized} proposed a ``mean-around-the-median'' estimator, which further averages the $m-\zeta$ closest neighbours of the Krum estimate. 
While both approaches demonstrate some degree of Byzantine tolerance, their statistical convergence rates remain unclear.
\citet{wang2024distributed} employed distance filtering in the context of distributed learning of empirical likelihood ratio statistics for the population mean. 
They used the efficient surrogate minimisation algorithm proposed by~\citet{jordan2019communication}, which requires aggregating local gradients. 
Each local gradient is selected for further aggregation if its $L^2$ distance to the majority of other gradients is below a pre-specified threshold. 
A recommended threshold is of order $n^{-1/2}\log m$ based on asymptotic analysis. 
However, the consistency of their filtering procedure is only established under the condition that the failure gradients are far from the failure-free ones.

As discussed in the introduction, due to the label-switching problem, none of these methods can be directly applied to SC learning of mixture models, whether in the single round of communication setting considered in our paper, or in the multiple-round version.
Furthermore, traditional robust aggregation techniques such as the median and trimmed mean rely on the totally ordered property of Euclidean space and hence do not have straightforward generalisations to non-Euclidean spaces, such as the space of mixing distributions.
Indeed, to the best of our knowledge, no existing work specifically addresses the issue of Byzantine failure in SC learning for finite mixtures.

Finally, we briefly describe two works that have considered similar problems.
\citet{del2019robust} introduced a trimmed $k$-barycentres method to robustly cluster non-Euclidean objects.
However, their method does not assume a Byzantine failure model and instead focuses only on how to handle outliers.
In particular, it
does not effectively address failures in mixing weights. 
Additionally,  the method lacks a clear strategy for determining an appropriate trimming level, whereas the failure rate in the Byzantine failure model is typically unknown.
\citet{tian2023unsupervised} proposed a distributed Expectation-Maximisation (EM) framework, which considers scenarios where the local datasets on different machines may be generated from different mixing distributions instead of arising from the same model. 
This approach iteratively aggregates local EM updates using penalized least squares to encourage consensus. 
While robust to a limited fraction of Byzantine failure (less than $1/3$), the statistical rate of its aggregate estimator does not improve upon that of a single-machine estimator.
Moreover, its theoretical analysis is only applicable to special classes of mixtures such as Gaussian mixtures with known subpopulation covariances.

%% file: sections/preliminary.tex

\section{Revisiting failure-free SC learning of finite mixture models}
\label{sec:preliminary_sc_gmm}
A finite mixture model is a special family of distributions defined formally as follows:
\begin{definition}[Finite mixture model]
\label{def:mixture_model}
Let $\gF = \{f(x; \theta): x \in \gX \subseteq\sR^{d}, \theta \in \Theta \subseteq \sR^{p}\}$ be density functions representing a classical parametric model, where $d$ and $p$ denote the dimensions of the observation and parameter space respectively. 
An order $K$ finite mixture of $\gF$ consists of mixture distributions, each with a density function given by:
\begin{equation*}
    f_{G}(x):= \int_{\Theta} f(x;\theta)\,dG(\theta) = \sum_{k=1}^K w_k f(x;\theta_k)
\end{equation*}
where the parameter $G=\sum_{k=1}^{K} w_k \delta_{\theta_k}$, called the mixing distribution, assigns probability $w_k$ to subpopulation parameter $\theta_k\in \Theta$, $k \in [K] =\{1,\ldots, K\}$. 
\end{definition}
Here, $f(x; \theta_k)$ represents subpopulation densities and mixing weights $(w_1,\ldots, w_K)^{\top}\in \Delta_{K-1}=\{(w_1,\ldots, w_K):w_k \geq 0, \sum_{k=1}^{K} w_k =1\}$. 
We use $\sG_{K}$ to denote the family of mixing distributions with at most $K$ distinct support points with $\sG_{0}=\emptyset$.
Note that for simplicity, we focus on distributions with densities with respect to Lebesgue measure, although both mixture models and our proposed method can be defined more generally.

\begin{remark}[The label switching problem and finite mixture model parameterisation]
\label{remark:parameterization}
One may consider using a vector such as $(w_1, w_2, \ldots, w_K, \theta_1, \theta_2, \ldots, \theta_K)^{\top}$ to parameterise the finite mixture, instead of a mixing distribution $G$. 
However, this parameterisation encounters the well-known label-switching problem, where the order of the parameters can be permuted without altering the mixture distribution. 
This leads to equivalent distributions that differ only in the ordering of their subpopulations. 
For example, $(w_K, w_{K-1}, \ldots, w_1, \theta_K, \theta_{K-1}, \ldots, \theta_1)^{\top}$ (and $K!$ such permutations in total) yields the same mixture, rendering the model non-identifiable. 
In contrast, parameterisation via the mixing distribution avoids this issue.
\end{remark}

\subsection{The mixture reduction estimator}
\label{subsec:MR_estimator}

Let $\gX = \{x_{ij}: i\in[m], j\in [n]\}$
\footnote{
For simplicity, we present the analysis assuming equal local sample sizes across machines. The algorithm itself accommodates heterogeneous sample sizes through the weights $\lambda_i = n_i / N$, where $n_i$ is the sample size on the $i$th machine. The same results extend to this more general setting if we assume a relatively balanced partition, i.e. that $\min_i n_i / \max_i n_i \geq c$ for some constant $c > 0$.} 
be $N=nm$ independently and identically distributed (IID) samples from a mixture distribution $f_{G^*}(x)$ for some $G^* \in \sG_{K} \backslash \sG_{K-1}$ with $K$ distinct subpopulations, $K \geq 2$.
Suppose $\gX$ is partitioned into $m$ subsets $\gX_1, \ldots, \gX_m$ completely at random, which are stored in a distributed manner on $m$ local machines, where $\gX_{i} = \{x_{ij}: j \in [n]\}$ for $i \in [m]$.
A general SC procedure for estimating $G^*$ consists of the following two steps:
\begin{itemize}[leftmargin=*]
\item
\emph{Local inference}: Obtain the local estimator $\widehat{G}_i = \sum_{k} \widehat{w}_{ik}\delta_{\widehat{\theta}_{ik}}$ by maximizing the log-likelihood function $\ell_i(G) = \sum_{j=1}^{n} \log f_G(x_{ij})$ or its penalized version in the case of Gaussian mixture~\citep{chen2009inference}, based on the data $\mathcal{X}_i$ on the $i$th machine. 
These local estimators are referred to as Maximum Likelihood Estimates (MLE) or penalized MLEs (pMLE) respectively. 
EM algorithms can be used for their numerical computations, see Appendix~\ref{app:em_algorithm}.

\item
\emph{Global aggregation}: Send $\{\widehat{G}_i: i \in [m]\}$ to a server and aggregate them to produce final estimator $\widehat G$.
The final estimator is
\be
\label{eq:sc_regular_aggregation}
\vspace{-13pt}
\widehat G^{\mathfrak{A}} = \mathfrak{A} (\{\widehat{G}_i: i \in [m]\} ),
\ee
where $\mathfrak{A}(\cdot)$ is some aggregation operator such as those in~\citet{zhang2022distributed}.
\end{itemize}

We now describe~\citet{zhang2022distributed}'s approach to performing aggregation via mixture reduction (MR).
Let the \emph{average mixing distribution} be
$
\widebar{G} = \sum_{i=1}^m \lambda_i \widehat{G}_i
$
where $\lambda_i = n/N$ is the proportion of data stored on $i$th machine.
The average mixing distribution assigns weight $\lambda_{i}\widehat{w}_{ik}$ at $\widehat{\theta}_{ik}$.
Because the local estimators we adopt are consistent, each $\widehat{G}_i$ is close to $G^*$, $\widebar{G}$ is also close to $G^*$.
However, $\widebar{G}$ potentially lies outside the parameter space $\sG_K$.
We hence seek a mixing distribution $G \in \sG_K$ that is closest to $\widebar{G}$, a task which can be framed as approximating a high-order mixture using a lower-order mixture.
This class of approximation problems has been previously studied in the literature, with the version pertaining to Gaussian mixtures known as Gaussian Mixture Reduction~\citep{crouse2011look}.

\cite{zhang2022distributed} proposed a novel approach to MR using a composite transportation divergence.
Specifically, let $c: \gF \times \gF \to \sR_{+}$ be a cost function, where we denote $c(\theta, \theta') = c(f(\cdot; \theta), f(\cdot; \theta') )$ for simplicity.
Let $G = \sum_{k=1}^{K} w_k \delta_{\theta_k}$ and  $G' = \sum_{k'=1}^{K'} w'_{k'} \delta_{\theta'_{k'}}$ be two mixing distributions with $K$ and $K'$ support points.
The composite transportation divergence between these mixing distributions is defined to be
\begin{equation}
\label{eq:CTD_def}
    T_{c}(G, G') = \min \left \{
    \sum_{k=1}
    ^{K}\sum_{k'=1}^{K'} \pi_{k, k'} c(\theta_k, \theta'_{k'}):~\sum_{k=1}^{K} \pi_{k,k'} = w'_{k'},  \sum_{k'=1}^{K'} \pi_{k,k'}= w_k, \pi_{k,k'} \geq 0
    \right \}.
\end{equation}
where the minimum is taking over all transportation plans $\pi$ that satisfies two sets of marginal constraints.
This divergence is the lowest cost of transporting subpopulations
of $G$ to those of $G'$~\citep{chen2017optimal,delon2020wasserstein,bing2022estimation,nguyen2013convergence}.
The MR estimator is then defined to be
\begin{equation}
\label{eq:GMR}
\widehat{G}^{\mr} =\argmin\left\{T_{c}(\widebar{G}, G):~G\in\sG_{K}\right\}.
\end{equation}
Note that while $\widehat G^{\mr}$ is ostensibly the solution to a complicated bilevel optimisation program, the program actually simplifies and has a clear interpretation, as discussed in~\citet{zhang2022distributed}.
More precisely, in \eqref{eq:CTD_def}, the column-wise marginal constraints on the transportation plan $\pi$ are redundant and the full mass of each support point $\theta_{ik}$ in $\widebar G$ is transported to its ``closest'' support point in $\widehat{G}^{\mr}$.
The optimal transportation plan can thus be interpreted as an optimal clustering of the subpopulation parameters $\braces*{\theta_{ik} \colon (i,k) \in [m] \times [K]}$ into $K$ clusters $\{\gC_{j}\}_{j=1}^{K}$, with the support points of $\widehat{G}^{\mr}$ forming the cluster ``barycentres''.
Mathematically, the simplified optimisation program can be written as:
\begin{equation} \label{eq:MR_clustering}
    \min \braces*{\sum_{j=1}^K \sum_{(i,k) \in \gC_j} w_{ik}c(\hat\theta_{ik},\theta_j) \colon \theta_1,\ldots,\theta_K \in \Theta, [m] \times [K] = \gC_1\sqcup\cdots\sqcup\gC_K }
\end{equation}
where $\sqcup$ denotes disjoint union.
This formulation shows that MR implicitly aligns local components through clustering while simultaneously aggregating them via barycentric averaging.

\subsection{Calculating the mixture reduction estimator}
The clustering interpretation of MR naturally leads to an efficient majorisation–minimisation (MM) algorithm for computation.
Each iteration alternates between assigning local subpopulations to the nearest tentative cluster centres and updating these centres by barycentric averaging.
Formally, let $G^\dagger = \sum_{j=1}^K w^\dagger_j \delta_{\theta_j^{\dagger}}$ denote an initial mixing distribution.
The MM algorithm then proceeds by repeating the following two steps until convergence:
\begin{itemize}[leftmargin=*]
  \item
        \underline{\emph{Majorisation step.}} 
        Partition the support points of $\widebar{G}$,
        $\{\widehat \theta_{i k}: i \in [m], k\in [K]\}$, into $K$ groups based on their proximity to subpopulation parameter $\{\theta_j^{\dagger}, j=1,2,\ldots,K\}$, of the current iterate.
        In other words, we assign the $k$th subpopulation on the $i$th machine to the $j$th cluster $\gC_j$ if $c(\widehat\theta_{ik},\theta_j^{\dagger})\leq c(\widehat\theta_{ik},\theta_{j'}^{\dagger})$ for any $j'$.
        We denote this as $(i,k)\in \gC_j$.
  \item
        \underline{\emph{Minimisation step.}} Update $\theta_j^{\dagger}$ by computing the barycentre of the subpopulation parameters in its cluster.
        Specifically, set $\theta_j^{\dagger} = \argmin_{\theta\in\Theta} \sum_{\{(i,k) \in \gC_j\}} c(\widehat\theta_{ik}, \theta)$.
        We also update the mixing weight via $w^\dagger_j=\sum_{\{(i,k) \in \gC_j\}} \lambda_i \widehat w_{ik}$.
\end{itemize}
\citet{zhang2024gaussian} showed that the majorisation and minimisation steps can be interpreted as the assignment and update steps, respectively, in $K$-means clustering within the space $\gF$.
In particular, when $c(\theta,\theta')=\|\theta-\theta'\|^2$, the algorithm reduces to the Lloyd's algorithm in the Euclidean space $\Theta$.
Similar to $K$-means, one can show that the objective function, $G \mapsto T_c(\widebar{G},G)$, is monotone in the MM iterates.
Therefore, the MM algorithm is guaranteed to converge to a local minimum.
On the other hand, the optimisation problem in~\eqref{eq:GMR} is non-convex and, as is the case with similar schemes such as $K$-means and the EM algorithm, it is difficult to prove the convergence of the MM algorithm to a global minimum.
Nonetheless, they observed empirically that the output from the MM algorithm often performs as well as the global estimator, especially when the algorithm is run several times with different initial values.

The pseudocode and the closed-form formula for the update when $c(\cdot,\cdot)$ is chosen to be KL divergence are given in Appendix~\ref{app:mm_algorithm}.

\subsection{Theoretical guarantees for mixture reduction estimator}
\citet{zhang2022distributed} established a $\sqrt{N}$-rate of convergence of the MR estimator under finite Gaussian mixture models when the number of machines $m$ is a constant. 
In this section, we greatly improve and generalize their analysis by (i) allowing $m$ to grow with $N$, (ii) analysing generic finite mixture models, (iii) obtaining the asymptotic distribution of the MR estimator.

\subsubsection{Notation}
\label{sec:notation}
We use the following notation. 
Let $x$ be a vector in $\mathbb{R}^{d}$; $\|x\|_1$ and $\|x\|_2$ denote the $L^1$ and $L^2$ norms of $x$, respectively. 
For a $d \times d$ matrix $H$, we write $H \succeq 0$ to indicate that $H$ is positive semidefinite. 
Similarly, $\preceq 0$, $\succ 0$, and $\prec 0$ denote negative semidefinite, positive definite, and negative definite matrices, respectively. 
Let $\lambda_{\min}(H)$ and $\lambda_{\max}(H)$ denote the minimum and maximum eigenvalues of $H$. 
The operator norm of $H$ is defined as $\matrixnorm{H} = \sup_{u \in \mathbb{R}^{d}: \|u\|_2 \leq 1} \|Hu\|_2 = \lambda_{\max}(H)$. 

For two sequences of positive numbers $a_n$ and $b_n$, we write $a_n = O(b_n)$ (or {\color{Melon}{$a_n = \Omega(b_n)$}}) if there exists a constant $C > 0$, independent of $n$, such that $a_n \leq C b_n$ (or {\color{Melon}{$a_n \geq C b_n$}}) for all $n$. 
We write $a_n = \Theta(b_n)$ if $a_n = O(b_n)$ and $a_n = \Omega(b_n)$. 
$\widetilde{O}$ and $\widetilde{\Omega}$ have the same meaning as $O$ and $\Omega$ respectively, except they omit logarithmic factors.
For a sequence of random variables $X_n$ and a sequence of numbers $a_n$, we write $X_n = O_{P}(a_n)$ (or {\color{Melon}{$X_n = \Omega_{P}(a_n)$}}) if, for any $\epsilon > 0$, there exist $n_{\epsilon}$ and $C_{\epsilon}$ such that $\mathbb{P}(|X_n/a_n| \leq C_{\epsilon} ) \geq 1 - \epsilon$ (or {\color{Melon}{$\mathbb{P}(|X_n/a_n| \geq C_{\epsilon} ) \geq 1 - \epsilon$}}) for all $n \geq n_{\epsilon}$. 
We write $X_n = \Theta_{P}(a_n)$ if $X_n = O_{P}(a_n)$ and $X_n = \Omega_{P}(a_n)$.
We write $X_n = o_{P}(a_n)$ if for any $\epsilon > 0$, $\lim_{n\to\infty} \sP(|X_n-a_n|>\epsilon) = 0$.

For a parametric density function $f(x;\theta)$, the gradient with respect to the parameters is denoted by $\nabla f(x;\theta') = \partial f(x;\theta)/\partial \theta|_{\theta=\theta'}$.
We use $\gN(\mu, \Sigma)$ to denote a Gaussian (normal) distribution with mean $\mu$ and covariance $\Sigma$.
Let $G = \sum_{k=1}^{K} w_k \delta_{\theta_k}$ represent a mixing distribution of order $K$, where $w_k$ are the mixing weights and $\theta_k$ are the subpopulation parameters. 
The true mixing distribution is denoted by $G^* = \sum_{k=1}^{K} w_k^* \delta_{\theta_k^*}$. 
Under a \emph{fixed arbitrary ordering}, the vectorized representation of $G^*$ is given by $\mG^* = (w_1^*, \ldots, w_K^*, \theta_1^*, \ldots, \theta_K^*)$. 
Let $\sigma: [K] \to [K]$ denote a permutation of the subpopulation orders. 
For any other mixing distribution $G$, define
\begin{equation} 
\label{eq:aligned_component_def}
\mG = \argmin \left\{ \| \vv_{\sigma} - \mG^* \|_2 : \vv_{\sigma} = (w_{\sigma(1)}, \ldots, w_{\sigma(K)}, \theta_{\sigma(1)}, \ldots, \theta_{\sigma(K)}) \right\},
\end{equation}
where the minimisation is taken over all possible permutations $\sigma \in S_K$.

The log-likelihood contribution of a single observation $X$ and its expectation under the true distribution are denoted by $\ell(\mG; X) = \log f_{\mG}(X)$ and $L(\mG) = \mathbb{E}_{\mG^*}\{\ell(\mG; X)\}$, respectively. 
We omit the subscript in the expectation when it is taken with respect to the true distribution.

\subsubsection{Assumptions}
\begin{assumption}[Parameters]
\label{assumption:parameter-space}
The parameter space $\Theta \subset \R^d$ of the subpopulation distribution family $\gF$ is compact and convex. 
The true mixing distribution is $G^* = \sum_{k=1}^{K} w^*_k \delta_{\theta^*_k}$ with distinct $\theta_k^*$ and mixing proportions $w^*_k \in (0, 1)$.
\end{assumption}

\begin{remark}
The conditions on the parameter space is standard. 
See for example~\citet{nguyen2013convergence},~\citet{ho2016strong}, and~\citet{cai2019chime}.
Under Gaussian mixtures,~\citet{cai2019chime} uses $\Theta = \{\mu, \Sigma: \|\mu\|_2\leq M_1, \|\Sigma\|_2\leq M_2\}$ for some constant $M_1$, $M_2>0$.
\end{remark}

\begin{assumption}[Identifiability]
\label{assumption:identifiability}
The mixture model on $\gF$ is identifiable; that is, for any $G, G' \in \sG_K$, $f_{G}(x) = f_{G'}(x)$ for all $x$ implies $G = G'$.
\end{assumption}
\begin{remark}
Note that this assumption implies that $G^*$ is well-defined and is the unique maximizer of the population log-likelihood (see Lemma~\ref{lemma:population-likelihood-global-maxima}).
Commonly used mixture models, such as Poisson, Gaussian, and Gamma mixtures, are identifiable. 
In contrast, mixtures of Beta or Binomial distributions are generally non-identifiable. 
For further details, see~\citet[Section 1.4]{chen2023statistical}.
\end{remark}

\begin{assumption}[Local strong concavity]
\label{assumption:strong-convexity}
The population log-likelihood $L(\mG)$ is twice differentiable, and the Hessian matrix $\nabla^2 L(\mG^*)$ is finite and negative definite, \emph{i.e.}, there exists a constant $\lambda > 0$ such that $\nabla^2 L(\mG^*) \preceq -\lambda I$.
\end{assumption}

\begin{assumption}[Smoothness]
\label{assumption:smoothness}
There exist finite constants $\delta$, $H$, and a function $W(x)$ with $\sE[|W(X)|^8]$ finite such that for all $\mG, \mG'$ within an $\delta$-neighbourhood of $\mG^*$ (in the Euclidean distance), we have $\sE\{\|\nabla \ell(\mG; X)\|^8\} \leq H$ and $\sE\{\matrixnorm{\nabla^2 \ell(\mG; X) - \nabla^2 L(\mG)}^8\} \leq H$, and the Hessian matrix is $W(x)$-Lipschitz continuous with respect to $\mG$, \ie $\matrixnorm{\nabla^2 \ell(\mG'; x) - \nabla^2 \ell(\mG; x)} \leq W(x) \|\mG' - \mG\|_2$.
\end{assumption}

\begin{assumption}[Cost function]
\label{assump:cost_bregman_divergence}
The cost function in the composite transportation divergence~\eqref{eq:CTD_def} is a reversed Bregman divergence\footnote{$D_{A}(\theta,\theta')$represents the Bregman divergence, and the arguments are reversed in our definition of the cost function.} induced by a strictly convex function $A(\cdot)$, that is~$c(\theta',\theta) = D_{A}(\theta,\theta') = A(\theta) - A(\theta') - (\theta - \theta')^\top \nabla A(\theta')$.
Furthermore, for each component parameter $\theta_k^*$, $k\in[K]$, there is a neighbourhood $B_{\delta}(\theta_k^*)$  s.t. for $\theta, \theta' \in B_{\delta}(\theta_k^*)$, we have
\begin{enumerate}[label=(\alph*)]
    \item (Strong convexity and smoothness) $\eta_{-}\mI\leq \nabla^2 A(\theta) \leq \eta_{+}\mI$ for some $\eta_+ \geq \eta_- > 0$;
    \item (Local Lipschitz) $\matrixnorm{\nabla^2 A(\theta) - \nabla^2 A(\theta')}\leq \widetilde{A}\|\theta-\theta'\|$ for some $\widetilde{A} >0$\footnote{We may assume the same $\delta$, $\eta_-$, and $\eta_+$ apply for each $k$ since the order of the mixture $K$ is finite.}. 
\end{enumerate}
\end{assumption}

\begin{remark}
Examples of Bregman divergence include Euclidean distance, Mahalanobis distance, and the Kullback-Leibler (KL) divergence for exponential families. 
All of them satisfy the required assumptions. 
Further details and derivations are provided in Appendix~\ref{sec:bregman_divergence_appendix}.
\end{remark}

\subsubsection{Rate of convergence and asymptotic normality}
We first establish useful properties of the local estimator.
\begin{lemma}[Properties of local MLE]
\label{lemma:local_mle_property}
Let $\{X_1,\ldots,X_n\}$ be $n$ IID observations from a finite mixing distribution $G^*$ of known order $K$.
Under Assumptions~\ref{assumption:parameter-space}--\ref{assumption:smoothness}, the MLE $\widehat{G} = \sum_{k=1}^{K} \widehat{w}_k \delta_{\widehat{\theta}_k}$ based on this sample is unique and satisfies the following:
\begin{enumerate}[label=(\alph*)]
\item (Berry-Esseen CLT) $\sup_{A \in \mathscr{A}} \left| \sP\left( \sqrt{n} \, \mI^{1/2}(\mG^*)(\widehat{\mG} - \mG^*) \in A \right) - \sP(Z \in A) \right| =O(n^{-1/2})$
where $Z$ is a standard Gaussian random vector of dimension $pK + (K-1)$, $\mI(\mG^*)$ is the Fisher information matrix at $\mG^*$, and $\mathscr{A}$ is the collection of all convex sets in $\R^{pK + (K-1)}$;
\item (Bias bound) $\|\sE \{\widehat{\mG} - \mG^*\}\| = O(n^{-1})$;
\item (Moment bounds) $\sE\{\|\widehat{\mG} - \mG^*\|^{q}\} = O(n^{-q/2})$ for $1\leq q \leq 8$.
\end{enumerate}
\end{lemma}
The proof of the Lemma is deferred to Appendix~\ref{sec:local_mle_properties}.
Note that $\widehat\mG$ is an oracle quantity since it uses the true parameter vector $\mG^*$ to align its components (see \eqref{eq:aligned_component_def}).
Indeed, if $m$ such aligned local estimates $\widehat\mG_1,\widehat\mG_2,\ldots,\widehat\mG_m$ were known, aggregation would be trivial via simple averaging: $\widebar{\mG} = m^{-1}\sum_{i=1}^m \widehat \mG_i$.
Combining a bias-variance decomposition of the mean squared error (MSE) of the simple averaging estimator together with Lemma~\ref{lemma:local_mle_property}(b) and (c) would then immediately yield a $\sqrt{N}$-rate of convergence. 
Similar calculations were performed by~\citet{zhang2015divide} and~\citet{huang2019distributed} in deriving convergence rates for SC learning of $M$-estimators for Euclidean parameters via simple averaging. 
Since we cannot practically compute these oracle quantities, we need to consider the more complicated MR estimator, whose statistical properties are detailed in the following theorem.

\begin{theorem}
\label{thm:oracle_rate}
Under assumptions~\ref{assumption:parameter-space}--\ref{assump:cost_bregman_divergence}, then the MR estimator in~\eqref{eq:GMR} satisfies
\begin{enumerate}[label=(\alph*)]
    \item (Convergence rate) $\|\widehat \mG^{\mr} - \mG^*\| = O_{P}(N^{-1/2})$ when $n \geq m$;
    \item (Asymptotic normality) $\sqrt{N}(\widehat \mG^{\mr} - \mG^*) \to \gN(0, \mI^{-1}(\mG^*))$ when $m = o(n)$, where $\mI(\mG^*)$ is the Fisher information matrix at the true parameter value.
\end{enumerate}
\end{theorem}

The proof details are deferred to Appendix~\ref{app:mr_rate_proof}, but we sketch the main ideas here.
Since each local estimator is consistent, we can formalize our discussion in Section \ref{subsec:MR_estimator} and prove that with high probability, the MR estimator indeed aligns the components of different local estimates and aggregates the aligned components via their barycentre (Lemma~\ref{lem:barycentre_equivalence}).
In other words, we have
\begin{equation}
\widehat w_k^{\mr} 
= 
m^{-1}\sum_{i=1}^m \widehat{w}_{ik},\quad\quad
\widehat \theta_k^{\mr} 
= 
\argmin_{\theta} \sum_{i=1}^m \widehat{w}_{ik} c(\widehat \theta_{ik},\theta),
\end{equation}
for $k\in[K]$.
As simple averages, the mixing weight estimators can be easily analysed.
Meanwhile, each subpopulation parameter can be interpreted as an $M$-estimator and unsurprisingly has an asymptotically linear expansion.
Further decomposing this expansion into bias and variance, we see that the bias terms are of $O(n^{-1})$ order due to Lemma \ref{lemma:local_mle_property} (b) and (c), while the variance term is asymptotically equivalent to that of the simple average estimator $m^{-1}\sum_{i=1}^m (\widehat \theta_{ik}-\sE\widehat\theta_{ik})$. 
Under $n \geq m$, this thus yields the optimal $O_{P}(N^{-1/2})$ rate of convergence.
Furthermore, for $m=o(n)$, we have $O(n^{-1}) = o(N^{-1/2})$ and the variance term dominates, yielding asymptotic normality.

Both conclusions in Theorem~\ref{thm:oracle_rate} mirror similar results for SC learning of $M$-estimators for Euclidean parameters via averaging~\citep{zhang2015divide,huang2019distributed}, showing that despite prior challenges, finite mixture models are not any harder to learn, at least from an asymptotic statistical perspective.
Note also that Theorem~\ref{thm:oracle_rate} (b) shows that the MR estimator is asymptotically equivalent to the global MLE when $m=o(n)$.
Therefore, the SC approach incurs no efficiency loss in this case.
Outside of this scaling regime, the SC approach incurs some degree of additional bias.

%% file: sections/method.tex
\section{Distance-filtered reduction under Byzantine failure}
\subsection{Formal definition of Byzantine failure}
\label{sec:byzantine_definition}

Byzantine failure is a well-recognized challenge in distributed learning~\citep{lamport1982byzantine}, which occurs when a subset of local machines within the system transmit erroneous or even malicious data to the central server. 
Such errors can arise from various sources, including hardware or software glitches, data corruption, or communication interruptions.
We formally define Byzantine failure model below.

\begin{definition}[Byzantine failure]
Let $\sB \subset [m]$ be a fixed, unknown subset of the local machines, called the Byzantine failure set.
Let $\xi_1,\xi_2,\ldots,\xi_m$ be fixed, arbitrary element of $\mathbb{G}_K$.
The central sever receives transmitted estimates
\begin{equation*}
    \widetilde{G}_i =
    \begin{cases}
        \widehat{G}_i, & \text{ when } i \notin \sB \\
        \xi_i,             & \text{ when } i \in \sB.
    \end{cases}
\end{equation*}
\end{definition}
We focus on scenarios where \emph{the majority of local estimates} are transmitted error-free; that is, we assume the size of the Byzantine failure set satisfies $|\sB| < 0.5 m$.
\footnote{The majority rule, which assumes that the majority of local estimates are failure-free, is a fundamental concept in distributed systems. 
Indeed, it is easy to see that it is impossible to distinguish between failure and failure-free machines when this rule does not hold.
This majority rule has been implemented in several well-known distributed systems. 
For instance, in decentralized networks like Bitcoin and Ethereum, the majority rule plays a critical role in maintaining the integrity of the blockchain~\citep{nakamoto2008bitcoin}. 
Similarly, in distributed databases such as Google's Chubby and Apache Cassandra~\citep{lakshman2010cassandra}, protocols are used to ensure that a majority of nodes agree on updates to the system's state.}
Our objective is to produce a final estimate $\widehat{G}$, based on local estimates transmitted $\{\widetilde{G}_i: i \in [m]\}$, that remains both statistical reliable and computationally efficient. 
In the rest of this paper, we will refer to the elements of $\braces{\widetilde G_i \colon i \in \sB}$ (respectively $\braces{\widetilde G_i \colon i \notin \sB}$) as the set of failed local estimates (respectively failure-free local estimates).

It is insightful to compare and contrast the Byzantine failure model with other models for data corruption studied in robust statistics.
To perform this comparison, we consider an abstract setting in which we observe a sample $Y_1,Y_2,\ldots,Y_m$.
The Byzantine failure model replaces an $\alpha$ fraction of these observations with arbitrary values.
Meanwhile, perhaps the most widely used data corruption model is the Huber contamination model, under which we have $Y_i \overset{\text{i.i.d.}}{\sim} (1-\alpha)F + \alpha B$ for $i=1,2,\ldots,m$, where $F$ is the distribution of interest, and $B$ is an arbitrary distribution \citep{huber1992robust,huber2008progressive}.
This is weaker than (i.e. is a special case of) the Byzantine failure model, since under the latter, the corrupted observations can be arbitrary instead of being drawn from a fixed distribution $B$.
On the other hand, the Byzantine failure model is weaker than the adversarial attack model, in which a malicious adversary is allowed to alter an $\alpha$ fraction of the sample $Y_1,Y_2,\ldots,Y_m$ arbitrarily \emph{after} having seen their values---in comparison, the Byzantine failure corruptions are constrained to be independent of the original values of $Y_1,Y_2,\ldots,Y_m$.
The adversarial attack model has recently become a topic of keen interest in the theoretical computer science community \citep{diakonikolas2016robust,diakonikolas2017being,diakonikolas2019recent,hopkins2019hard,diakonikolas2020outlier,cheng2019high,cheng2020high,lugosi2021robust}.

These differences of formulation stem from differences in motivation.
The Huber contamination model was classically motivated by the existence of outliers, arising from natural variations, data collection errors, or model misspecification.
On the other hand, Byzantine failure can be the result of malicious or intentional attempts to disrupt the system and hence may not have a well-defined distribution.
Yet, because of the distributed nature of the data, it is difficult for an attacker to intercept and observe all transmitted local estimates, which is the basis for a truly adversarial attack.
Moreover, we note that the objects of disruption differ between traditional settings for robust statistics (which concern observed data) and the Byzantine failure model (which pertain to parameter estimates rather than the raw observations themselves).

Due to the differences mentioned above, specialized tools are needed, particularly in non-Euclidean parameter spaces such as those of mixture models.
Lastly, we remark that it is easy to see that vanilla mixture reduction is not Byzantine-tolerant.
Indeed, if any subpopulation estimate suffers Byzantine failure, then according to the clustering interpretation of MR~\eqref{eq:MR_clustering}, the corrupted subpopulation estimates must be assigned to one of the clusters.
Consequently, the barycentre of this cluster can be arbitrarily distorted, making the aggregated estimate $\widehat G^{\mr}$ unreliable.

\subsection{$L^2$ distances between mixture densities}
\label{subsec:l2_distance_densities}
Towards designing a Byzantine tolerant aggregation method, we first study properties of the $L^2$ distances between densities of mixture distributions.
Given two mixing distributions $G, G' \in \sG_K$, this is defined as
\be
\label{eq:L2_distance}
  D(G,G') =  \left\{\int \{f_{G}(x)-f_{G'}(x)\}^2\,dx\right\}^{1/2}.
\ee
For discrete distributions, the integration becomes a summation. 
To make our presentation more accessible, we express all results in terms of Lebesgue integration.
We choose the $L^2$ distance for the following two key reasons. 
First, it is one of the few distances between finite mixtures that has an analytical form, unlike other distances such as the Kullback-Leibler (KL) divergence, which typically require numerical evaluation.
Second, it is theoretically tractable.
More precisely, under standard regularity assumptions (Assumption~\ref{assumption:parametric_family}), the squared $L^2$ distance~\eqref{eq:L2_distance} is well approximated by a quadratic form in the difference between the aligned parameter vectors, $\mG - \mG'$, whenever $G$ and $G'$ are sufficiently close to the true mixing distribution $G^*$.
This follows from a Taylor expansion shown in Lemma~\ref{lemma:L2_distance_densities} to come.

\begin{assumption}[Parametric family]
\label{assumption:parametric_family}
The parametric family $\gF$ satisfies:
\begin{enumerate}[label=(\roman*)]
\item For all $\theta\in \Theta$, the density function is square integrable, \ie $\int f^2(x;\theta)\, dx <\infty$.
\item The gradient  $\nabla f$ with respect to $\theta$ satisfies $\int \left\| \nabla f(x; \theta_k^*) \right\|^2 \, dx < \infty$ for $k \in [K]$.
\item For each $k \in [K]$, there is a neighbourhood and a square integrable function $V: \sR^{d} \to \sR_+$ such that for all $\theta\in B_\delta(\theta_k^*)$, we have $
\left\| \nabla f(x; \theta) - \nabla f(x; \theta_k^*)  \right\| \leq V(x) \|\theta-\theta_k^*\|$ for all $x \in \sR^d$.
\end{enumerate}
\end{assumption}

\begin{remark}
These assumptions can be easily verified for most well-known parametric families, such as exponential and Gaussian distributions. 
For Gamma distributions, the conditions hold when the shape parameter satisfies $r>0.5$.
\end{remark}

\begin{lemma}[$L^2$ distance between mixture densities]
\label{lemma:L2_distance_densities}
Let $G$ and $G'$ be two mixing distributions. 
Under assumptions~\ref{assumption:parameter-space}--\ref{assump:cost_bregman_divergence} and assumption~\ref{assumption:parametric_family}, as $G, G' \to G^*$, the $L^2$ distance between their densities satisfies
\[
D^2(G, G') =({\mG} - \mG')^{\top} \mH^*( {\mG} - \mG')  + O(\max\{\|\mG - \mG^*\|^3,\|\mG'-\mG^*\|^3\}).    
\]
where $\mH^*$ is a positive semi-definite matrix defined in Section \ref{sec:distance_between_mixtures}.
\end{lemma}
See the proof of the Lemma in Appendix~\ref{app:L2_closed_form}.
While we cannot control distances concerning failed local estimates, we may combine Lemma \ref{lemma:L2_distance_densities} together with Lemma \ref{lemma:local_mle_property}(a) to see that, when rescaled by $\sqrt{n}$, the $L^2$ distance between the density of a failure-free local estimate and the true mixture distribution converges in distribution to a generalized chi-squared distribution.
Further using the well-known uniform convergence of the empirical CDF shows that the empirical quantiles of these distances also converge uniformly.
Finally, the moment bound from Lemma \ref{lemma:local_mle_property}(c) shows that the probability of large deviations for these distances is relatively small.
In summary, we have the following concentration lemma.

\begin{lemma}[$L^2$ distance concentration]
\label{lemma:distance_concentration}
Under assumptions~\ref{assumption:parameter-space}--\ref{assump:cost_bregman_divergence} and assumption~\ref{assumption:parametric_family}, the following concentration results hold for the uncorrupted local estimates $\widehat{G}_1, \widehat{G}_2,\ldots,\widehat{G}_m$:
\begin{enumerate}[label=(\alph*)]
\item (Bulk) Let $j_1,\ldots, j_{m}$ be a permutation of $[m]$ such that $D(\widehat{G}_{j_1},G^*)\leq \cdots \leq D(\widehat{G}_{j_{m}},G^*)$.
For any $0 < \epsilon < 1/2$, as $n, m\to\infty$, we have
\begin{equation}
\label{eq:distance_bulk_concentration}
D(\widehat{G}_{j_{\lfloor \epsilon m \rfloor}},G^*) = \Omega_P(n^{-1/2}), \quad\quad D(\widehat{G}_{j_{\lfloor (1-\epsilon)m \rfloor}},G^*) = O_P(n^{-1/2}).
\end{equation}
\item (Large deviations) For $i \in [m]$, as $n, \rho_{n} \to \infty$, we have
\begin{equation}
\label{eq:distance_deviation_concentration}
\sP(D(\widehat{G}_i,G^*) \geq \rho_{n} n^{-1/2}) = O(\rho_{n}^{-8}).
\end{equation}
\end{enumerate}
\end{lemma}
See the proof of the lemma in Appendix~\ref{app:rCOAT_COAT}.
The significance of this lemma is threefold.
First, together with the majority rule, the second equation in \eqref{eq:distance_bulk_concentration} implies that a majority of the local estimates are within $O(n^{-1/2})$ $L^2$-distance from $G^*$, which means that a good initial estimate is easy to obtain.
Second, the first equation in \eqref{eq:distance_bulk_concentration} implies that the dispersion of the failure-free local estimates can be estimated from data.
Third, \eqref{eq:distance_deviation_concentration} implies that a ball of radius $O(\rho_n n^{-1/2})$ around a good initial estimate, given a small inflation factor $\rho_n$, will contain almost all of the failure-free local estimates.

\subsection{Distance-filtered mixture reduction}
\label{sec:proposed_estimator_machine}
\begin{wrapfigure}{r}{0.4\textwidth}
\begin{center}
\includegraphics[width=0.35\textwidth]{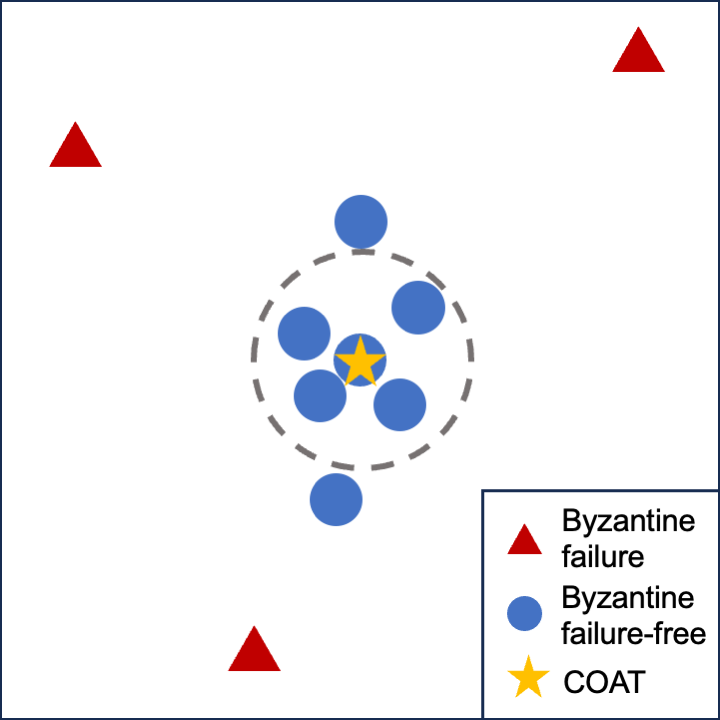}
\end{center}
\caption{
    $3$ failure machines out of $m = 10$ machines.
    Failure-free local estimates (blue dots) are clustered, while failed ones (red triangles) are scattered.
    The COAT (yellow star) and other $4$ local estimates that are closest to COAT are in the grey circle, and they are failure-free.}
\label{fig:red_coat_demo}
\end{wrapfigure}
We are now ready to describe our proposed method concretely.
Denote the open ball centred at $G$ with radius $r$ by
$B(G; r) = \{ G': D(G', G) < r, G' \in \sG_K\}$.
Let $|\sA|$ denote the cardinality of a set $\sA$.
Given the set of transmitted local estimates $\widetilde{\gG} = \{\widetilde{G}_i: i\in [m]\}$, define
\[
  r({G}) = \inf\left\{r: |\{G'\in \widetilde{\gG}:
  G' \in B(G; r) \}|\geq m/2\right\}
\]
as the smallest radius such that the closure of the ball $B(G; r(G))$ centred at $G$ contains at least $50\%$ of local estimates.

Define the Centre Of ATtention (COAT) as in~\citet{nissim2007smooth}:
\begin{equation}
\label{eq:coat}
\widehat{G}^{\coat} = \arg\min\left\{r(G): G \in \widetilde{\gG}\right\}.
\end{equation}
An example illustrating the idea of COAT is given in Figure~\ref{fig:red_coat_demo}.
It is worth noting that~\citet{nissim2007smooth} originally introduced COAT to obtain differentially private estimates for Euclidean parameters, which is conceptually distinct from our purpose.
Within our framework, COAT can be interpreted as a robust aggregation estimator; however, since it effectively relies on a single local estimator, its statistical efficiency is very low--its convergence rate is only $O_P(n^{-1/2})$, as shown in Proposition~\ref{prop:COAT_bound}.
An efficient aggregation procedure should instead leverage all non-corrupted local estimators.
To achieve this, we propose a two-step approach: first, a distance-based filtering step identifies non-corrupted local estimators; then, the MR method is applied to aggregate these filtered estimates, yielding a statistically efficient and robust global estimator.
Specifically, let $r^{\coat} = r(\widehat{G}^{\coat})$ and $\rho \geq 1$ be a tuning parameter, we define the set of selected machines as
\be
\label{eq:cred-set-estimate}
\selected_\rho = \{i:~D(\widehat{G}^{\coat}, \widetilde{G}_i) \leq \rho r^{\coat}\}.
\ee
Using this set, we define a family of Byzantine failure-tolerant \emph{Distance Filtered Mixture Reduction} (DFMR) estimators:
\be
\label{eq:DFMR}
\widehat{G}_{\rho}^{\dfmr} = \mathfrak{A}(\{\widetilde{G}_i: i\in \selected_{\rho}\}).
\ee
where $\mathfrak{A}(\cdot)$ is the mixture reduction operator.

\begin{algorithm}[!ht]
\centering
\begin{algorithmic}
\State{\textbf{Input}: $\rho \geq 1$, local estimator $\widetilde{G}_i$.}
\State{Initialize an empty matrix $D \in \sR^{K\times K}$ on the server.}
\For {$i\in[K]$}
\For {$j \in \{i+1,\ldots, K\}$}
\State {Compute $D_{ij} = D(\widetilde{G}_i,\widetilde{G}_j)$}
\State{Let $D_{ji} \leftarrow D_{ij}$.}
\EndFor
\EndFor
\For {$i\in[K]$}
\State{Let $\text{med}_i = \text{median}(D_{i1},\ldots, D_{iK})$.}
\EndFor
\State{Let $i^* = \argmin_{i} \text{med}_i$}
\State{Let initial estimate be COAT $\widehat{G}^{\coat} \leftarrow \widetilde{G}_i^* $.}
\State{Compute $r^{\coat} = \text{med}_{i^*}$.}
\State{Send $\{\widetilde{G}_j: D(\widetilde{G}_j, \widehat{G}^{\coat}) < \rho r^{\coat}\}$ to Algorithm~\ref{alg:mm_reduction} for aggregation.} 
\end{algorithmic}
\caption{Numerical algorithm for DFMR($\rho$).}
\label{alg:our_method}
\end{algorithm}
When $\rho = 1$, only $50\%$ of the local estimates are used for aggregation.
If the failure rate $\alpha$ is below $50\%$, this may not fully leverage all failure-free local estimates, limiting the estimator's efficiency. 
To address this, we recommend choosing a value of $\rho$ larger than $1$.
On the other hand, setting $\rho$ to be too large would include too many failed local estimates in the selected set, leading to performance degradation.
Consequently, the MSE of DFMR exhibits a U-shaped behaviour as $\rho$ increases, as demonstrated empirically in Section~\ref{sec:robustness_rho_exp}. 
In our theoretical analysis in Section~\ref{sec:theory}, we obtain convergence rate results with the choice $\rho = m^{\delta}$ for any $\delta > 1/14$.
In practice, under common attack models, we observe that the MSE is relatively insensitive to the choice of $\rho$ over a fairly wide range of values and that setting $\rho = 1.3$\footnote{This is approximately equal to $50^{1/14}$. 
The value of $m^{1/14}$ changes slowly with $m$.} seems to work fairly well.
The pseudo-code for the numerical implementation of our proposed DFMR method is provided in Algorithm~\ref{alg:our_method}.

%% file: sections/theory.tex
\subsection{Theoretical guarantees}
\label{sec:theory}

In this section, we establish the rate of convergence of our proposed DFMR estimator.
Following our discussion in Section~\ref{subsec:l2_distance_densities}, it is easy to see that for a suitably chosen inflation factor $\rho$, the selected set $\selected_\rho$ contains the vast majority of the failure-free estimates.
What we therefore need to show is that the most harmful failed estimates are indeed filtered out.
We, however, run into the following issue:
Local estimates are filtered based on the $L^2$ distances between mixture densities, but the error of our estimate is measured in terms of the Euclidean parameter distance.
Lemma~\ref{lemma:L2_distance_densities} shows that the former is locally Lipschitz continuous with respect to the latter, but does not guarantee the converse to hold.
As such, it does not preclude a (bad) local estimate from being selected on account of having a density with small $L^2$ distance from that of the COAT, but having a large Euclidean parameter distance from $G^*$, thereby hurting the efficiency of the estimate.
This issue does not occur when the parametric family $\gF$ is first-order strongly identifiable~\citep{ho2015identifiability,ho2016strong}.
As we show in Lemma~\ref{lem:strong_identifiability}, this ensures that the matrix $\mH^*$ in Lemma~\ref{lemma:L2_distance_densities} is invertible, which implies that the Euclidean parameter distance is locally Lipschitz with respect to the $L^2$ distance between densities.

\begin{definition}[First-order strong identifiability]
The family $\gF$ is first-order strongly identifiable if $f(x;\theta)$ is differentiable in $\theta$ and for any finite $k$ and $k$ distinct parameter values $\theta_1,\ldots, \theta_k$, the equality
$\sum_{i=1}^{k}\alpha_i f(x;\theta_i) + \beta_i^{\top}\nabla f(x;\theta_i) =0$
almost for all $x$ implies that $\alpha_i=0$, $\beta_i =\gamma_i =0$ for $i=1,\ldots, k$.
\end{definition}
\begin{remark}
The first-order strong identifiability condition is weaker than the second-order strong identifiability condition proposed by~\citet{chen1995optimal}.
A wide range of univariate parametric families, including most exponential families~\citep{chen1995optimal}, are known to be second-order identifiable.
With multidimensional parameter space, second-order strong identifiability holds for multivariate Gaussian distributions in location or scale, certain classes of student $t$-distributions, von Mises, Weibull, logistic and generalized Gumbel distributions~\citep{ho2016strong}, as well as multinomial distributions~\citep{manole2021estimating}. 
The Gaussian mixture in location and scale is not second-order identifiable but first-order identifiable.
Gamma distributions, Location-exponential distributions, and skew Gaussian distributions, however, are not first-order identifiable~\citep{ho2015identifiability}.
\end{remark}


We now establish the convergence rates of both the initial estimate, COAT, as well as the DFMR estimator, under the strong identifiability condition.\footnote{The general rate for $\gF$ that is not first-order strongly identifiable is given in Appendix~\ref{app:DFMR_rate}.}
\begin{proposition}
\label{prop:COAT_bound}
Under assumptions~\ref{assumption:parameter-space}--\ref{assump:cost_bregman_divergence} and~\ref{assumption:parametric_family}, the COAT satisfies $D(\widehat G^{\coat},G^*) = O_P(n^{-1/2})$.
When $\gF$ is also 1st-order strongly identifiable, we have $\|\widehat{\mG}^{\coat} - \mG^*\| = O_P(n^{-1/2})$.
\end{proposition}
The proof of the proposition is given in Appendix~\ref{app:DFMR_rate}.
It implies that the COAT estimator is Byzantine-failure tolerant, but its convergence rate is not optimal as anticipated.

\begin{theorem}
\label{thm:main_theorem}
Under assumptions~\ref{assumption:parameter-space}--\ref{assump:cost_bregman_divergence} and~\ref{assumption:parametric_family}, suppose further that $n \geq m$.
For any $r > 0$, let $\widetilde\alpha_m(r) = (r/m)\abs*{\braces{i \in \sB \colon D(\xi_i,G^*) \leq r }}$.
Let $\widehat G^{\omr}$ be the oracle mixture reduction estimator, \ie the solution to the optimization problem~\eqref{eq:GMR} in which the aggregation is performed over the failure-free machines $\sB^c$.
Suppose further that $\gF$ is first-order strongly identifiable.
Recall that $\mG$ is the vectorized form of the mixing distribution $G$.
The DFMR estimator with the choice $\rho = \Omega(m^{\delta})$ for any $\delta > 1/14$ satisfies
\begin{equation}
\label{eq:main_theorem}
    \widehat{\mG}^{\dfmr} - \mG^* = \gamma \paren*{\widehat{\mG}^{\omr} - \mG^*} + O_P(\widetilde\alpha_m(2\rho n^{-1/2})) + o_P(N^{-1/2})
\end{equation}
for some $0 < \gamma \leq 1$.
In particular, $\|\widehat{\mG}^{\dfmr} - \mG^*\| = O_{P}(N^{-1/2}+\widetilde\alpha_m(2\rho n^{-1/2}))$.
\end{theorem}
The proof of the theorem is deferred to Appendix~\ref{app:DFMR_rate}.
Since $\widetilde\alpha_m(r) \leq \alpha r$, the theorem implies that the DFMR estimator achieves a rate of $O_{P}(N^{-1/2} + \alpha \rho n^{-1/2})$.
Up to the factor of $\rho$, this is equivalent to the optimal rate suggested by \cite{yin2018byzantine} for SC algorithms applied in the Euclidean parameter setting.
Note that the additional factor of $\rho$ arises not because of any slack in our analysis, but because of the inherent difficulty in estimating mixture models.
Most other works establish the rate of convergence under the assumption that the local estimates are sub-exponential~\citep{chen2017distributed,yin2018byzantine,su2019securing,tu2021variance} or sub-Gaussian~\citep{yin2019defending}, whereas our assumptions on the mixture family are only sufficient to guarantee that our local estimates have bounded 8th moments.
This heavy tailed nature requires us to set $\rho$ large enough so as to contain almost all the estimates from failure-free machines.
Indeed, if we are able to guarantee finite $q$-th moments in Lemma \ref{lemma:local_mle_property} for some $q > 8$, we are able to relax the requirement on $\rho$ and achieve the results of Theorem~\ref{thm:main_theorem} with a choice of $\rho = m^{\delta}$, $\delta > 1/2(q-1)$.
If we are able to guarantee $\norm{\mG-\mG^*}$ is sub-Gaussian, then we may even choose $\rho = \log m$.

Theorem~\ref{thm:main_theorem} also shows potential slack in the previous analysis of Byzantine-tolerant methods.
Specifically, the effect of the failed local estimates on the optimal convergence rate is not in terms of the total failure fraction $\alpha$, but in terms of the fraction of failed local estimates that are relatively indistinguishable from the failure-free ones.
Indeed, under reasonable attack models, $\widetilde\alpha_m(2\rho n^{-1/2})$ can be relatively negligible, as illustrated theoretically in the following corollary and empirically in our simulations to come (Section~\ref{sec:exp}).

\begin{corollary}
\label{cor:failure_density}
Under the assumptions of Theorem \ref{thm:main_theorem}, assume further that $\rho = o(m^{1/6})$ and suppose that the failed local estimates $\braces*{\xi_i \colon i \in \sB}$ are drawn i.i.d. from a fixed distribution on $\sG_K$ satisfying $\sP(D(\xi_i,G^*) \leq r) = O(r^3)$ as $r \to 0$.
Then we have
$\sP\paren*{\widetilde\alpha_m(2\rho n^{-1/2}) > 0}= o(1)$.
Hence, $\widehat{\mG}^{\dfmr} = \widehat{\mG}^{\omr} + o_P(N^{-1/2})$.
\end{corollary}
The proof of the corollary is deferred to Appendix~\ref{app:DFMR_rate}.
The same optimal rate is achieved by the method developed in~\citet{wang2024distributed} for Byzantine-tolerant distributed learning of Lagrange multipliers in empirical likelihood-based inference. 
However, their result does not hold unless there is a constant-width separation margin between failure-free estimates and failed estimates, which is a more strict scenarios than ours.

Under the conditions of Corollary~\ref{cor:failure_density}, Theorem~\ref{thm:oracle_rate} (b) further implies the asymptotic normality of our proposed DFMR estimator. 
The asymptotic variance of the DFMR estimator matches that of the oracle estimator, ensuring full efficiency. 
In contrast,~\citet{tu2021variance}'s result, which established asymptotic normality result for VRMOM in Byzantine-robust distributed learning of Euclidean parameters, demonstrates loss of efficiency. 
More precisely, the asymptotic variance of VRMOM is larger than that of the oracle estimator by a factor of $\pi/3$. 
Furthermore, their asymptotic normality result requires a low failure rate ($\alpha = o(m^{-1/2})$) and additional assumptions on the size of $n$ and $m$.

%% file: sections/expt.tex

\section{Numerical experiments}
\label{sec:exp}
\subsection{Empirical comparison with other approaches}
\subsubsection{Experiment setting}
\label{sec:exp_setting}
We generate data from a widely used Gaussian mixture. 
Additional results for Gamma and Poisson mixtures are given in Appendix~\ref{app:gamma_mixture} and Appendix~\ref{app:poisson_mixture}, respectively.
To mitigate potential bias, the parameters of the Gaussian mixtures are generated at random using the \texttt{R} package \texttt{MixSim}~\citep{melnykov2012mixsim}.
An important quantity of a mixture is pairwise overlap $o_{j|i} = \sP(w_{i} f(X;\theta_i)< w_{j} f(X;\theta_j)|X \sim f(\cdot;\theta_i))$.
The maximum overlap of a finite mixture is defined as $
\texttt{MaxOmega} = \max_{i, j \in [K] } \{ o_{j|i} + o_{i|j}\}$.
Larger values of \texttt{MaxOmega} indicate greater difficulty in learning the model.  
We generate $R = 300$ sets of parameter values with dimension $d = 10$, order $K = 5$, and \texttt{MaxOmega} values of $10\%$, $20\%$, and $30\%$.  
The number of machines $m$ ranges from $20$, $50$, to $100$, with local sample sizes $n$ ranging from $5000$, $10000$, to $20000$.  
These choices for $d$, $n$, and $m$ are motivated by the nature of distributed learning methods, which are designed to handle large-scale learning scenarios, particularly when observations are high-dimensional and sample sizes are substantial.

We next generate Byzantine failure as follows.
We let the percentage of Byzantine failure machines, denoted as $\alpha$, vary from $0.1$ to $0.4$ with an increment $0.1$. 
For each given $\alpha$, we then randomly select Byzantine failure machines out of all machines.
To fully investigate various failure scenarios, we consider three representative types of failures:

\begin{itemize}[leftmargin=*]
    \item \emph{Mean failure}. 
    Following the convention for simulating Byzantine failure in Euclidean parameter space~\citep{tu2021variance}, each subpopulation mean estimate on a Byzantine failure machine is replaced by a vector of same length with entries generated independently from $N(0, 100^2)$.  

    \item \emph{Covariance failure}. 
    We perturb each covariance estimate on each failure machine by adding additive noise $d^{-2}\sum_{i=1}^{d}\xi_i\xi_i^{\top}$ where $\xi_i$ are IID standard Gaussian random vectors. 
    Such additive noise can ensure that the attacked covariance is still positive definite.

    \item 
    \emph{Weight failure}. 
    We replace the estimated mixing weights on the Byzantine failure by a random vector generated from the Dirichlet distribution. 
    The parameter values for the Dirichlet distribution differ across machines and are random integers within $[10, 20]$. 
    We choose parameter values within this range to ensure that the failure weights are significantly different from the failure-free weights.
\end{itemize}

These three failure types are representative of typical Byzantine failures in finite mixture models.  
Other failure cases, such as incorrect model order or estimates in the wrong parameter space, can often be easily identified with prior knowledge of the mixture model and are therefore not considered in this simulation.

\subsubsection{Baselines and performance measure metrics}
\label{sec:estimators}
There are no specific Byzantine-tolerant SC methods for learning finite mixture models in the literature, as pointed out in the introduction. 
However, some existing algorithms can be used as Byzantine-tolerant SC methods. 
Examples include the trimmed $k$-barycentre for robust clustering by~\citet{del2019robust}.
To showcase the effectiveness of the proposed method, we also include the following estimators in the simulation:
\begin{enumerate}
\item \textbf{Vanilla}. 
The reduction estimator in~\eqref{eq:GMR} that aggregates all local estimates. 
The cost function $c(\cdot,\cdot)$ is chosen to be the KL divergence. The reduction method is also used in our proposed method, and Oracle approaches but with filtered local estimates.

\item \textbf{Trim}. This is the trimmed $k$-means by~\citet{del2019robust} summarised in Algorithm~\ref{alg:trimmed_barycentre} in Appendix~\ref{app:trim_algorithm}. 
The trimming level $\eta$ in the approach is a hyperparameter and there is no consensus choice.
We let $\eta=0.5$ so that the estimator is robust yet efficient.
For fair comparisons, we also use the KL divergence as its cost function.

\item \textbf{COAT}. 
This is the estimator introduced in~\eqref{eq:coat}.

\item \textbf{Oracle}.
This estimator aggregates all failure-free local estimates.
\end{enumerate}

In each simulation setting, we generate random samples $\gX^{(r)}$ from a finite mixture with mixing distribution $G^{(r)} = \sum_{k} w_k^{(r)} \delta_{(\mu_k^{(r)},\Sigma_k^{(r)})}$ for $r = 1,\ldots, R$. 
After which $\gX^{(r)}$ is randomly partitioned into $m$ local sets. 
We use the pMLE~\citep{chen2009inference} locally with a penalty size of $n^{-1/2}$. 
The EM algorithm is used for numerical computation, and we declare convergence when the change of the per observation penalized log-likelihood function is less than $10^{-6}$. Given the large sizes of the simulated data, the maxima of the penalized likelihood are expected to be close to the true mixing distribution. 
Therefore, we use the true mixing distribution as the initial value for the EM algorithm. 
During aggregation, we use the COAT estimator as the initial value for the reduction. 
The estimated mixing distribution is denoted as $\widehat G^{(r)}= \sum_{k} \widehat w_k^{(r)} \delta_{(\widehat\mu_k^{(r)}, \widehat\Sigma_k^{(r)})}$. 
We use the following performance metrics:
\begin{itemize}
  \item 
   \textbf{Transportation distance} ($W_1 = W_1(\widehat G^{(r)},G^{(r)})$): The Wasserstein distance between the estimated and true mixing distributions. 

  \item 
  \textbf{Adjusted Rand Index} (ARI): Mixture models are widely used for model-based clustering. 
  The ARI serves as a common metric for assessing the similarity between two clustering results. 
  Specifically, we compare the clustering outcomes based on $\widehat G^{(r)}$ and $G^{(r)}$ respectively. 
  ARI values closer to $1$ indicate better performance.
\end{itemize}
These metrics provide a comprehensive evaluation of the proposed methods in comparison to other estimators across various settings.
Details about the metrics are given in Appendix~\ref{app:exp_performance_metric}.

\subsubsection{Robustness of the inflation factor $\rho$}
\label{sec:robustness_rho_exp}
The inflation factor $\rho$ is a hyper-parameter in our proposed method. While our theoretical analysis provides guidance on its order of magnitude, determining its precise value is more nuanced. Through this experiment, we demonstrate that our method exhibits remarkable robustness across a wide range of $\rho$ values, underscoring its practical utility.

In the top row of Figure~\ref{fig:EEI_threshold}, we set $m=100$ and $n=10^{4}$, and the failure rate to $\alpha=0.1$.
We vary $\rho$ from $1.0$ to $3.0$ in increments of  $0.05$ over $R=300$ repetitions and compute the mean and standard error of the distances between the DFMR($\rho$) estimator and the true parameter. 
For both mean and covariance failures, we observe that when $\rho\geq 1.3$, the DFMR($\rho$) estimator performs as well as the Oracle approach (dashed line) and significantly outperforms the COAT approach (dash-dotted line). 
In the case of weight failure, the optimal distance is achieved within a narrower range of $\rho$, and the expected U-shaped trend in distance is evident. 
Notably, even when $\rho$ is increased to $3.0$, the DFMR($\rho$) estimator remains robust and consistently outperforms the COAT approach.

In the bottom row of the figure, we further investigate the impact of varying the number of local machines. The results reveal that the optimal choice of $\rho$ remains consistent across different numbers of machines, highlighting the stability and scalability of our method.
\begin{figure}[!htbp]
\centering
\includegraphics[width=0.3\textwidth]{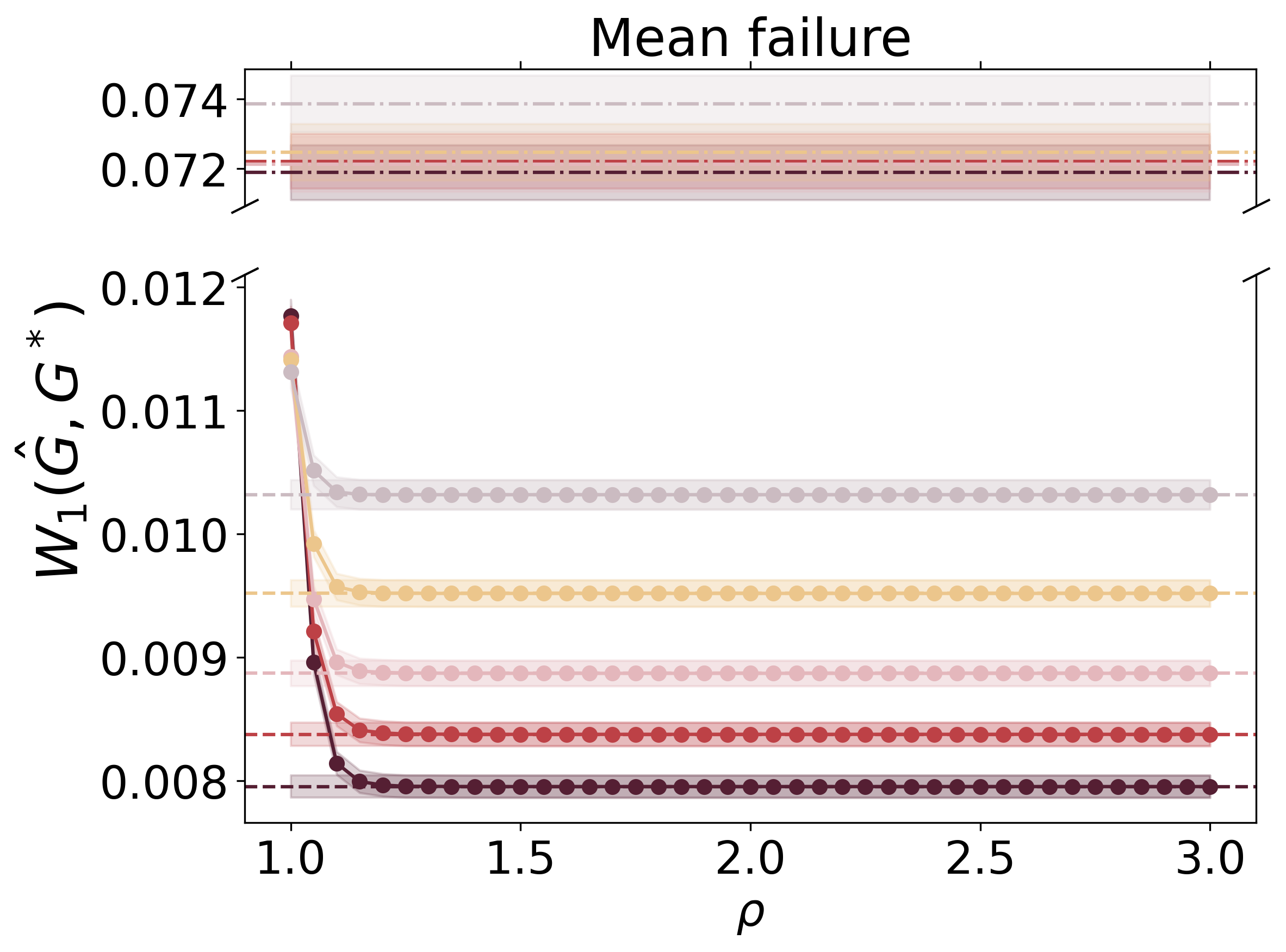}
\includegraphics[width=0.3\textwidth]{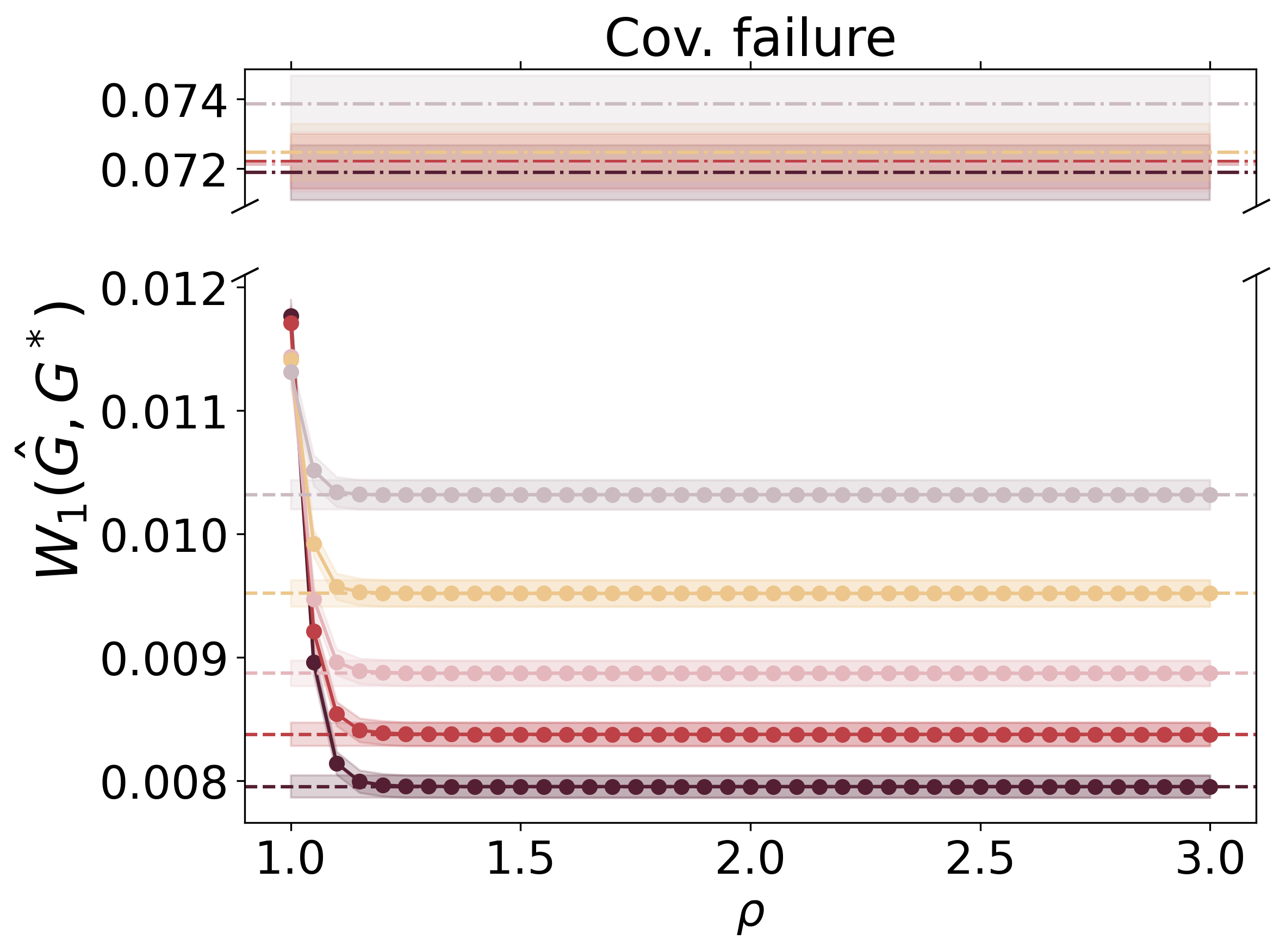}
\includegraphics[width=0.3\textwidth]{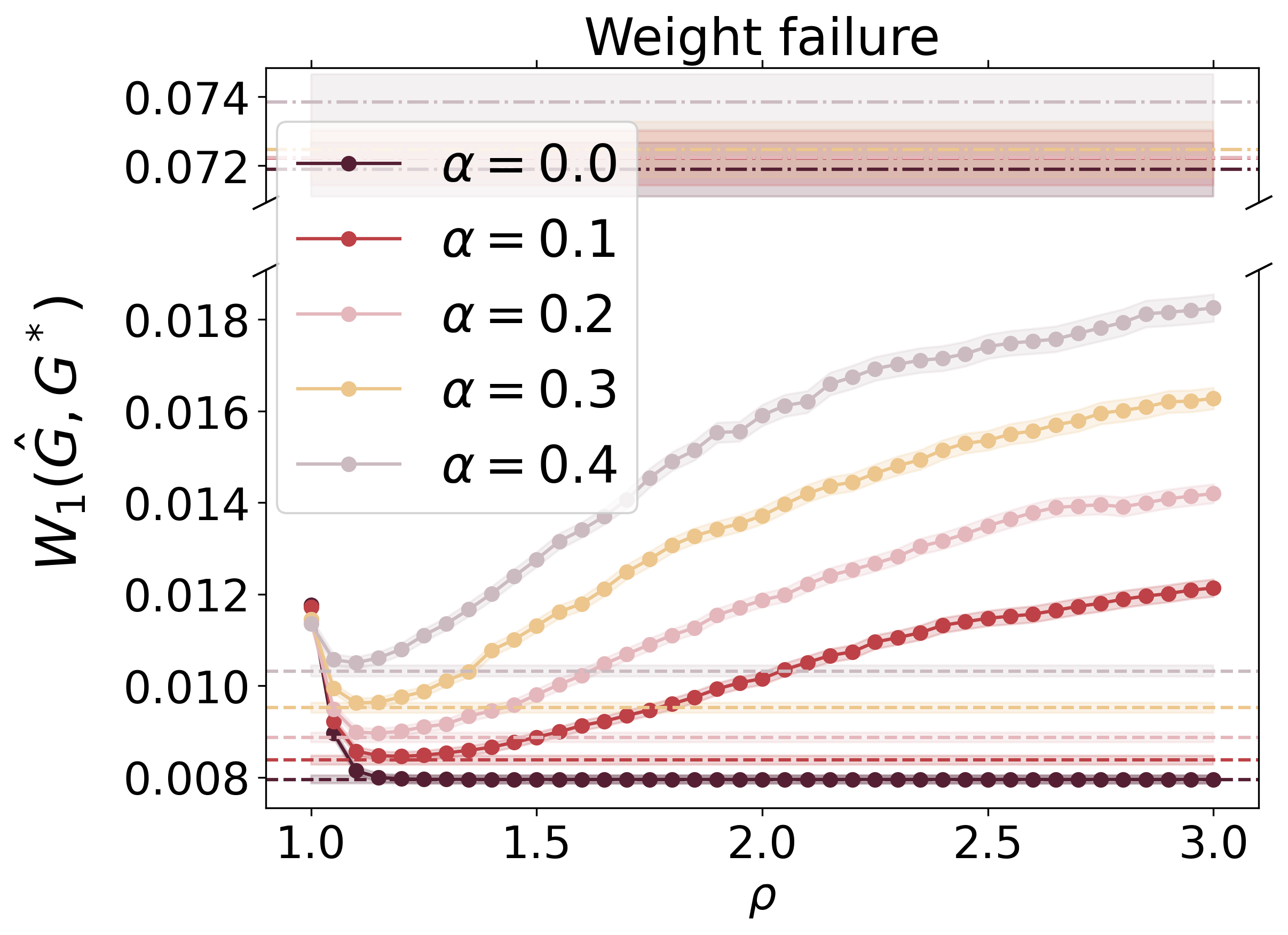}\\
\includegraphics[width=0.3\textwidth]{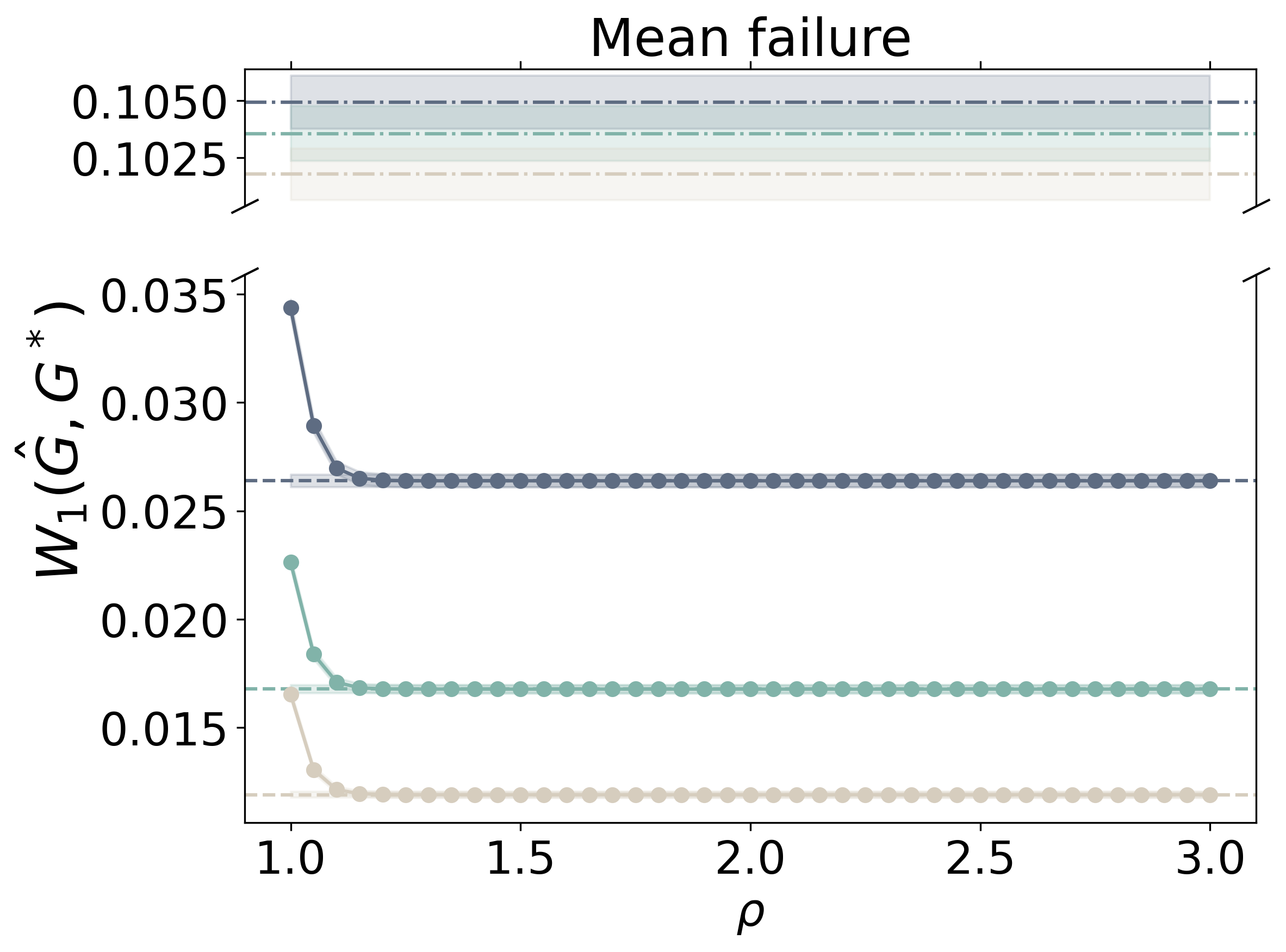}
\includegraphics[width=0.3\textwidth]{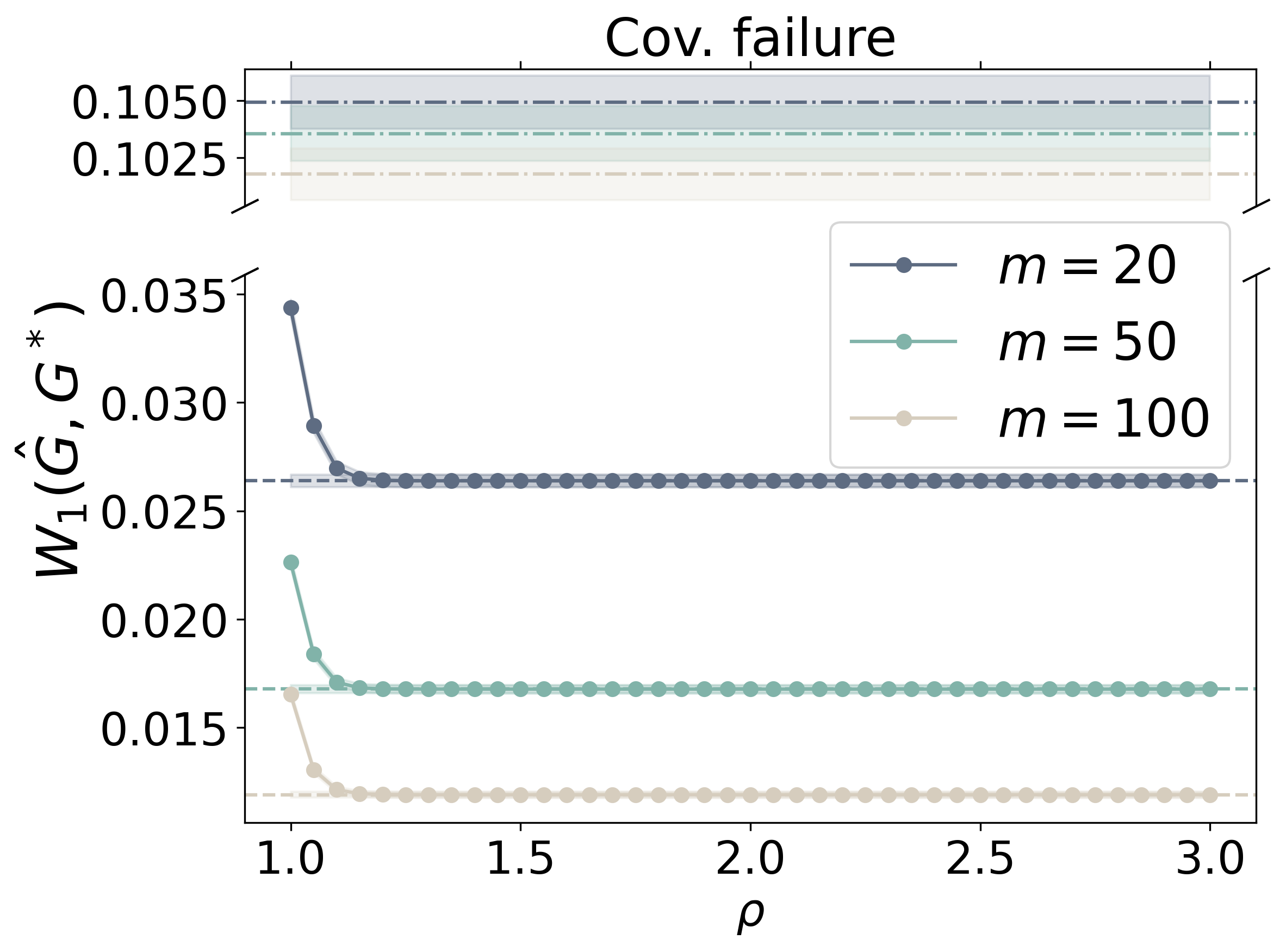}
\includegraphics[width=0.3\textwidth]{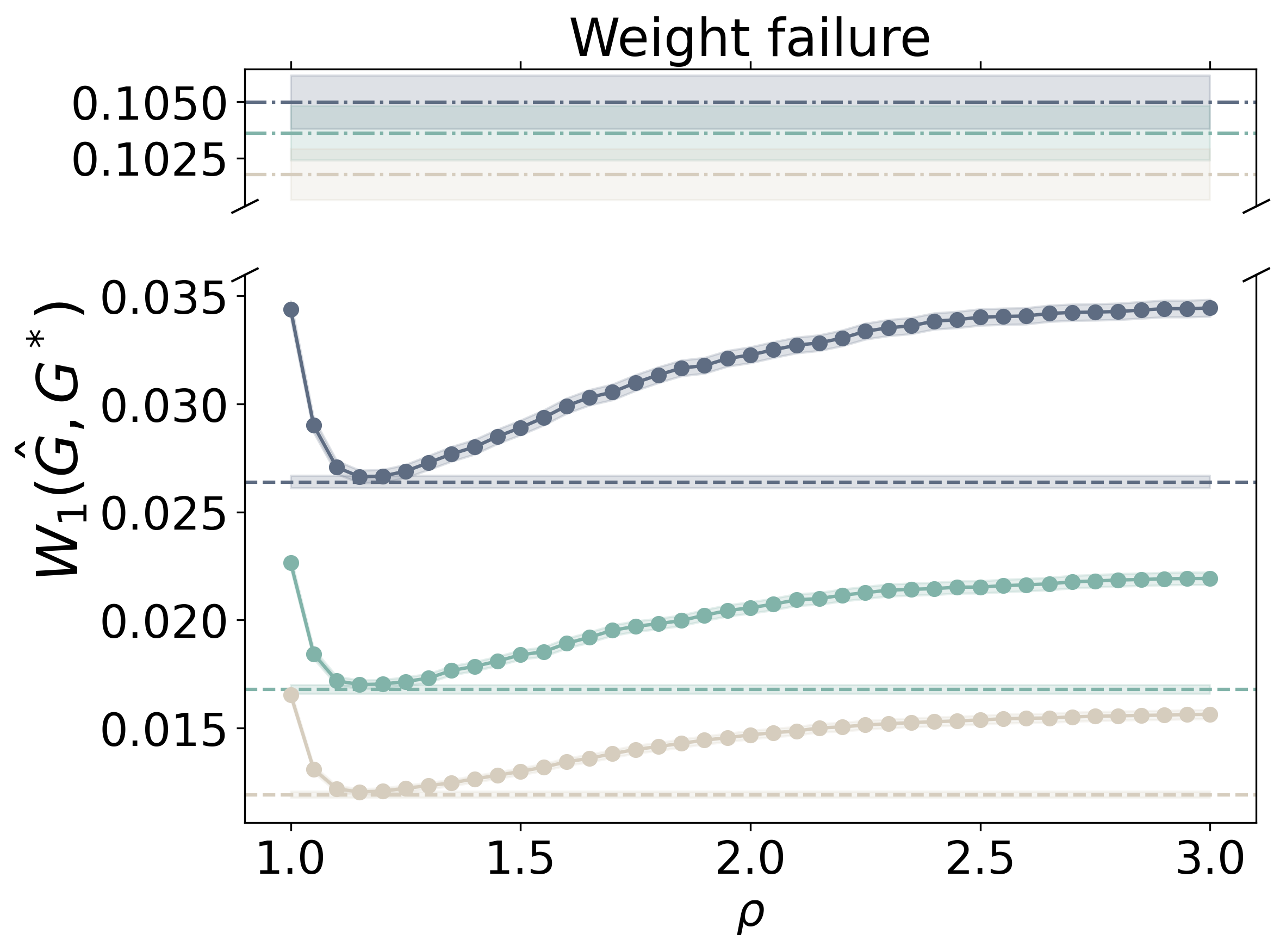}
\caption{The $W_1$ values of the DFMR($\rho$) approach as a function of the inflation factor $\rho$ under varying failure types, failure rates, and numbers of local machines. 
The dotted lines represent the DFMR($\rho$) approach, the dashed lines indicate the performance of the Oracle approach, and the dash-dotted lines correspond to the performance of the COAT approach.}
\label{fig:EEI_threshold}
\end{figure}
This analysis demonstrates that, while $\rho$ is a parameter in our method, it is not overly sensitive and does not require meticulous tuning. 
Based on these results, we identify $\rho=1.3$ as a robust and effective choice across all scenarios. 
Consequently, we adopt $\rho=1.3$ for the DFMR($\rho$) approach in subsequent experiments and we denote this estimator as DFMR(1.3). 
This choice aligns with the theoretically recommended value of $m^{1/10}$, which evaluates to $1.3$, $1.5$, and $1.6$ for $m=20$, $50$, and $100$, respectively. 
The robustness of DFMR to the choice of $\rho$ further enhances its practicality and ease of deployment in real-world applications.

\subsubsection{Simulation results} 
We present overall comparison in terms of performance metric $W_1$ for all estimators over $R=300$ repetitions. 
The ARI has a similar pattern and is deferred to Appendix~\ref{app:more_results} for conciseness.

\textbf{Performance with respect to total sample sizes and Byzantine failure rates.}
Figure~\ref{fig:w1_fixoverlap} shows the Wasserstein distance ($W_1$) of different estimates across various experimental settings, while maintaining a fixed number of local machines ($m=100$) and a specified degree of overlap (\texttt{MaxOmega}$=0.3$).
The rows represent results for mean failure, covariance failure, and weight failure respectively, while the columns denote failure rates ranging from $0.1$ to $0.4$. 
The relative performance of methods remains consistent across other combinations of $m$ and \texttt{MaxOmega}. 
Therefore, we present only the results under the most challenging scenario where \texttt{MaxOmega} is at its maximum.

Our proposed estimators clearly outperform other estimators in general, and their performance improves as $N$ increases.
The performance of the proposed DFMR(1.3) matches that of the oracle closely, while our DFMR(1) is comparable to oracle and is much better than the TRIM even though the latter also utilizes $50\%$ of local estimates. 
Under mean failure and covariance failure, Vanilla is not Byzantine tolerant.
The COAT estimator is Byzantine-tolerant, but it has the weakest performance among Byzantine-tolerant methods because of its reliance on a single local estimate. 
\begin{figure}[!htbp]
  \centering
  \includegraphics[width=0.9\textwidth]{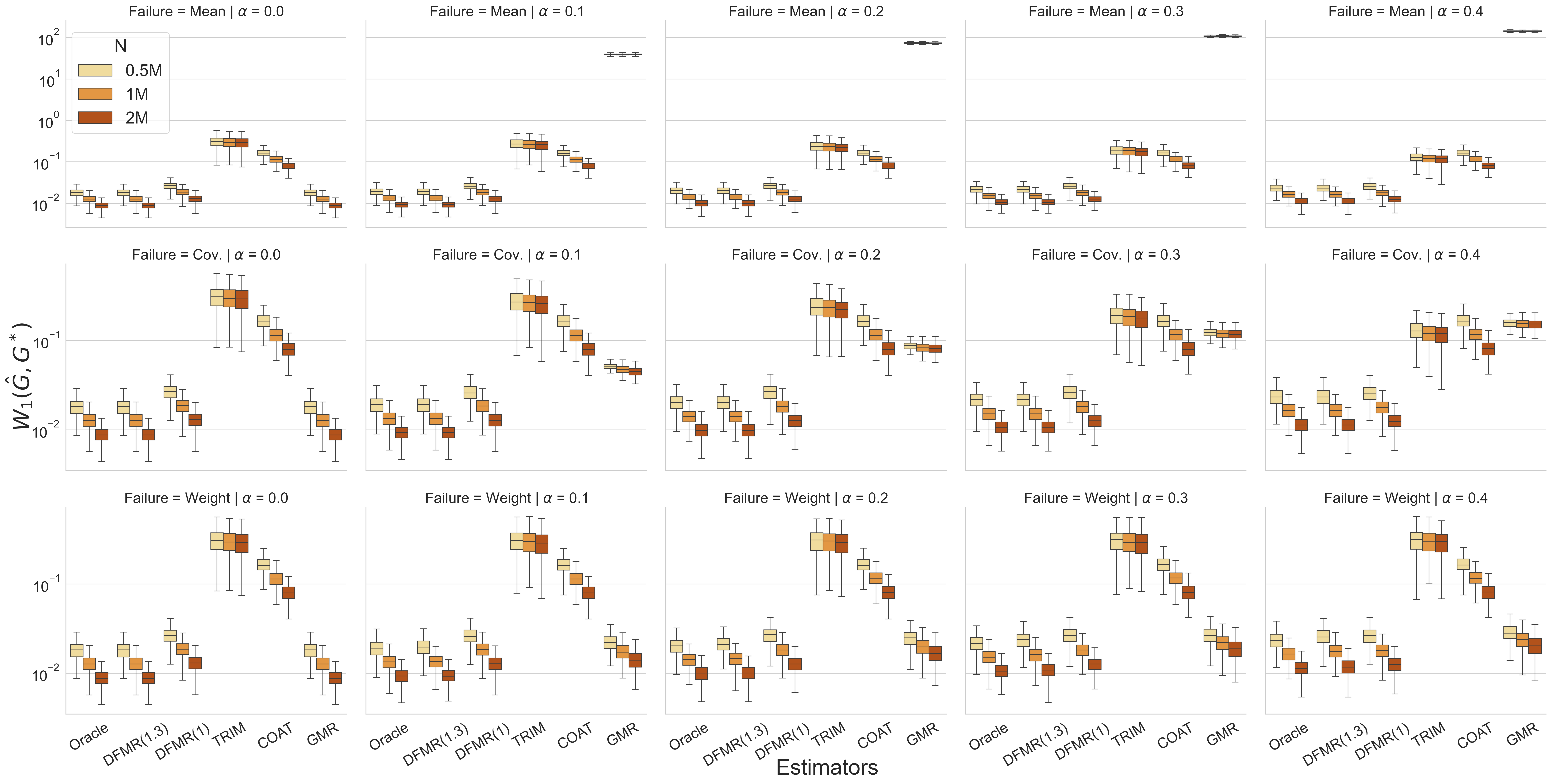}
  \caption{$W_1$ value comparison as $N$ and $\alpha$ varies, with $m=100$ and \texttt{MaxOmega}$=0.3$.}
  \label{fig:w1_fixoverlap}
\end{figure}
Under weight failure, both our DFMR(1) and DFMR(1.3) estimators remain highly competitive with the oracle estimator.
Interestingly, Vanilla also works well, whereas other estimators perform poorly, with TRIM being particularly ineffective. 
The suboptimal performance of TRIM is largely due to its ignorance to the mixing weights in detecting potential failed local estimates.

\textbf{Performance when the number of local machines scale up.}
In addition to sample sizes, the scalability of the estimators in terms of the number of local machines ($m$) is crucial in distributed learning. 
Figures~\ref{fig:w1_fixn_varym} shows the $W_1$ values of various estimators as $m$ increases in a representative scenario where \texttt{MaxOmega}$=0.3$ and $\alpha=0.1$, under various failure scenarios.
We examine two cases: first, where the local machine sample size $n$ is fixed, as illustrated in Figure~\ref{fig:w1_fixn_varym}, and second, where the total sample size $N$ is fixed, as shown in Figure~\ref{fig:w1_fixN_varym}. 
In both cases, the relative performance of different methods aligns with the findings in the previous setting.
Our proposed DFMR estimators consistently outperform other approaches.
\begin{figure}[!htbp]
\centering
\includegraphics[width=0.9\textwidth]{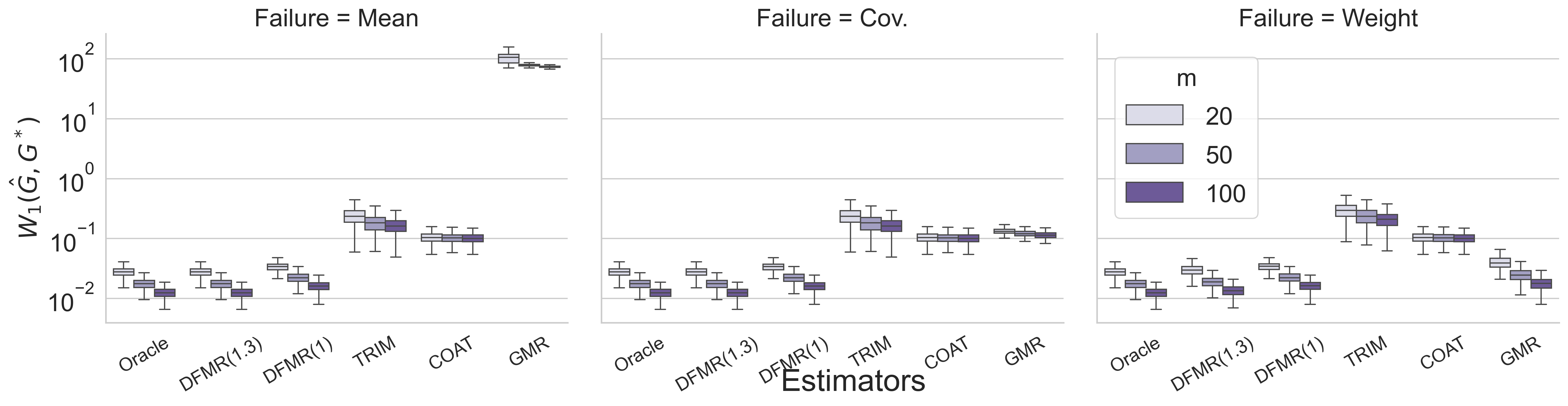}
\caption{$W_1$ values of different methods as $m$ varies, with $\alpha=0.1$, and \texttt{MaxOmega}$=0.3$. 
The local sample size $n=5000$, the total sample $N=nm$ increasing with $m$.}
\label{fig:w1_fixn_varym}
\end{figure}

In the first case where the total sample size increases with more machines, the performance of the proposed methods improves.
This is anticipated as our proposed estimators are consistent.
In the second case where the total sample size does not increase with more machines, the performance of the proposed methods does not change with different number of machines.
In comparison, the performance of COAT is determined solely by the local sample size $n$. 

\begin{figure}[!htbp]
\centering
\includegraphics[width=0.9\textwidth]{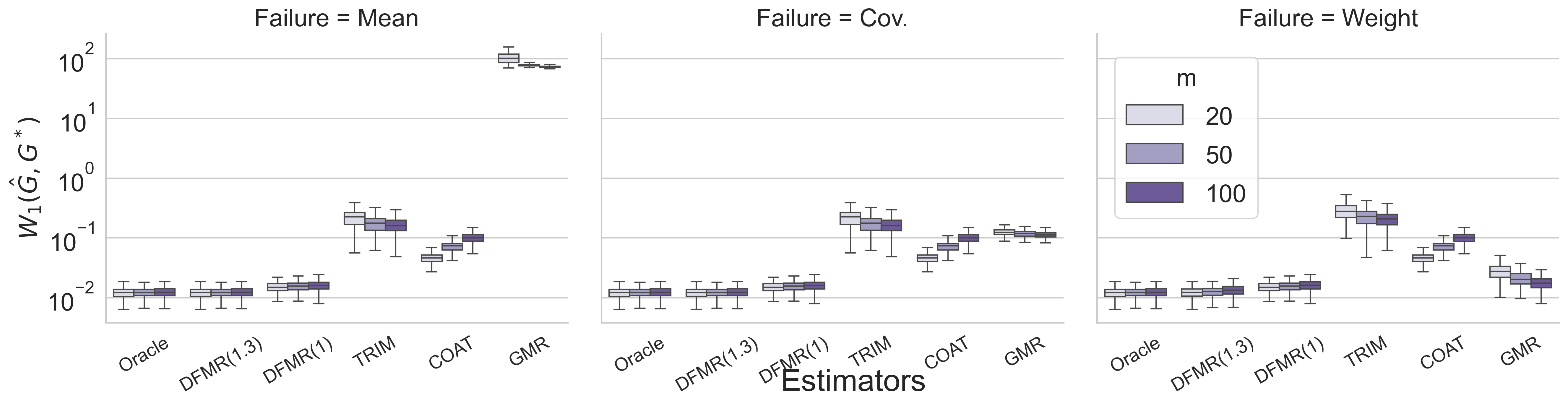}
\caption{$W_1$ values of different methods as $m$ varies, with $\alpha=0.1$, and \texttt{MaxOmega}$=0.3$. The total sample size $N=500K$, the local sample $n=N/m$ decreasing with $m$.}
\label{fig:w1_fixN_varym}
\end{figure}

\textbf{Performance with different degree of overlaps.} 
We next examine the performance of different approaches while varying the degree of overlap (\texttt{MaxOmega}). 
Figure~\ref{fig:w1_fixalpha_fixn} illustrates the $W_1$ of different methods in a representative case where $\alpha=0.2$, $m=100$, and local sample size $n=5000$, under various failure scenarios.
\begin{figure}[!htbp]
\centering
\includegraphics[width=0.9\textwidth]{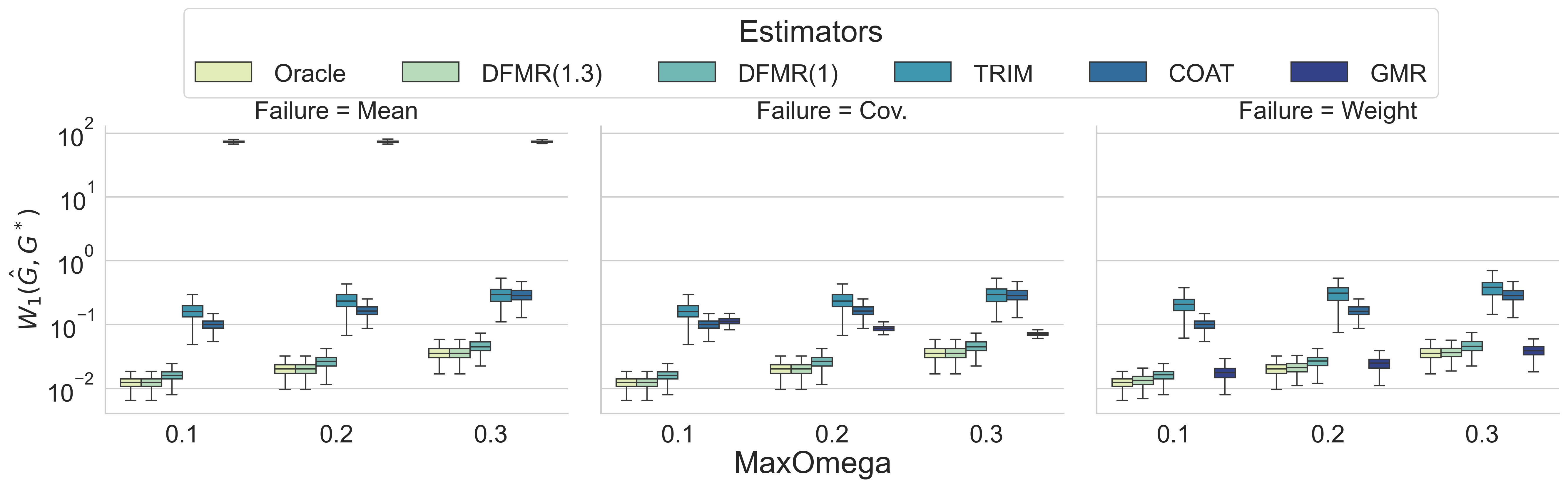}
\caption{The $W_1$ of different methods as the degree of overlap \texttt{MaxOmega} varies, with $\alpha=0.2$, $m=100$, and local sample size $n=5000$.}
\label{fig:w1_fixalpha_fixn}
\end{figure}
Our estimators remain comparable to the oracle estimator, regardless of the failure type and degree of overlap.
Other methods do not compete effectively with our method.

\subsection{Distance concentration under different failure types}
\label{sec:exp_distance_concentration}
To further explain the good empirical performance of DFMR demonstrated in the previous section, we perform further numerical simulations to show how the filter step in our method can distinguish between failed and failure-free local estimates, based on the concentration of the $L^2$ distance of failure-free mixture densities from the COAT.
In this experiment, we still consider Gaussian mixture and set $m=100$, $n=10^{4}$, and failure rate to $\alpha=0.3$.
We randomly generate $m$ IID sets of samples each with size $n$ from a randomly generated Gaussian mixture in $d$-dimensional space with $K$-components.
We consider different values for $K$ and $d$, hence leading to different number of parameters.
Same as the previous experiment, we fit a Gaussian mixture on each local machine.
After which, we randomly attack $\alpha m$ machines using different attack strategies as described earlier in Section~\ref{sec:exp_setting}. 
The results are shown in Figure~\ref{fig:dist_concentration}. 
The panels, arranged from left to right, correspond to an increasing number of parameters, with configurations $K=5, d=2$, $K=5, d=10$, and $K=10, d=10$, respectively.
\begin{figure}[!htbp]
\centering
\includegraphics[width=0.3\textwidth]{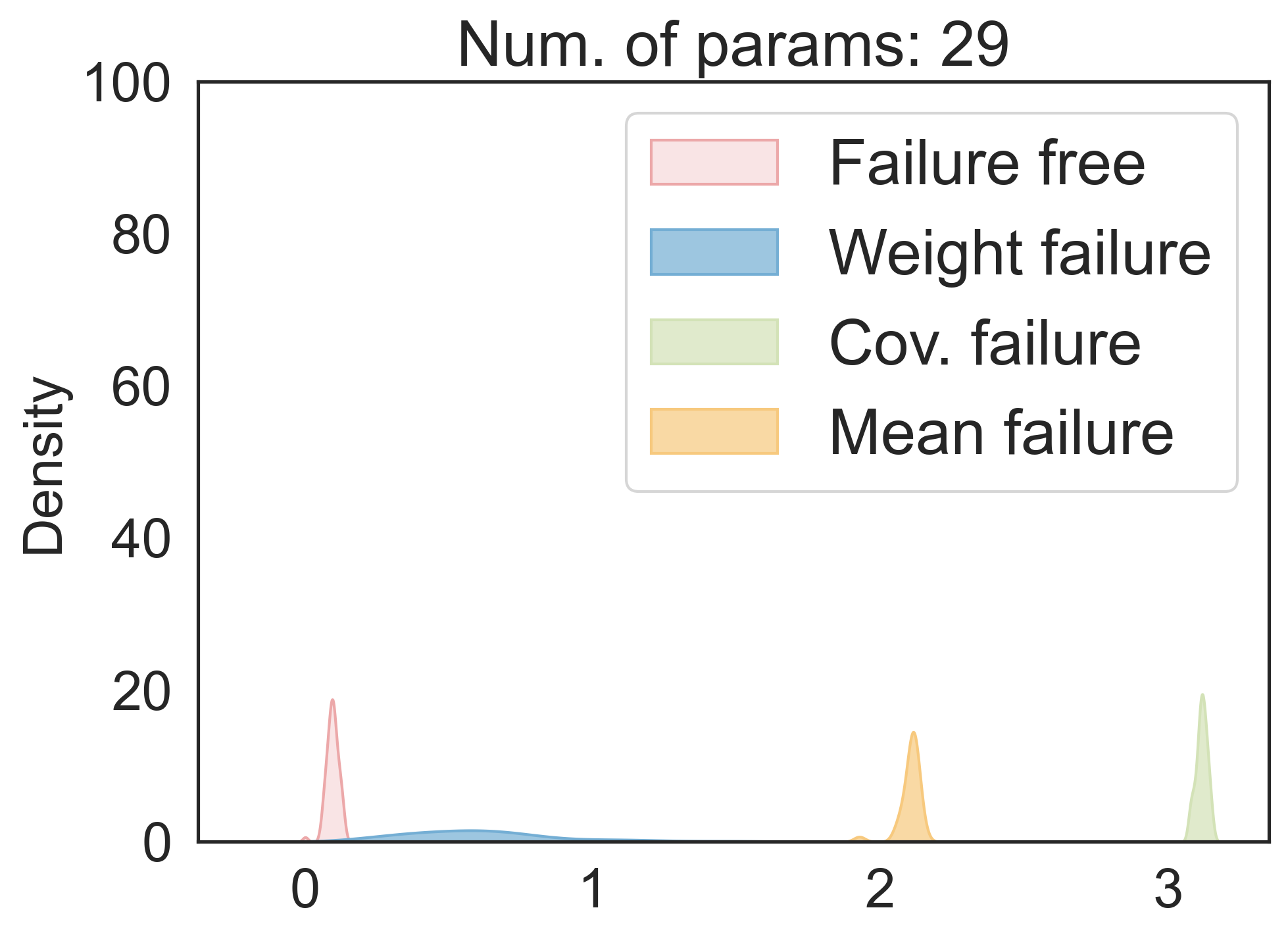}
\includegraphics[width=0.3\textwidth]{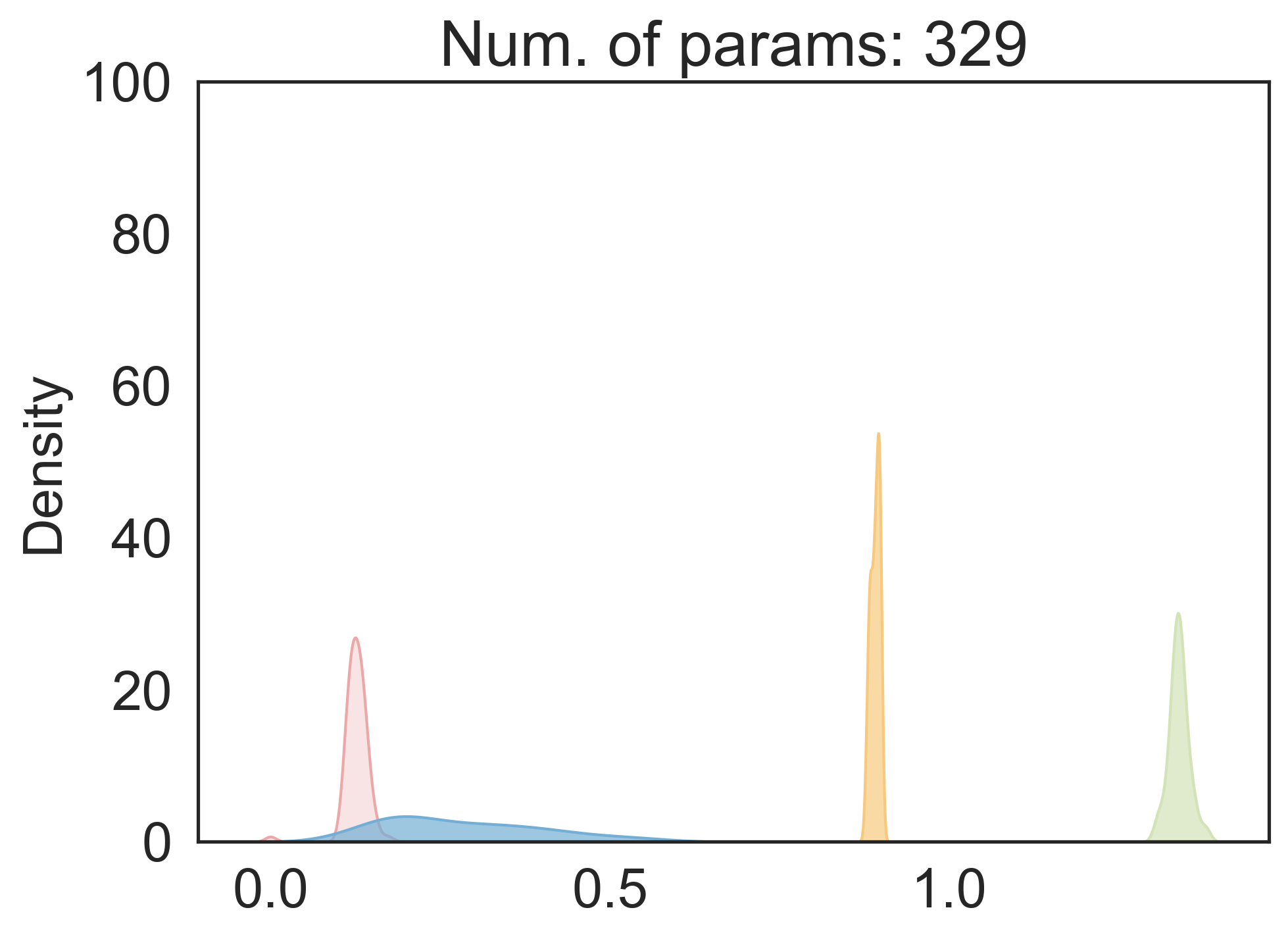}
\includegraphics[width=0.3\textwidth]{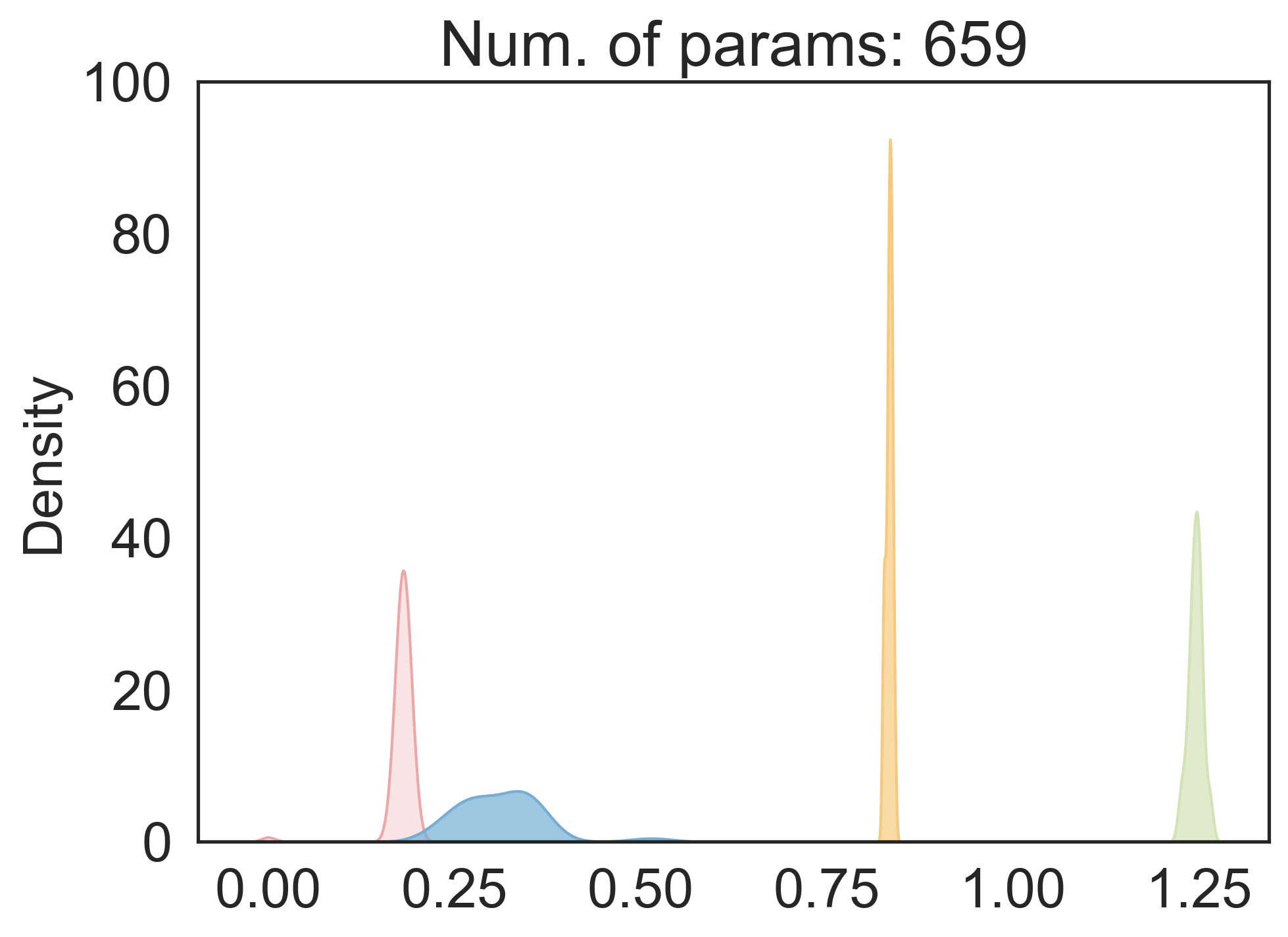}
\caption{Distribution of distances to the COAT mixture density under different failure types, represented by distinct colours, as the number of parameters increases.}
\label{fig:dist_concentration}
\end{figure}

Within each panel, we observe a pronounced gap in the distribution of distance to the true mixture density for failure-free local estimates and that for failed local estimates under both mean and covariance failure models, though the distance distribution under the weight failure model is less distinguishable. 
In fact, as the number of parameters grows, the distributions of the $L^2$ distances become sharper, highlighting the phenomenon of distance concentration in higher dimensions \citep{vershynin2018high}.
This complements our theoretical results, which are asymptotic in $n$ and $m$ but hold only for fixed $d$ and $K$, suggesting that DFMR continues to perform well even as $d$ and $K$ are large compared with the sample size $N$.
In other words, we expect DFMR to be applicable to modern high-dimensional datasets with high data heterogeneity.

%% file: sections/real_data.tex

\section{Real data analysis}
\label{sec:real_data}
Although naturally occurring Byzantine failures are an important application scenario, real-world failure data are rarely accessible due to security and confidentiality constraints. 
Consequently, it is standard in the distributed and robust learning literature to evaluate methods by introducing controlled perturbations into real datasets. 
Following this approach, we study Byzantine-tolerant aggregation of finite Gaussian mixtures for clustering using the NIST dataset~\citep{grother2016nist}, a widely used benchmark for handwritten digit recognition, under experimentally induced failures similar to those in~\citet{yin2018byzantine, zhu2023byzantine}.

We use the second edition of the dataset~\footnote{\scalebox{0.9}{Available at \url{https://www.nist.gov/srd/nist-special-database-19}.}}. 
It comprises approximately 4 million images of handwritten digits and characters (0--9, A--Z, and a--z) by different writers, partitioned into training and test sets. 
Figure~\ref{fig:nist} (a) gives example images in the dataset.
The number of training images for each digit and letter A--J are listed in Table~\ref{tab:nist_summary}.
\begin{table}[htpb]
\caption{Numbers of training images in the NIST dataset.}
\label{tab:nist_summary}
\centering
\small
\begin{tabular}{c||cccccccccc}
\toprule
Digits & 0 & 1 & 2 & 3 & 4 & 5 & 6 & 7 & 8 & 9 \\
\hline
Training & 34803& 38049& 34184& 35293& 33432& 31067& 34079& 35796& 33884& 33720\\
Test & 5560& 6655& 5888& 5819& 5722& 5539& 5858& 6097& 5695& 5813 \\
\midrule
Letter & A & B & C & D & E & F & G & H & I & J \\
Training & 7010& 4091& 11315& 4945& 5420& 10203& 2575& 3271& 13179& 3962\\
\bottomrule
\end{tabular}
\end{table}

Each image is a $128 \times 128$ pixel greyscale matrix, with entries ranging from $0$ to $1$, representing the darkness of the corresponding pixels, where a darker pixel has a value closer to $1$.
We reduce each image to a $d=50$ real valued vector via deep convolutional neural network. 
Further details of the datasets and the architecture of the neural network are provided in Appendix~\ref{app:real_data}.
Included in Figure~\ref{fig:nist} (b) is the t-SNE visualisation of the $50$-dimensional features of randomly selected samples in the dataset.
The green and red dots are randomly selected digits and letters respectively.
It is evident that features of the same digit or the same letter are clustered.
Given this, an approach to build digit recognizer is to regard the features of each digit as a random sample from a distinct Gaussian distribution. 
Thus, the pooled data forms a sample from a finite Gaussian mixture of order $K=10$. 

\begin{figure}[htbp]
\centering
\begin{subfigure}{0.4\textwidth}
 \includegraphics[width=0.18\textwidth]{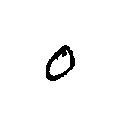}
 \includegraphics[width=0.18\textwidth]{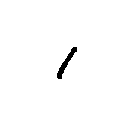}
 \includegraphics[width=0.18\textwidth]{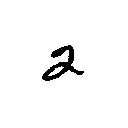}
 \includegraphics[width=0.18\textwidth]{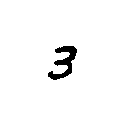}
 \includegraphics[width=0.18\textwidth]{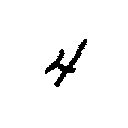}\\
 \includegraphics[width=0.18\textwidth]{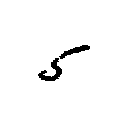}
 \includegraphics[width=0.18\textwidth]{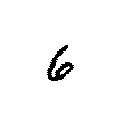}
 \includegraphics[width=0.18\textwidth]{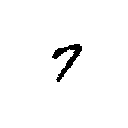}
 \includegraphics[width=0.18\textwidth]{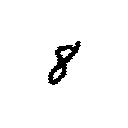}
 \includegraphics[width=0.18\textwidth]{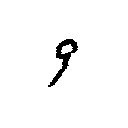}\\
 \includegraphics[width=0.18\textwidth]{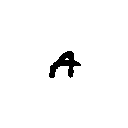}
 \includegraphics[width=0.18\textwidth]{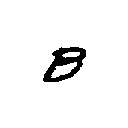}
 \includegraphics[width=0.18\textwidth]{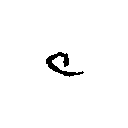}
 \includegraphics[width=0.18\textwidth]{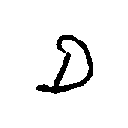}
 \includegraphics[width=0.18\textwidth]{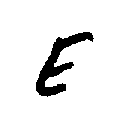}\\
 \includegraphics[width=0.18\textwidth]{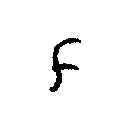}
 \includegraphics[width=0.18\textwidth]{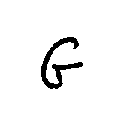}
 \includegraphics[width=0.18\textwidth]{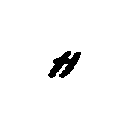}
 \includegraphics[width=0.18\textwidth]{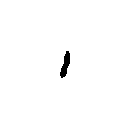}
 \includegraphics[width=0.18\textwidth]{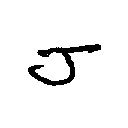}
 \caption{Raw data.}
 \label{fig:nist_raw}
\end{subfigure}
\begin{subfigure}{0.4\textwidth}
\includegraphics[width=\textwidth]{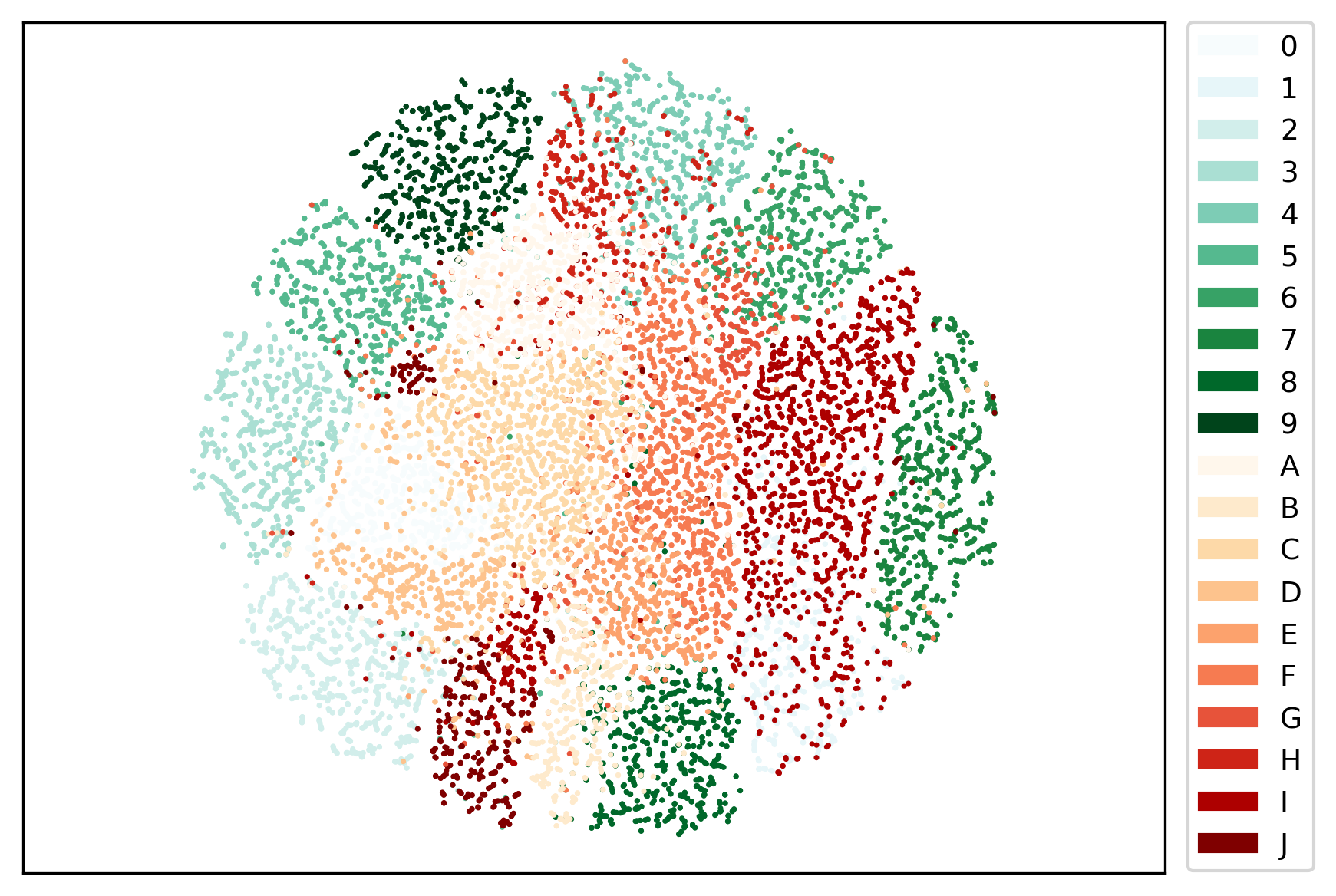}
\caption{t-SNE visualisation.}
\label{fig:nist_tsne}
\end{subfigure}
\caption{Random examples of digits and letters in the NIST dataset.}
\label{fig:nist}
\end{figure}

For illustration, we randomly select a dataset of size $N = 150,000$ and partition it into $m = 30$ subsets. 
We train a $K = 10$ order Gaussian mixture model on each local machine. 
Next, we introduce Byzantine failure to a number of local estimates in a specific manner.
For each Byzantine failure machine, we independently sample $n = 5,000$ random images of letters A--J and fit a $K = 10$ mixture model to this data. 
We then replace the mixture on the corresponding Byzantine failure machine with the fitted model. 
Note that in Figure~\ref{fig:nist} (b), the letter features differ from digit features, thus serving as a natural source of Byzantine failure.
Both procedures are repeated $R = 300$ times.

We use the estimators in Section~\ref{sec:estimators} for $\alpha\in \{0.0, 0.1,0.2,0.3,0.4\}$. 
These mixture estimates are then used to cluster images of handwritten digits in the test set. 
The ARI between the true label of the image and the predicted cluster label on the test set is given in Table~\ref{tab:real_data}.
\begin{table}[htbp]
\centering
\caption{Median (IQR) ARI on the NIST digits test set of different methods under different failure rates. For the DFMR($\rho$) with $\rho\in [1.35, 3.0]$ leads to the same ARI.}
\label{tab:real_data}
\resizebox{0.8\textwidth}{!}{
\begin{tabular}{ccccccc}
\toprule
$\alpha$ & Oracle & DFMR($\rho$) & DFMR(1) & Trim & COAT & Vanilla \\
\midrule
$0.0$ & 0.9195 \scriptsize{(0.0014)} & \textbf{0.9195 \scriptsize{(0.0014)}} & \textbf{0.9186 \scriptsize{(0.0018)}} & 0.9034 \scriptsize{(0.0116)} & 0.8896 \scriptsize{(0.0108)} & 0.9195 \scriptsize{(0.0014)} \\
$0.1$ & 0.9193 \scriptsize{(0.0015)} &  \textbf{0.9194 \scriptsize{(0.0014)}} & \textbf{0.9185 \scriptsize{(0.0018)}} & 0.9035 \scriptsize{(0.0118)} & 0.8898 \scriptsize{(0.0106)} & 0.9043 \scriptsize{(0.0050)} \\
$0.2$ & 0.9192 \scriptsize{(0.0015)} &  \textbf{0.9194 \scriptsize{(0.0013)}} & \textbf{0.9186 \scriptsize{(0.0020)}} & 0.9042 \scriptsize{(0.0112)} & 0.8898 \scriptsize{(0.0106)} & 0.9046 \scriptsize{(0.0044)} \\
$0.3$ & 0.9189 \scriptsize{(0.0017)} &  \textbf{0.9194 \scriptsize{(0.0015)}} & \textbf{0.9186 \scriptsize{(0.0018)}} & 0.9040 \scriptsize{(0.0107)} & 0.8898 \scriptsize{(0.0104)} & 0.9041 \scriptsize{(0.0046)} \\
$0.4$ & 0.9189 \scriptsize{(0.0017)} &  \textbf{0.9195 \scriptsize{(0.0014)}} & \textbf{0.9186 \scriptsize{(0.0018)}} & 0.9037 \scriptsize{(0.0117)} & 0.8892 \scriptsize{(0.0110)} & 0.9042 \scriptsize{(0.0049)} \\
\bottomrule
\end{tabular}}
\end{table}

Our proposed methods, DFMR(1) and DFMR($\rho$), demonstrate greater robustness to Byzantine failures compared to the Vanilla method. 
They also exhibit higher efficiency than other baseline approaches.
In the DFMR($\rho$) approach, we observe that a wide range of inflation factors, $\rho\in[1.35, 3.0]$, yields consistent ARI values. 
This indicates the robustness of DFMR($\rho$) to variations in the choice of the inflation factor.
Furthermore, we found that the DFMR($\rho$) approach outperforms the Oracle estimator, particularly as the failure rate approaches $0.4$.
This aligns with our theoretical analysis, which suggests that filtering out failure-free estimates can enhance the efficiency of the estimator.

The above semi-synthetic design provides a controlled framework that combines the structural richness of real data with carefully introduced adversarial perturbations, enabling reproducible and interpretable evaluation of DFMR. 
This setting goes beyond standard simulations by incorporating realistic feature representations while retaining control over the failure mechanism. 
Moreover, it offers a complementary perspective to purely synthetic experiments: since the learned feature embeddings for each digit are not guaranteed to follow a Gaussian distribution, the mixture model is inherently misspecified. 
The strong empirical performance of DFMR in this setting therefore further illustrates its robustness to both Byzantine contamination and model misspecification.

%% file: sections/conclusion.tex

\section{Discussion and concluding remarks}
\label{sec:conclusion}
This paper addresses Byzantine-tolerant aggregation in split-and-conquer learning of finite mixture models. 
We propose Distance Filtered Mixture Reduction (DFMR), a robust aggregation procedure tailored to mixture models, and show that it is both computationally efficient and statistically well-founded. 
Our analysis establishes that, even under adversarial contamination, DFMR achieves near-optimal estimation accuracy while remaining scalable to large distributed systems.

It is useful to place our contribution in the broader context of distributed mixture estimation. 
When the true mixing distribution is fixed with well-separated components and the number of components $K$ is known, label switching becomes asymptotically negligible, as local estimators naturally align. 
Thus, label switching is not the primary difficulty in such regimes, and various aggregation strategies could in principle recover the correct alignment. 
The main challenge addressed here is instead to achieve robustness against Byzantine failures while preserving statistical efficiency. 
Our results show that DFMR effectively filters out corrupted local estimators and aggregates the remaining information in a principled manner, leading to optimal convergence guarantees. 
Technically, this requires non-trivial analysis: the filtering step is based on $L^2$ distances between mixture densities, whereas estimation accuracy is measured in terms of Euclidean distances between parameters. 
Relating these metrics, establishing concentration properties, and controlling misclassification in the filtering step are key challenges underlying our theoretical results.

Our analysis assumes that the number of components $K$ is known and that the mixing distribution remains fixed as the sample size increases in the usual asymptotic sense. 
While standard, these assumptions exclude more challenging regimes. 
For example, allowing component locations to vary with the sample size or to become less well-separated would require fundamentally new techniques. 
In such settings, key ingredients of our analysis,
such as concentration results based on asymptotic normality and Berry–Esseen-type arguments,
may no longer apply. 
Extending Byzantine-robust aggregation to these regimes is therefore non-trivial and an important direction for future work. 
At the same time, we view the present results as a necessary foundation for studying these more demanding scenarios. 
Similarly, selecting $K$ with theoretical guarantees remains open even in centralized settings, and becomes even more challenging in distributed environments.

To reduce communication costs, we focus on a single-round communication regime. 
Nevertheless, the method naturally extends to multi-round settings, where the aggregated estimator can serve as initialization for subsequent EM iterations on local machines. 
While such extensions are promising, our primary goal here is to establish a robust and theoretically grounded aggregation procedure, leaving multi-round refinements for future work.

Finally, we assume that all components on a failed machine are affected by Byzantine failures. 
In practice, some components may remain intact, and exploiting such partial reliability could further improve efficiency and robustness—an interesting direction for future research.

The mixture reduction technique of \citet{zhang2022distributed} has been extended to distributed learning of mixtures of experts~\citep{chamroukhi2023distributed}. 
Mixtures of experts are powerful tools for scaling neural networks to handle complex and heterogeneous data, and are widely used in large language models~\citep{du2022glam, liu2024deepseek}, computer vision~\citep{lin2024moe}, and multimodal learning~\citep{mustafa2022multimodal}. 
Training such models requires substantial computational resources, making distributed learning essential for scalability~\citep{duan2024efficient}. 
Ensuring robustness in these settings is therefore critical, and the ideas underlying DFMR may provide a useful starting point for developing Byzantine-tolerant procedures in these more complex models.

%% file: sections/appendix_algorithms.tex
\section{Numerical algorithms}
\subsection{EM algorithm}
\label{app:em_algorithm}
The Maximum Likelihood Estimate (MLE) is widely used for local inference under finite mixture models. 
The preferred numerical method for computing the MLE is the Expectation-Maximisation (EM) algorithm, which is most conveniently described using the latent variable interpretation of the mixture model.

For the $j$th unit with an observed value $x_{ij}$ from the finite mixture $f_{G}(x)$ on the $i$th machine, there exists a latent variable $z_{ij}$ associated with the observed value. 
When the latent variables $\{z_{ij}, j \in [n]\}$ are known, the complete dataset $\{(z_{ij}, x_{ij})\}$ is obtained. 
This leads to the complete data log-likelihood function:
\[
\ell_{n}^c(G) 
 = \sum_{j=1}^{n} 
 \sum_{k=1}^K \mathbbm{1}(z_{ij}=k)\log\{w_kf(x_{ij}; \theta_k)\}.
 \]
However, the complete data log-likelihood cannot be directly used to estimate $G$ since the latent variables are unknown. 
In the EM algorithm, one replaces $z_{ij}$ in the complete data log-likelihood by its conditional expectation in each iteration.
Let $G^{(t)}$ denote the value of the mixing distribution after the $t$th iteration. 
The conditional probability of $z_{ij}$ given the observed data is:
\begin{equation}
\label{eq:conditional_expectation}
w_{ijk}^{(t)} 
= 
\sP(z_{ij}=k| G^{(t)}, \gX)
= 
\dfrac{w_k^{(t)} f(x_{ij};\theta_k^{(t)})}
{\sum_{k'=1}^{K} w_{k'}^{(t)} f(x_{ij}; \theta_{k'}^{(t)})}.
\end{equation}
The conditional expectation of $\ell_{n}^c(G) $ is then given by:
\begin{equation*}
Q(G|G^{(t)}) 
= 
\sum_{j=1}^{n} \sum_{k=1}^K w_{ijk}^{(t)}\log \{w_kf(x_{ij}; \theta_k)\}.
\end{equation*}
This constitutes the E-step of the algorithm.

In the M-step, instead of maximizing the complete data log-likelihood $\ell_{n}^c(G)$, the algorithm seeks the maximizer of $Q(G|G^{(t)})$ for $G \in \sG_K$. 
This task is simpler because $Q(G |G^{(t)})$ is additive in subpopulation parameters which allows for easier numerical computation.
As a result, the updated mixing distribution $G^{(t+1)}$ is composed of mixing weights:
\begin{equation*}
w_{k}^{(t+1)} =   n^{-1} \sum_{j=1}^{n} w_{ijk}^{(t)}
\end{equation*}
and subpopulation parameters: 
\begin{equation*}
\theta_{k}^{(t+1)} = \argmax_{\theta} \left\{\sum_{j=1}^{n} w_{ijk}^{(t)} \log f(x_{ij};\theta)\right\}.
\end{equation*}
For many parametric subpopulation models $\gF$, there exists an analytical solution for $\theta_{k}^{(t+1)}$, making the EM iteration straightforward to execute. 
Repeating the iteration generates a sequence of mixing distributions. 
Under certain conditions,~\citet{wu1983convergence} show that $\ell_{n}(G^{(t)}) = \sum_{j=1}^{n} \log f_{G^{(t)}}(x_{ij})$ is an increasing sequence, and the sequence $\{G^{(t)}, t=1,2,\ldots\}$ converges to a local maximum of the log-likelihood function.

Despite its widespread use, the EM algorithm suffers from a slow convergence rate and can easily become trapped in local maxima. 
Various approaches have been proposed to accelerate convergence, see~\citet{meilijson1989fast},~\citet{meng1993maximum},~\citet{liu1994ecme}, and~\citet{balakrishnan2017statistical}.

\textbf{Issues with the MLE under finite Gaussian mixtures and finite location-scale mixtures.}
The MLE is not well-defined for widely used finite Gaussian and location-scale mixtures in general. 
To address this issue, several approaches have been proposed. 
For example,~\citet{hathaway1985constrained} introduces a constrained maximum likelihood formulation to avoid the unboundedness of the likelihood function, resulting in an estimator with desirable consistency properties. 
However, this approach alters the parameter space, which may be undesirable. 
Alternatively,~\citet{ridolfi2001penalized} proposes a Bayesian approach by imposing an inverse gamma prior on the scale parameter in Gaussian mixture model. 
The posterior mode, known as the maximum a posterior (MAP) estimator, is later shown to be consistent by~\citet{chen2008inference}.

In~\citet{chen2008inference}, a penalized log-likelihood function is defined as:
\begin{equation*}
p\ell_n(G) = \ell_n(G) - a_n\sum_{k=1}^K p(\theta_k)
\end{equation*}
where $p(\cdot)$ is a penalty function and $a_n > 0$ is the penalty strength. 
The penalized MLE (pMLE) is then:
\begin{equation*}
\widehat{G}^{\text{pMLE}}= \argsup_{G\in\sG_{K}} p\ell_n(G)
\end{equation*}
The penalty function is applied to individual subpopulation parameters, ensuring ease of implementation in the modified EM algorithm.
Three penalty functions are recommended for different mixtures of $\gF$:
\begin{itemize}
\item 
Under \emph{finite Gaussian mixtures}, let $S_{x}$ be the sample covariance matrix of $\gX_{i}$.
\citet{chen2009inference} recommends the penalty function:
\[p(\theta) = \text{tr}(\Sigma^{-1}S_{x}) + \log\text{det}(\Sigma)\]
where $\text{tr}(\cdot)$ is the trace of a square matrix.
The penalty function reduces to $p(\theta) = s_x^2/\sigma^2 + \log \sigma^2$ under univariate Gaussian mixtures~\citep{chen2008inference}. 

\item 
Under \emph{finite location-scale mixtures} with location $\mu$ and scale $\sigma$, the sample variance $s_x^2$ in the previous penalty function can be replaced by a scale-invariant statistic, such as the squared sample interquartile range.
This is particularly useful when the variance of $f_0(\cdot)$ is not finite.

\item 
Under finite mixture of \emph{two-parameter Gamma distributions}, the likelihood is also unbounded. 
\citet{chen2016consistency} recommends the penalty function to be $p(\theta)= r - \log r$, where $r$ is the shape parameter in the Gamma distribution.
\end{itemize}

\textbf{The EM algorithm for penalized MLE.}
The EM algorithm can be easily adapted to compute the pMLE with the recommended penalty functions~\citep{chen2009inference}. 
Using the latent variable interpretation, the penalized complete data log-likelihood is:
\[
p\ell_{n}^c(G) = \sum_{j=1}^{n} \sum_{k=1}^K \mathbbm{1}(z_{ij}=k)\log \{ w_k f(x_{ij};\theta_k)\} - a_n \sum_{k=1}^{K} p(\theta_k) 
\]
The only random quantity in $p\ell_{n}^c(G)$ is $\{z_{ijk}, j \in [n], k\in[K]\}$ when conditioning on $\gX_i$.
The E-step involves computing the conditional expectation of $p\ell_{n}^c(G)$, which remains valid as in~\eqref{eq:conditional_expectation}. 
The M-step maximizes $Q(G|G^{(t)})$, which now includes a penalty term:
\begin{equation*}
Q(G|G^{(t)}) = \sum_{j=1}^{n} \sum_{k=1}^K w_{ijk}^{(t)} \log \{w_k f(x_{ij};\theta_k)\} - a_n \sum_{k=1}^{K} p(\theta_k).
\end{equation*}

With the recommended penalty function, the updated mixing weights and subpopulation parameters are
\[
w_{k}^{(t+1)} = n^{-1} \sum_{j=1}^{n} w_{ijk}^{(t)}
\]
and
\begin{equation}
\label{eq:m_step_subpop}
\theta_{k}^{(t+1)} 
= \argmax_{\theta} \left\{\sum_{j=1}^{n} w_{ijk}^{(t)} \log f(x_{ij};\theta) - a_n p(\theta)\right\}.
\end{equation}

For Gaussian mixtures~\citep{chen2009inference}, the solution to~\eqref{eq:m_step_subpop} has the closed-form solution:
\[
\begin{split}
  \mu_{k}^{(t+1)}  
      &= \big \{ n w_k^{(t+1)} \big \}^{-1}  \sum_{j=1}^{n} w_{ijk}^{(t)} x_{ij}, \\
  \Sigma_{k}^{(t+1)}  
      &= \big \{ 2a_{n} + n w_{k}^{(t+1)} \big \}^{-1} \big \{ 2a_{n} S_{x} + S_k^{(t+1)} \big \},
\end{split}
\]
where 
\[
S_k^{(t+1)} = \sum_{j=1}^{n} w_{ijk}^{(t)} (x_{ij} - \mu_k^{(t+1)}) (x_{ij}-\mu_k^{(t+1)} )^{\top}.
\] 
For general location-scale mixtures, the M-step may not have a closed-form solution but only requires solving a two-variable optimisation problem, which is usually simple.

The EM algorithm for pMLE, like its MLE counterpart, increases the value of the penalized likelihood after each iteration. 
For all $t$, $\Sigma_{k}^{(t)} \geq { 2a_{n}/ (n+2a_{n})} S_{x} > 0$, ensuring the covariance matrices in $G^{(t)}$ have a lower bound that does not depend on the parameter values. 
This property guarantees that the log-likelihood under Gaussian mixture at $G^{(t)}$ has a finite upper bound. 
The above iterative procedure is guaranteed to have $p\ell_{n}^c(G^{(t)})$ converge to at least a non-degenerate local maximum.

\subsection{MM algorithm for reduction estimator}
\label{app:mm_algorithm}
In this section, we summarize the Majorisation Minimisation (MM) algorithm for minimizing~\eqref{eq:GMR} in~\citet{zhang2022distributed} for reference.
The algorithm begins with an initial $G^\dagger = \sum_{j=1}^K w^\dagger_j \{\theta_j^{\dagger}\}$ and alternates between two steps until convergence:

\begin{itemize}
  \item 
  Majorisation step: Partition the support points of $\widebar{G}$, $\{\widehat \theta_{i k}: i \in [m], k\in [K]\}$, into $K$ groups based on their closeness to $\theta_j^{\dagger}$.
  The $k$th component on the $i$th machine belongs to the $j$th cluster $\gC_{j}$ if $c(\widehat\theta_{ik},\theta_j^{\dagger})\leq c(\widehat\theta_{ik},\theta_{j'}^{\dagger})$ for any $j'$.
  
  \item Minimisation step:
  Update $\theta_j^{\dagger}$ by computing the barycentre of $\widehat \theta_{ik}$ that belongs to $j$th cluster with respect to the cost function $c(\cdot,\cdot)$.
  Specifically, $\theta_j^{\dagger} = \argmin_{\theta\in\Theta}\sum_{\{(i,k)\in \gC_{j}\}} c(\widehat\theta_{ik}, \theta)$ and assign their mixing weights to $w^\dagger_j=\sum_{\{(i,k)\in \gC_{j}\}} (n_i/N) \widehat w_{ik}$.
\end{itemize}

\begin{algorithm}
\begin{algorithmic}
\State {\bfseries Initialisation:} $f(\cdot;\theta^{\dagger}_{k})$, $k\in [K]$
\Repeat
\For {$k\in[K]$}
\State {\textbf{Majorisation:} For $\gamma\in[mK]$, let
\begin{equation*}
  \pi_{\gamma k}=
  \begin{cases}
  (n_{\gamma}/N)\widehat w_{\gamma} & \text{if}~k=\argmin_{k'} c(f(\cdot;\widehat\theta_{\gamma}),f(\cdot;\theta^{\dagger}_{k'}))\\
  0& \text{otherwise}
  \end{cases}
\end{equation*}}
\State{\textbf{Minimisation:} Let 
\begin{equation}
\label{eq:barycentre}
\theta^{\dagger}_{k}=\argmin_{\theta} \sum_{\gamma} \pi_{\gamma k} c(f(\cdot;\widehat\theta_{\gamma}), f(\cdot;\theta))
\end{equation}}
\EndFor
\Until 
the change in $\sum_{\gamma,k}\pi_{\gamma k}c(f(\cdot;\widehat\theta_{\gamma}), f(\cdot;\theta^{\dagger}_k))$ is below some threshold
\State{
Let $w^{\dagger}_{k} = \sum_{\gamma} \pi_{\gamma k}$}
\State {\bfseries Output:} $\sum_{k}w_k^{\dagger}f(\cdot;\theta^{\dagger}_{k})$
\end{algorithmic}
\caption{MM algorithm for GMR estimator under a general cost function $c(\cdot, \cdot)$.}
\label{alg:mm_reduction}
\end{algorithm}

The MM algorithm is summarized in Algorithm~\ref{alg:mm_reduction} for a general cost function $c(\cdot,\cdot)$ defined on $\gF$.
For simplicity, we order $\{\widehat\theta_{ik}\}$ arbitrarily and use a single index $\gamma$ in the algorithm.
When $\gF = \{\phi(x;\mu,\Sigma)\}$ is the space of Gaussian distributions and the cost function is the KL divergence between two Gaussian distributions:
\begin{align*}
&c(\phi(x;\widehat\mu_{\gamma},\widehat\Sigma_{\gamma}),\phi(x;\mu^{\dagger}_k,\Sigma^{\dagger}_k))\\
=&~\KL(\phi(x;\widehat\mu_{\gamma},\widehat\Sigma_{\gamma})\|\phi(x;\mu^{\dagger}_k,\Sigma^{\dagger}_k))\\
=&~-\log \phi(\widehat\mu_{\gamma};\mu^{\dagger}_k,\Sigma^{\dagger}_k) 
- \frac{1}{2}\Big \{\log\text{det}(2\pi\widehat\Sigma_{\gamma})-\mbox{tr}(\{\Sigma^{\dagger}_k\}^{-1}\widehat\Sigma_{\gamma}) + d
\Big\},
\end{align*}
the minimisation step~\eqref{eq:barycentre} has an analytical solution:
\begin{equation*}
\begin{split}
&\mu^{\dagger}_{k} =\Big\{\sum_{\gamma}\pi_{\gamma k}\Big\}^{-1}\sum_{\gamma}\pi_{\gamma k}\widehat\mu_{\gamma},\\
&\Sigma^{\dagger}_{k}= \Big\{\sum_{\gamma}\pi_{\gamma k}\Big\}^{-1}\sum_{\gamma}\pi_{\gamma k}\{\widehat\Sigma_{\gamma} + (\widehat\mu_{\gamma}-\mu^{\dagger}_{k})(\widehat\mu_{\gamma}-\mu^{\dagger}_{k})^{\top}
\}.  
\end{split}
\end{equation*}

\subsection{Numerical algorithm for TRIM}
\label{app:trim_algorithm}
In this section, we summarize the iterative algorithm for the TRIM estimator proposed in~\citet{del2019robust} for reference in Algorithm~\ref{alg:trimmed_barycentre}.
\begin{algorithm}[!htbp]
\begin{algorithmic}
\State {\bfseries Initialisation:} $\{f(\cdot;\theta^{\dagger}_{k})\}_{k=1}^{K}$, threshold $\eta \in (0,1)$
\Repeat
\For {$\gamma\in[mK]$}
\State{$c_{\gamma} = \argmin_{j} c(f(\cdot;\tilde\theta_{\gamma}), f(\cdot;\theta^{\dagger}_{j})), l_{\gamma} = c(f(\cdot;\tilde\theta_{\gamma}), f(\cdot;\theta^{\dagger}_{c_{\gamma}}))$}
\Comment{Find the cluster assignment and compute its distance to the cluster centre}
\EndFor
\State Find the permutation $\{(1),\ldots, (mK)\}$ such that $l_{(1)}\leq l_{(2)}\leq \cdots\leq l_{(mK)}$
\State{Set $\tau = \inf\{\eta \in [mK]: \sum_{\gamma=1}^{\eta}\widetilde{w}_{(\gamma)}\leq 1-\eta\}$, let
\begin{equation*}
\kappa_{\gamma} = 
\begin{cases}
\widetilde{w}_{(\gamma)} & \gamma <\tau\\
1-\eta - \sum_{\gamma < \tau} \widetilde{w}_{(\gamma)} & \gamma = \tau\\
0 & \text{otherwise}
\end{cases}
\end{equation*}}
\Comment{Trimming}
\For {$k\in[K]$}
\State{$\theta^{\dagger}_{k} = \argmin_{\theta} \sum_{\{\gamma: c_{\gamma} =k \}} \kappa_{\gamma} c(f(\cdot;\tilde\theta_{\gamma}), f(\cdot;\theta))$}
\Comment{Update cluster centres}
\State{$w^{\dagger}_{k} = (1-\eta)^{-1}\sum_{\{\gamma: c_{\gamma} =k \}} \kappa_{\gamma} \widetilde{w}_{\gamma}$}
\Comment{Update weights}
\EndFor
\Until 
there is no change in $c_{\gamma}$ for all $\gamma\in[mK]$
\State {\bfseries Output:} $\sum_{k}w_k^{\ddagger}f(\cdot;\theta^{\dagger}_{k})$
\end{algorithmic}
\caption{Iterative algorithm to compute the TRIM estimator in~\citet{del2019robust}.}
\label{alg:trimmed_barycentre}
\end{algorithm}

%% file: sections/appendix_theory.tex
\section{Theoretical results}
\label{app:theory}
\begin{figure}[!htbp]
\centering
\includegraphics[width=\textwidth]{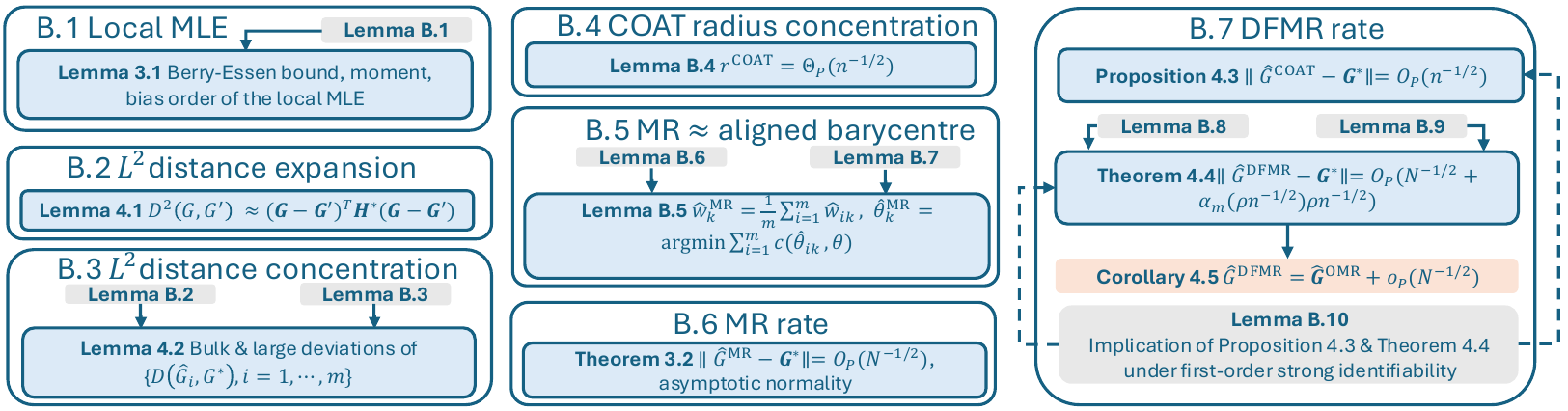}
\caption{The structure of Appendix B. The blue boxes are main results and the grey boxes are technical lemmas used for proving the main results.}
\label{fig:proof_outline}
\end{figure}
Figure~\ref{fig:proof_outline} summarizes the organisation of technical results.

\subsection{Properties of local MLE}
\label{sec:local_mle_properties}

In this section, we present the order of the moments of the local MLE, its asymptotic normality, and the Berry-Esseen bound for the local MLE.

\begin{lemma}[Unique global maxima for population log-likelihood]
\label{lemma:population-likelihood-global-maxima}
The population log-likelihood $L(\mG)$ has a unique global maximum at the true parameter value: \ie,
\[\mG^* = \argmax L(\mG).\]
\end{lemma}

\begin{proof}
Recall the population log-likelihood is
\[L(\mG) = \sE_{\mG^*}\{\ell(\mG;X)\} = \int\log f_{\mG}(x) f_{\mG^*}(x)\,dx.\]
Then we have
\[
\max_{\mG} L(\mG) = \min_{\mG} D_{KL}(f_{\mG^*} \| f_{\mG}),
\]  
where $D_{KL}(f_{\mG^*} \| f_{\mG})$ denotes the Kullback-Leibler (KL) divergence from $f_{\mG^*}$ to $f_{\mG}$.  
The KL divergence attains its minimum value of $0$ when $f_{\mG^*}(x) = f_{\mG}(x)$ for all $x$.  
Since both $\mG^*$ and $\mG$ are mixing distributions for an order $K$ mixture, the identifiability of the mixture model implies that $\mG = \mG^*$.  
This completes the proof.
\end{proof}
\begin{proof}[\textbf{Proof of Lemma~\ref{lemma:local_mle_property}}]

Statement (a) follows essentially from~\citet[Theorem 3.1]{shao2022berry}'s Berry-Esseen central limit theorem for $M$-estimators.
To translate our setting into their terminology and notation, note that the local MLE is an $M$-estimator.
The vectorized representation $\mG$ corresponds to their parameter vector, and the log-likelihood $\ell(\mG; x)$ corresponds to the function $m_{\theta}(x)$.
We check that the conditions of their theorem hold (conditions (3.6)--(3.10) in their paper).
In their proof of the lemma, conditions (3.6) and (3.7) are required to show $\sE \|\widehat{\mG} - \mG^*\|^q = O(n^{-q/2})$. 
This can alternatively be shown as in our Lemma~\ref{lemma:local_mle_property} under Assumptions~\ref{assumption:parameter-space}--\ref{assumption:smoothness}. 
Therefore, these two conditions are not necessary.
Their condition (3.8) is equivalent to the $W(x)$-Lipschitz condition in Assumption~\ref{assumption:smoothness}. 
Condition (3.9) is equivalent to the second part of our Assumption~\ref{assumption:strong-convexity}.
Finally, condition (3.10) requires finite $4$th moment of $\|\nabla \ell(\mG^*;X)\|$ and is implied by the first part of our Assumption~\ref{assumption:smoothness}, which assumes finite $8$-th moments.

Statements (b) and (c) follow immediately from Lemma 23 and (B.10) in~\citet{huang2019distributed}, who were interested in establishing the rate of convergence of a simple average of $M$-estimators.
To translate our setting into their terminology and notation, note that the local MLE is an $M$-estimator and the true parameter value $\mG^*$ is the unique global maxima as shown in Lemma~\ref{lemma:population-likelihood-global-maxima}.
Note that their Assumptions 3--6 are identical to our Assumptions~\ref{assumption:parameter-space}--\ref{assumption:smoothness}.
However, we identified a gap in their proof of Lemma 20, which is used to derive (B.10).
Specifically, their proof that the local estimates lie within a small neighbourhood of the true parameter is incorrect. 
We address this gap below in the context of finite mixture models.

We aim to show that for any $\delta > 0$, we have $\sP\paren{\|\widehat\mG-\mG^*\| < \delta} \to 1$.
To see this, observe that the vectorized representation of a mixing distribution is an element of the space
\[
\sG_K^{\text{vec}} = \Delta_{K-1} \times \underbrace{\Theta \times \cdots \times \Theta}_{K~\text{times}}.
\]
Consider any element of this space $\mG$ such that $\norm{\mG - \mG^*} \geq \delta$.
Letting $B(\mG,\delta)$ denote the open Euclidean ball of radius $\delta$ around $\mG$, Assumption \ref{assumption:identifiability} implies that there exists some $\delta_{\mG} > 0$ such that
\[
\sE \left\{ \sup_{\mG' \in B(\mG,\delta_\mG)} \log f_{\mG'}(X) \right\} < \sE\{\log f_{\mG^*}(X)\}.
\]
Next, consider the set 
\[\sG_K^{\text{vec}} \big\backslash \bigcup_{\sigma \in S_K} B(\mG_{\sigma}^*, \delta),\] 
which is obtained by removing neighbourhoods around all permutations of the vectorized representation of the true mixing distribution $G^*$.
By Assumption~\ref{assumption:parameter-space}, this set is compact and can be covered by a finite number of neighbourhoods, $B_1, B_2,\ldots, B_Q$, such that for each $q$, $B_q = B(\mG,\delta_\mG)$ for some $\mG$.
By the strong law of large numbers,
\[
\frac{1}{n}\sum_{i=1}^n \sup_{\mG \in B_q} \log f_{\mG}(X_i) \to \sE \left\{ \sup_{\mG' \in B(\mG,\delta_\mG)} \log f_{\mG'}(X) \right\}
\]
almost surely for $q=1,2,\ldots,Q$, which gives an event with probability converging to 1 on which
\[
\max_{1 \leq q \leq Q}\frac{1}{n}\sum_{i=1}^n \sup_{\mG \in B_q} \log f_{\mG}(X_i) < \frac{1}{n}\sum_{i=1}^n\log f_{\mG^*}(X_i).
\]
On this event, for any $\mG$ such that $\norm{\mG - \mG^*} \geq \delta$, we have $\mG \in B_q$ for some $q$, which implies
\[
\begin{split}
    \frac{1}{n}\sum_{i=1}^n \log f_{\mG}(X_i) & \leq \sup_{\mG' \in B_q}\frac{1}{n}\sum_{i=1}^n \log f_{\mG'}(X_i) \\
    & \leq \frac{1}{n}\sum_{i=1}^n \sup_{\mG' \in B_q}\log f_{\mG'}(X_i)  < \frac{1}{n}\sum_{i=1}^n\log f_{\mG^*}(X_i).
\end{split}
\]
Hence, $\mG$ cannot be equal to the MLE $\widehat\mG$.
\end{proof}

\subsection{$L^2$ distance between mixture densities}
\label{sec:distance_between_mixtures}

We define 
\begin{equation}
   \label{eq:l2_h_matrix}
\mH^* = \begin{pmatrix}
  \mS_{11}^* & \mS_{12}^*\\
  \mS_{12}^{*\top} & \mS_{22}^*
\end{pmatrix}, 
\end{equation}
where the blocks $\mS_{11}^*$, $\mS_{12}^*$, and $\mS_{22}^*$ are defined as follows: let $\mS_{11}^*$ be a ${(K-1)\times (K-1)}$ matrix with the $(k_1,k_2)$th element defined as 
\[
\mS_{11}^*[k_1,k_2] 
= \int \{f(x;\theta_{k_1}^*)-f(x;\theta_{K}^*)\}\{f(x;\theta_{k_2}^*)-f(x;\theta_{K}^*)\}\, dx.
\]
Let $\mS_{12}^*$ be a ${(K-1)\times pK}$ matrix 
with the $(k_1,k_2)$th block being a length $p$ vector given by 
\[
\mS_{12}^*[k_1,k_2] 
= \int w_{k_2}^* \{f(x;\theta_{k_1}^*)-f(x;\theta_{K}^*)\}\nabla_{\theta}^{\top} f(x;\theta_{k_2}^*)\,dx,
\]
and $\mS_{22}^*$ be a ${pK\times pK}$ block matrix with the $(k_1,k_2)$th block being a $p\times p$ matrix defined as
\[
\mS_{22}^*[k_1,k_2]  
= \int w_{k_1}^* w_{k_2}^* \nabla_{\theta}f(x;\theta_{k_1}^*)\nabla_{\theta}^{\top}f(x;\theta_{k_2}^*)\, dx.
\]
Assumption~\ref{assumption:parametric_family} and Cauchy-Schwarz inequality ensures that every element in $\mH^*$ is finite.

\begin{proof}[\textbf{Proof of Lemma~\ref{lemma:L2_distance_densities}}]
We first expand \( f(x; \theta_k) \) at \( \theta_k^* \) to get
\[
f(x; \theta_k) 
= f(x; \theta_k^*) + (\theta_k - \theta_k^*)^{\top} \nabla f(x; \theta_k^*) 
+ \epsilon_{f}(x; \theta_k, \theta_k^*),
\]
where the remainder term is given by
\[
\epsilon_{f}(x; \theta_k, \theta_k^*) = (\theta_k - \theta_k^*)^{\top} \int_{0}^{1} \left\{ \nabla f(x; (1-\rho)\theta_k^* + \rho\theta_k) - \nabla f(x; \theta_k^*) \right\} \, d\rho.
\]
Based on Assumption~\ref{assumption:parametric_family} (iii), when $\theta_k \in B_{\delta}(\theta_k^*)$, the remainder term satisfies
\be
\label{eq:epsilon_f_bound}
|\epsilon_{f}(x; \theta_k, \theta_k^*)| \leq V(x) \|\theta_k - \theta_k^*\|^2 \int_{0}^{1} \rho \, d\rho = \frac{V(x)}{2} \|\theta_k - \theta_k^*\|^2.
\ee
Applying this expansion to the density function of the mixture, we get
\[
\begin{split}
f_{\mG}(x) - f_{\mG^*}(x) =&~ \sum_{k=1}^{K-1} (w_k - w_k^*) \left\{ f(x; \theta_k^*) - f(x; \theta_K^*) \right\} \\
&+ \sum_{k=1}^{K} w_k^* (\theta_k - \theta_k^*)^{\top} \nabla f(x; \theta_k^*) + \sum_{k=1}^K \epsilon_{D}(x; w_k, \theta_k, \theta_k^*),
\end{split}
\]
where
\[
\epsilon_{D}(x; w_k, \theta_k, \theta_k^*) = w_k \epsilon_{f}(x; \theta_k, \theta_k^*) + (w_k - w_k^*) (\theta_k - \theta_k^*)^{\top} \nabla f(x; \theta_k^*).
\]
Based on~\eqref{eq:epsilon_f_bound} and Cauchy-Schwarz inequality, we have 
\[
\begin{split}
|\epsilon_{D}(x; w_k, \theta_k, \theta_k^*)| \leq&~ |\epsilon_{f}(x; \theta_k, \theta_k^*)| + |(w_k - w_k^*) (\theta_k - \theta_k^*)^{\top} \nabla f(x; \theta_k^*)| \\
\leq&~ \frac{V(x)}{2} \|\theta_k - \theta_k^*\|^2 + \|\nabla f(x; \theta_k^*)\||w_k - w_k^*| \|\theta_k - \theta_k^*\|. 
\end{split}
\]
Moreover, we have
\be
\label{eq:density_remainder_order}
\begin{split}
&\int \epsilon_{D}^2(x; w_k, \theta_k, \theta_k^*)\,dx\\
=&~\|\theta_k-\theta_k^*\|^4\int \frac{V^2(x)}{4}\,dx  + \|\theta_k-\theta_k^*\|^2 |w_k - w_k^*|^2\int \|\nabla f(x;\theta_k^*)\|^2 \, dx  \\
&+ |w_k - w_k^*|\|\theta_k-\theta_k^*\|^3 \int V(x)\|\nabla f(x;\theta_k^*)\|\,dx \\
\leq &~\frac{C}{4}\|\theta_k-\theta_k^*\|^4  \\
&+ C|w_k - w_k^*|^2\|\theta_k-\theta_k^*\|^2  + |w_k - w_k^*|\|\theta_k-\theta_k^*\|^3 \left\{\int V^2(x)\,dx \int \|\nabla f(x;\theta_k^*)\|^2\,dx\right\}^{1/2} \\
\leq &~\frac{9C}{4}\|\mG -\mG^*\|^4
\end{split}
\ee
where $C =\max\left\{\int V^2(x)\,dx, \int \|\nabla f(x;\theta_k^*)\|^2\,dx\right\} <\infty$ since according to Assumption~\ref{assumption:parametric_family} (ii) and (iii), both $V(x)$ and $\|\nabla f(x;\theta_k^*)\|$ are square integrable.
Hence, the expansion for $f_{\mG'}(x) - f_{\mG^*}(x)$ becomes
\[
\begin{split}
f_{\mG}(x) - f_{\mG'}(x)
=&~ \left\{ f_{\mG}(x) - f_{\mG^*}(x) \right\} - \left\{ f_{\mG'}(x) - f_{\mG^*}(x) \right\} \\
=&~ \sum_{k=1}^{K-1} (w_k - w_k') \left\{ f(x; \theta_k^*) - f(x; \theta_K^*) \right\} + \sum_{k=1}^{K} w_k^* (\theta_k - \theta_k')^{\top} \nabla f(x; \theta_k^*) \\
&+ \sum_{k=1}^K \left\{ \epsilon_{D}(x; w_k, \theta_k, \theta_k^*) + \epsilon_{D}(x; w_k', \theta_k', \theta_k^*) \right\}.
\end{split}
\]

When both $\mG, \mG' \to \mG^*$, note that in the above expansion, only the first two terms are the leading terms as they are linear in $\mG - \mG'$. 
The rest of the terms are of higher order. Hence, we get
\[
\begin{split}
&\int \left\{ 
\sum_{k=1}^{K-1} (w_k - w_k') \left\{ f(x; \theta_k^*) - f(x; \theta_K^*) \right\} + \sum_{k=1}^{K} w_k^* \nabla^{\top} f(x; \theta_k^*) (\theta_k - \theta_k') 
\right\}^2 dx \\
=&~ (\mG - \mG')^{\top} \mH^* (\mG - \mG'),
\end{split}
\]
where $\mH^*$ is defined in~\eqref{eq:l2_h_matrix}.

We now consider the remainder term. Based on~\eqref{eq:density_remainder_order}, we have
\[
\begin{split}
&~\int \left\{ \epsilon_{D}(x; w_k, \theta_k, \theta_k^*) + \epsilon_{D}(x; w_k', \theta_k', \theta_k^*) \right\}^2 dx \\
=&~\int \epsilon_{D}^2(x; w_k, \theta_k, \theta_k^*) \,dx + \int \epsilon_{D}^2(x; w_k', \theta_k', \theta_k^*) \,dx + 2\int \epsilon_{D}(x; w_k, \theta_k, \theta_k^*) \epsilon_{D}(x; w_k', \theta_k', \theta_k^*)\,dx\\
\leq &~\frac{9C}{4}\max\{\|\mG-\mG^*\|^4,\|\mG'-\mG^*\|^4\} + 2\left\{\int \epsilon_{D}^2(x; w_k, \theta_k, \theta_k^*)\,dx\int \epsilon_{D}^2(x; w_k', \theta_k', \theta_k^*)\,dx\right\}^{1/2}\\
\leq &~\frac{9C}{4}\max\{\|\mG-\mG^*\|^4,\|\mG'-\mG^*\|^4\} + \frac{9C}{2}\max\{\|\mG-\mG^*\|^4,\|\mG'-\mG^*\|^4\} \\
=&~\frac{27C}{4}\max\{\|\mG-\mG^*\|^4,\|\mG'-\mG^*\|^4\}.
\end{split}
\]
Similarly, the cross term satisfies
{\small
\begin{equation*}
\begin{split}
&\int \left\{ 
\sum_{k=1}^{K-1} (w_k - w_k^*) \left\{ f(x; \theta_k^*) - f(x; \theta_K^*) \right\} 
\right\} \sum_{k=1}^K \left\{ \epsilon_{D}(x; w_k, \theta_k, \theta_k^*) + \epsilon_{D}(x; w_k', \theta_k', \theta_k^*) \right\} \, dx \\
&+ \int \left\{ 
\sum_{k=1}^{K} w_k^* \nabla^{\top} f(x; \theta_k^*) (\theta_k - \theta_k^*) 
\right\} \sum_{k=1}^K \left\{ \epsilon_{D}(x; w_k, \theta_k, \theta_k^*) + \epsilon_{D}(x; w_k', \theta_k', \theta_k^*) \right\} \, dx \\
&= O(\max\{\|\mG - \mG^*\|^3,\|\mG'- \mG^*\|^3\}\}).
\end{split}    
\end{equation*}
}
Combining the three order assessments, we obtain the conclusion in the lemma. 
\end{proof}

\subsection{$L^2$ distance concentration}
\label{app:rCOAT_COAT}

We first define two quantities that will be used throughout this and subsequent sections of the appendix.
Let $q_n$ denote the quantile function for the distribution of $D(\widehat G, G^*)$ where $\widehat G$ is the MLE based on $n$ samples.
Let $q^*$ denote the quantile function for the generalized chi-squared distribution with matrix $\mI^{-1/2}(\mG^*)\mH^*\mI^{-1/2}(\mG^*)$, \ie 
\[q^*(t) = \inf\left\{x:\sP(Z^{\top}\mI^{-1/2}(\mG^*)\mH^*\mI^{-1/2}(\mG^*)Z\leq x)\geq t\right\}
\]
where $Z$ is a standard Gaussian random vector.

To ease notation, we extend the domain of each quantile function $q$ to the real line, with the convention that $q(a) = - \infty$ for $a < 0$ and $q(a) = \infty$ for $a > 1$.

\begin{lemma} 
\label{lemma:coat_radius_conc_part2}
Let $q_n$ and $q^*$ be as defined above.
Under Assumptions~\ref{assumption:parameter-space}--\ref{assumption:parametric_family}, when $n$ is sufficiently large, there exists a universal constant $C > 0$ such that for all $a \in [0,1]$, we have
\begin{equation}
\label{eq:quantile_CLT}
q^*\paren*{a - Cn^{-1/2}} - Cn^{-1/2}\log^{3/2} n \leq nq_n^2(a) \leq q^*\paren*{a + Cn^{-1/2}} + Cn^{-1/2}\log^{3/2} n.
\end{equation}
\end{lemma}

\begin{proof}
Recall that $q_n$ is the quantile function of $D(\widehat{G}, G^*)$.
Note that $nq_n^2 = \widehat q_n$, where $\widehat q_n$ is the quantile function of $n D^2\paren*{\widehat{G},G^*}$.
Using this relationship, we analyse the distribution of $nD^2\paren*{\widehat{G},G^*}$ directly.

Let $\mathbf{Z}$ be a standard Gaussian random vector with dimension $\nu = Kp+K-1$.
By the Berry-Esseen bound in Lemma~\ref{lemma:local_mle_property} (a), which quantifies the difference between the asymptotic distribution $\sqrt{n}\paren*{\widehat{\mG}-\mG^*}$ and a standard Gaussian random vector, we have
\[
\sP\paren*{\norm*{\sqrt{n}\paren*{\widehat{\mG}-\mG^*}}_2 \geq t} \leq \sP\paren*{\norm*{\mI^{-1/2}\mathbf{Z}}_2 \geq t} + \frac{C}{\sqrt{n}} \leq \sP\paren*{\norm*{\mathbf{Z}}_2^2 \geq \lambda_{\min}(\mI)t^2} + \frac{C}{\sqrt{n}}
\]
where we write $\mI(\mG^*)$ as $\mI$ for simplicity of notation. 
The first term on the right-hand side is the tail of a chi-squared random variable with $\nu$ degrees of freedom.
Making use of the CDF of the distribution (or alternatively, concentration of norm \citep{vershynin2018high}), we get
\[\sP\paren*{\norm*{\mathbf{Z}}_2^2 \geq \lambda_{\min}(\mI)t^2}\leq C\exp\paren*{-\lambda_{\min}(\mI)(t-\sqrt{\nu})^2/2}\]
for $t \geq \nu$.
Setting this term to be of order $O(n^{-1/2})$ and solving for $t$, we get
\begin{equation}
\label{eq:quantile_CLT_helper}
    \sP\paren*{\norm*{\sqrt{n}\paren*{\widehat{\mG}-\mG^*}}_2 \geq \paren*{\frac{\log n}{2\lambda_{\min}(\mI)}}^{1/2} + \sqrt{\nu}} \leq \frac{C}{\sqrt{n}}.
\end{equation}
Next, using Lemma~\ref{lemma:L2_distance_densities}, there exists some $r_0 > 0$ and $C > 0$ such that
\[
\abs*{D^2(G,G^*) - \paren*{\mG-\mG^*}^{\top}\mH^*\paren*{\mG-\mG^*}} \leq C\norm*{\mG-\mG^*}^3
\]
whenever $\norm*{\mG-\mG^*} \leq r_0$.
Suppose $n$ is large enough so that
\[
n^{-1/2}\paren*{\paren*{\frac{\log n}{2\lambda_{\min}(\mI)}}^{1/2} + \sqrt{\nu}} \leq r_0.
\]
On the complement of the event considered in~\eqref{eq:quantile_CLT_helper}, we then have
\[
\begin{split}
    \abs*{nD^2(\widehat G,G^*) - n\paren*{\widehat\mG-\mG^*}^{\top}\mH^*\paren*{\widehat\mG-\mG^*}} & \leq Cn^{-1/2} \norm*{n^{1/2}\paren*{\widehat{\mG} - \mG^*}}^3 \\
    & \leq Cn^{-1/2}\paren*{\paren*{\frac{\log n}{2\lambda_{\min}(\mI)}}^{1/2} + \sqrt{\nu}}^3.
\end{split}
\]
Using Lemma~\ref{lemma:local_mle_property}(a) again, we get
\[
\sP\paren*{n\paren*{\widehat\mG-\mG^*}^{\top}\mH^*\paren*{\widehat\mG-\mG^*} \leq t} \leq \sP\paren*{\mathbf{Z}^{\top}\mI^{1/2}\mH^*\mI^{-1/2}\mathbf{Z} \leq t} + \frac{C}{\sqrt{n}}.
\]
Combining these last two statements gives
\[
\begin{split}
&\sP\paren*{nD^2(\widehat G,G^*) \leq t} \\
\leq&~\sP\Bigg(n(\widehat\mG-\mG^*)^{\top}\mH^*(\widehat\mG-\mG^*) \leq t + Cn^{-1/2}\paren*{\Big(\frac{\log n}{2\lambda_{\min}(\mI)}\Big)^{1/2} + \sqrt{\nu}}^3\Bigg) + \frac{C}{n^{1/2}} \\
\leq&~\sP\paren*{\mathbf{Z}^{\top}\mI^{-1/2}\mH^*\mI^{-1/2}\mathbf{Z} \leq t + C\frac{\log^{3/2}n}{n^{1/2}}} + \frac{C}{n^{1/2}}
\end{split}
\]
which holds for any $t > 0$ and for $n$ sufficiently large.
Note that $\mathbf{Z}^{\top}\mI^{-1/2}\mH^*\mI^{-1/2}\mathbf{Z}$ is a generalized chi-squared distribution with matrix $\mI^{-1/2}\mH^*\mI^{-1/2}$.
Hence the first term on the right-hand-size is the CDF of this random variable.
Since quantile functions are the almost sure inverses of CDFs, we may reflect these functional inequalities around the $y=x$ line to obtain 
\[
\widehat q_n\paren*{a} \geq q^*(a - Cn^{-1/2}) - \frac{C\log^{3/2}n}{\sqrt{n}}
\]
for any $0 \leq a \leq 1$.

A similar calculation gives the bound
\[
\sP\paren*{nD^2(\widehat G,G^*) \leq t} \geq \sP\paren*{\mathbf{Z}^{\top}\mI^{-1/2}\mH^*\mI^{-1/2}\mathbf{Z} \leq t - C\frac{\log^{3/2}n}{\sqrt{n}}} - \frac{C}{\sqrt{n}}
\]
and 
\[
\widehat q_n(a) \leq q^*\paren*{a + Cn^{-1/2}} + \frac{C\log^{3/2}n}{n^{1/2}}. \]
Using the relationship between $\widehat q_n$ and $q_n$ completes the proof of \eqref{eq:quantile_CLT}.
\end{proof}

\textbf{Proof of Lemma~\ref{lemma:distance_concentration}}
The proof relies on Lemma~\ref{lem:quantile_conc}, which is deferred to Appendix~\ref{app:technical_lemmas} for conciseness.
Fix $0 < \epsilon < 1/2$.
Applying Lemma \ref{lem:quantile_conc} to $D(\widehat{G}_{j_1},G^*), \cdots, D(\widehat{G}_{j_{m}},G^*)$, we get an event with probability at least $1-2\exp(-2m(\epsilon/4)^2)$ on which 
\[
D(\widehat{G}_{j_{\lfloor \epsilon m \rfloor}},G^*) \geq q_n\paren*{\epsilon - \epsilon/4}.
\]
Using Lemma \ref{lemma:coat_radius_conc_part2}, we see that for $n \geq (4C/\epsilon)^2$, the right hand side is further bounded from below by
\[
n^{-1/2}q^*(\epsilon/2) - Cn^{-1}\log^{3/2} n = \Omega(n^{-1/2}).
\]
The upper bound in~\eqref{eq:distance_bulk_concentration} can be proved similarly.

Next, to prove \eqref{eq:distance_deviation_concentration}, we make use of Lemma \ref{lem:distance_conc_helper} to get
\[
    \sP\paren{D(\widehat G_i,G^*) \geq \rho n^{-1/2}} \leq \sP\paren{\norm{\widehat\mG_i-\mG^*} \geq \rho n^{-1/2}}
\]
for $i \in \sB^c$ for $n$ large enough.
We now apply Lemma \ref{lemma:local_mle_property} (c) together with Markov's inequality to get
\[
    \sP\paren{\norm{\widehat\mG_i-\mG^*} \geq \rho n^{-1/2}} \leq \frac{\sE\braces{\norm{\widehat \mG_i - \mG^*}^8}}{(\rho n^{-1/2})^8} = O(\rho^{-8}). \qedhere
\]

\begin{lemma}
\label{lem:distance_conc_helper}
    There is some $C, c > 0$ such that $D(G,G^*) \leq C\norm{\mG - \mG^*}$ whenever $\norm{\mG - \mG^*} \leq c$.
\end{lemma}

\begin{proof}
By Lemma~\ref{lemma:L2_distance_densities}, there is an $r_0 > 0$ such that for all $\norm{\mG - \mG^*} \leq r_0$, we have
\[
\abs*{D^2(G,G^*) - \paren*{\mG-\mG^*}\mH^*\paren*{\mG-\mG^*}} \leq C\norm*{\mG-\mG^*}^3
\]
for some $C > 0$.
Hence, whenever $\norm{\mG - \mG^*} \leq \min\braces*{r_0,\matrixnorm{\mH^*}/C}$, we have
\[
\begin{split}
    D^2(G,G^*) & \leq \paren*{\mG-\mG^*}\mH^*\paren*{\mG-\mG^*} + C\norm*{\mG-\mG^*}^3 \\
    & \leq \matrixnorm{\mH^*}\norm*{\mG-\mG^*}^2 + C\norm*{\mG-\mG^*}^3 \\
    & \leq 2\matrixnorm{\mH^*}\norm*{\mG-\mG^*}^2
\end{split}
\]
as we wanted.
\end{proof}
\subsection{COAT radius concentration}
\label{sec:coat_radius_conc}

\begin{lemma}[COAT radius concentration]
\label{lem:coat_radius_conc}
Under assumptions~\ref{assumption:parameter-space}--\ref{assump:cost_bregman_divergence} and assumption~\ref{assumption:parametric_family}, we have $r^{\coat} = \Theta_P(n^{-1/2})$.
\end{lemma}

\begin{proof}
Without loss of generality, let the local estimates in $\{\widehat{G}_i: i\in\sB^c\}$ such that 
\[D(\widehat{G}_{j_1},G^*)\leq \cdots \leq D(\widehat{G}_{j_{(1-\alpha)m}},G^*).\]
Since the majority of the machines are failure-free, by the definition of $r^{\coat}$, we have $r^{\coat} \leq r(\widehat{G}_{j_1})$, the radius of the smallest ball centred around $\widehat{G}_{j_1}$ containing at least $m/2$ other estimates.
Since $B(\widehat{G}_{j_1}; r(\widehat{G}_{j_1}))$ must be contained within the ball encompassing all $\widehat{G}_{j_1}, \ldots, \widehat{G}_{j_{m/2}}$, the radius $r(\widehat{G}_{j_1})$ is upper bounded by
\[
\max_{1 \leq i \leq m/2} D(\widehat G_{j_{i}}, \widehat G_{j_1})   \leq \max_{1 \leq i \leq m/2} D(\widehat{G}_{j_{i}},G^*) + D(\widehat{G}_{j_1},G^*)  \leq 2D(\widehat G_{j_{m/2}}, G^*)
\]
By \eqref{eq:distance_bulk_concentration}, this last quantity satisfies $D(\widehat G_{j_{m/2}}, \widehat G_{j_1}) = O_P(n^{-1/2})$, which concludes the upper bound for $r^{\coat}$.

We now show the lower bound.
We first introduce a new piece of notation.
For any fixed mixture distribution $G$, we let $q_{n,G}$ denote the quantile function for the distribution of $D(\widehat G, G)$ where $\widehat G$ is the MLE based on $n$ samples.

\textit{\underline{Step 1: Lower bound via quantile functions for off-centre distance.}}
We begin by lower bounding $r^{\coat}$ in terms of the maximum pairwise distance between failure-free estimates within the COAT ball:
\begin{equation}
\label{eq:coat_conc_helper}
\begin{split}
    r^{\coat} & = \frac{1}{2}\operatorname{diam}\paren*{B(\widehat G^{\coat}; r^{\coat})} \\
    & \geq \frac{1}{2}\max \braces*{D(\widetilde{G}_i,\widetilde{G}_j) \colon \widetilde{G}_i,\widetilde{G}_j \in B(\widehat G^{\coat}; r^{\coat})}\\
    & \geq \frac{1}{2}\max \braces*{D(\widehat{G}_i,\widehat{G}_j) \colon i, j \in \sB^c, \widehat{G}_i,\widehat{G}_j \in B(\widehat G^{\coat}; r^{\coat})}.
\end{split}
\end{equation}
Because it is assumed that there are at most $\alpha m$ failure machines, the number of failure-free estimates contained in $B(\widehat G^{\coat}; r^{\coat})$ is at least $(1/2-\alpha)m$.
This implies that~\eqref{eq:coat_conc_helper} can be further lower bounded by a minimax quantity:
\begin{equation}
\label{eq:radius_lower_bound_helper}
\frac{1}{2}\max_{1 \leq i \leq (1-\alpha)m, \widehat{G}_i \in B(\widehat G^{\coat}; r^{\coat})}\min_{\substack{S \subset \sB^c \\ |S| = (1/2 - \alpha)m}} \max_{j \in S} D(\widehat{G}_i,\widehat{G}_j).
\end{equation}
Note that by Lemma~\ref{lem:quantile_conc}, for each fixed $\widehat{G}_i$, there is a probability at least $1-2\exp(2(1-\alpha)m\epsilon^2)$ event such that
\[
\min_{\substack{S \subset \sB^c \\ |S| = (1/2 - \alpha)m}} \max_{j \in S} D(\widehat{G}_i,\widehat{G}_j) \geq q_{n,\widehat{G}_i}\paren*{\frac{1-2\alpha}{2-2\alpha}-\epsilon},
\]

Hence, consider all possible $\widehat{G}_i$, with probability at least $1-2(1-\alpha)m\exp(2(1-\alpha)m\epsilon^2)$,~\eqref{eq:radius_lower_bound_helper} is further lower bounded by
\[
\frac{1}{2}\max_{1 \leq i \leq (1-\alpha)m, \widehat{G}_i \in B(\widehat G^{\coat}; r^{\coat})}q_{n,\widehat{G}_i}\paren*{\frac{1-2\alpha}{2-2\alpha}-\epsilon}.
\]
Next we connect this lower bound with the quantile function $q^*$.

\textit{\underline{Step 2: Lower bound for quantile function for off-centre distance.}}
Consider a fixed mixing distribution $G$ that satisfies $\norm*{\mG - \mG^*} \leq Cn^{-1/2}$ for some fixed $C > 0$ not depending on $n$ or $m$. 
We may lower bound $q_{n,G}\paren*{\frac{1-2\alpha}{2-2\alpha}-\epsilon}$ as follows.
Let $t \geq 0$ be any threshold.
Following the argument in Proposition \ref{lemma:coat_radius_conc_part2}, there is a probability at least $1-Cn^{-1/2}$ event (with respect to the distribution of $\widehat G$) on which we have
\[
\begin{split}
&\abs*{nD^2(\widehat G,G) - n\paren*{\widehat\mG-\mG}^{\top}\mH^*\paren*{\widehat\mG-\mG}} \\
\leq&~Cn^{-1/2}\max\braces*{\norm*{\sqrt{n}\paren*{\widehat{\mG} - \mG^*}}^3, \norm*{\sqrt{n}\paren*{\mG - \mG^*}}^3} \\
\leq&~Cn^{-1/2}\paren*{\paren*{\frac{\log n}{2\lambda_{\min}(\mI)}}^{1/2} + \sqrt{\nu}}^3 + Cn^{-1/2}
\end{split}
\]
which implies that
\[
\sP\paren*{nD^2(\widehat G,G) \leq t}  \leq \sP\paren*{(\mathbf{Z}+\vv)^{\top}\mI^{-1/2}\mH^*\mI^{-1/2}(\mathbf{Z}+\vv) \leq t + C\frac{\log^{3/2}n}{n^{1/2}}} + \frac{C}{n^{1/2}},
\]
where $\vv = n^{1/2} \mI^{1/2}\paren*{\mG^*-\mG}$.
Next, applying Lemma \ref{lem:convex_function_CDF}, we see that the right hand side is upper bounded by
\[
\sP\paren*{\mathbf{Z}^{\top}\mI^{-1/2}\mH^*\mI^{-1/2}\mathbf{Z} \leq t + C\frac{\log^{3/2}n}{n^{1/2}}} + \frac{C}{n^{1/2}}.
\]
Since this holds for any $t$, we get
\[
nq_{n,G}^2\paren*{a} \geq q^*\paren*{a - Cn^{-1/2}} - C\frac{\log^{3/2}n}{n^{1/2}}
\]
for any $0 \leq a \leq 1$.

\textit{\underline{Step 3: Upper bound on parameter distance.}}
We wish to show that with high probability there exists $\widehat{G}_i \in B(\widehat G^{\coat}; r^{\coat})$, $i \in \sB^c$ such that $\norm*{\widehat\mG_i - \mG^*} \leq Cn^{-1/2}$, which is needed for the lower bound in Step 2 is valid.
Based on the lower bound $\abs{\braces{ i \in \sB^c \colon \widehat G_i \in B(\widehat G^{\coat}, r^{\coat})}} \geq (1/2-\alpha)m$ discussed in step 1, it suffices to show that
\[
\abs{\{ \widehat G_i \colon i \in \sB^c, \norm{\widehat \mG_i - \mG^*} > Cn^{-1/2} \}} < (1/2-\alpha)m.
\]
To begin, note that for any $t \geq 0$, Markov's inequality together with the moment bound from Lemma~\ref{lemma:local_mle_property} gives
\[
\sP\paren*{\norm*{\widehat\mG_i - \mG^*} \geq t}  \leq t^{-8}\sE\sqbracket*{\norm*{\widehat\mG_i - \mG^*}^8} \leq \frac{C}{(n^{1/2}t)^8}
\]
for some $C > 0$.
Choosing $t^* = n^{-1/2}\paren*{\frac{C(1-2\alpha)}{4-4\alpha}}^{1/q}$, we thus get
\[
\sP\paren*{\norm*{\widehat\mG_i - \mG^*} \geq t^*} \leq \frac{1-2\alpha}{4-4\alpha}.
\]
Since
\[
\abs{\{ \widehat G_i \colon i \in \sB^c, \norm{\widehat \mG_i - \mG^*} > t^* \}} = \sum_{i \in \sB^c} \mathbbm{1}\paren*{\norm*{\widetilde\mG_i - \mG^*} \geq t^*},
\]
it is a sum of $(1-\alpha)m$ independent Bernoulli random variables each with probability at most $\frac{1-2\alpha}{4(1-\alpha)}$.
Applying Chernoff's inequality \citep[Theorem 2.3.1]{vershynin2018high}, we therefore get
\[
\begin{split}
& \sP\paren*{\text{There exists}~\widehat{G}_i \in B(\widehat G^{\coat}; r^{\coat}), i \in \sB^c~\text{such that}~\norm*{\widehat\mG_i - \mG^*} \leq t^*} \\
\geq &~1 - 
\sP\paren*{\abs{\{ \widehat G_i \colon i \in \sB^c, \norm{\widehat \mG_i - \mG^*} > t^* \}} \geq (1/2-\alpha)m} \geq 1 - \exp\paren*{-\frac{(1-2\alpha)m}{12}}.
\end{split}
\]

\textit{\underline{Step 4: Putting everything together.}}
Condition on the high probability events guaranteed by Step 1 and Step 3.
We may then apply Step 2 (which is deterministic) to get the result we want.
\end{proof}

\subsection{Equivalence between MR and barycentre estimators}
\label{app:mr_barycentre_equivalence}
\begin{lemma}
\label{lem:barycentre_equivalence}
Consider the mixture reduction aggregate $\widehat G^{\mr}$ in~\eqref{eq:GMR} obtained from $m$ fixed mixing distributions $G_1,G_2,\ldots, G_m$, with $G_i = \sum_{k=1}^K w_{ik}\delta_{\theta_{ik}}$ for $i=1,2,\ldots,m$.
There exists $r_0 > 0$ such whenever $D(G_i,G^*) < r_0$ for $i=1,2,\ldots,m$, the mixing weights and subpopulation parameters of $\widehat G^{\mr}$ can be written respectively as
\begin{equation}
\label{eq:barycentre_formulas}
w_k^{\dagger} 
= 
\frac{1}{m}\sum_{i=1}^m {w}_{ik},\quad\quad
\theta_k^{\dagger} 
= 
\argmin_{\theta} \sum_{i=1}^m {w}_{ik} c(\theta_{ik},\theta),
\end{equation}
for $k=1,2,\ldots,K$.
\end{lemma}

\begin{remark}
The conclusion of Lemma~\ref{lem:barycentre_equivalence} states that minimizing the composite transportation distance defined in~\eqref{eq:CTD_def} fully aligns all mixture components of the distributions $G_1,G_2,\ldots,G_m$ (c.f.~\eqref{eq:aligned_component_def}). 
Furthermore, the value for each weight in the mixture reduction aggregate is equal to the simple average of the aligned component weights in $G_1,G_2,\ldots,G_m$, while each subpopulation parameter in the aggregate is equal to the barycentre of the aligned subpopulation parameters.
\end{remark}

\begin{proof}
    \textit{\underline{Step 1: Upper bound for composition transportation divergence.}}
    Suppose $\norm{\mG_i - \mG^*} \leq \epsilon'$ for $i=1,2,\ldots,m$ for some $0 < \epsilon' < \epsilon/2$, where $\epsilon$ is the neighbourhood radius in Assumption \ref{assump:cost_bregman_divergence}.
    We aim to obtain an upper bound for the composite transportation divergence~\eqref{eq:CTD_def} between the MR estimate $\widehat G^{MR}$ and the simple average $\widebar{G} = m^{-1}\sum_{i=1}^m G_i$.

    The convexity of the Kantorovich minimisation problem in optimal transport~\citep[Section 7.4]{villani2003topics} allows us to write
    \begin{equation}
    \label{eq:CTD_convexity}
        T_{c}\left(\widebar G, \widehat{G}^{\mr}\right) \leq T_{c}\left(\widebar G, G^{*}\right)\leq \frac{1}{m}\sum_{i=1}^m  T_{c}(G_i, G^{*}).
    \end{equation}
    We then bound each term on the right hand side.
    Recall the composite transportation divergence is the following minimisation problem
    \[
    T_{c}(G_i, G^{*}) = \min\left\{\sum_{jk} \pi_{jk} c(\theta_{ij}, \theta_k^*): \pi_{jk}\geq 0,~\sum_{j} \pi_{jk} = w_k^*,~\sum_{j} \pi_{jk} = w_{ik}\right\}.
    \]
    Consider a specific $\pi$ that satisfies the marginal constraints with diagonal elements being $\pi_{kk} = \min\{w_{ik}, w_k^*\}$, then we must have     
    \[
        \begin{split}
        T_{c}(G_i, G^{*}) & \leq \sum_{k=1}^{K} \min\{{w}_{ik}, w^{*}_k\}c(\theta_{ik}, \theta^{*}_k) + \sup_{\theta,\theta'\in \Theta}c(\theta, \theta')\|\vw_{i} - \vw^{*}\|_1\\
        & \leq \eta_+\sum_{k=1}^{K}\|{\theta}_{ik}- \theta^{*}_k\|_2^2 + \text{Diam}^2(\Theta)K^{1/2}\|{\vw}_{i} -\vw^{*}\|_2,
        \end{split}
    \]
    where $\vw_i = (w_{i1},w_{i2},\ldots,w_{iK})$.
    Note that for the second inequality, we have used Lemma \ref{lem:Bregman_upper_lower}(a).
    
    Since
    \[
        \norm{\mG_i-\mG^*}^2 = \sum_{k=1}^{K}\|{\theta}_{ik}- \theta^{*}_k\|_2^2 + \|{\vw}_{i} -\vw^{*}\|_2^2,
    \]
    we get 
    \[
        T_{c}(G_i, G^{*}) \leq \eta_+\norm{\mG_i-\mG^*}^2 + \text{Diam}^2(\Theta)K^{1/2}\norm{\mG_i-\mG^*}
    \]
    Plugging this back into \eqref{eq:CTD_convexity} and using our assumption on $\norm{\mG_i-\mG^*}$ allows us to conclude that
    \[
    T_{c}\left(\widebar G, \widehat{G}^{\mr}\right) \leq \epsilon'\paren*{\text{Diam}^2(\Theta)K^{1/2} + \eta_+\epsilon'}.
    \]

    \textit{\underline{Step 2: MR aggregate is close to true parameter values.}}
    Denote the mixing weights and subpopulation parameters of $\widehat G^{\mr}$ as $w_k^{\mr}$ and $\theta_k^{\mr}$, $k=1,2,\ldots,K$.
    Suppose $\epsilon' < \min_{k} w_k^*/2$.
    Since $\|\mG_i -\mG^*\| < \epsilon'$, we have $\|\theta_{ik} - \theta_k^*\| < \epsilon'$ for all $i$ and $k$.
    Denote
    \[
        \epsilon'' = \max_{1 \leq k \leq K}\min_{1 \leq l \leq K}\norm{\theta_l^{\mr} - \theta_{k}^*}
    \]
    and let $k_0$ be the index achieving the outer maximum.
    We show below that $\epsilon'' \to 0$ as $\epsilon' \to 0$.

    Suppose for now that $\epsilon'' > 2\epsilon'$.
    Using the reverse triangle inequality, this implies that for $i=1,2,\ldots,M$, we have
    \[
        \min_{1 \leq l \leq K}\norm{\theta_l^{\mr} - \theta_{ik_0}}  \geq \min_{1 \leq l \leq K}\norm{\theta_l^{\mr} - \theta^*_{k_0}} - \norm{\theta_{k_0}^* - \theta_{ik_0}}  \geq \epsilon'' - \epsilon'  \geq \epsilon''/2.
    \]
    Together with Lemma~\ref{lem:Bregman_upper_lower} (b), this implies
    \[
        \min_{1 \leq l \leq K}c(\theta_{ik_0},\theta_l^{\mr}) \geq \frac{\eta_{-}}{4}\min\braces{\epsilon^2, (\epsilon'')^2}.
    \]
    We now plug this into the expression for the optimal transport divergence to get
    \[
        \begin{split}
            T_{c}\left(\widebar G, \widehat{G}^{\mr}\right) & = \min\left\{\sum_{i,k,j}\pi_{ikj}c(\theta_{ik}, \theta_{j}^{\mr}): \sum_{i,k}\pi_{ikj} = w_j^{\mr}, \sum_{j}\pi_{ikj} = \frac{{w}_{ik}}{m}\right\}\\
            & \geq \min\left\{\sum_{i,j}\pi_{ik_0j}c(\theta_{ik_0}, \theta_{j}^{\mr}):  \sum_{j}\pi_{ik_0j} = \frac{{w}_{ik_0}}{m}\right\}\\
            & \geq \frac{1}{m}\sum_{i=1}^m w_{ik_0}\frac{\eta_{-}}{4} \cdot \min\braces{\epsilon^2, (\epsilon'')^2}.
        \end{split}
    \]
    Since
    \[
        \max_{1 \leq i \leq m}\abs{w_{ik_0} - w_{k_0}^*} \leq \max_{1 \leq i \leq m}\norm{\mG_i - \mG^*} \leq \epsilon',
    \]
    our assumption $\epsilon' < w_{k_0}^*/2$ implies
    \[
            \frac{1}{m}\sum_{i=1}^m w_{ik_0} \geq w_{k_0}^* - \frac{1}{m}\sum_{i=1}^m \abs{w_{ik_0} - w_{k_0}^*} \geq w_{k_0}^*/2,
    \]
    which yields
    \[
    T_{c}\left(\widebar G, \widehat{G}^{\mr}\right) \geq \frac{\eta_-}{8}\min_{1 \leq k \leq K} w_k^*\cdot \min\braces{\epsilon^2, (\epsilon'')^2}.
    \] 
    Combining this with Step 2, we get
    \[
            \min\braces{\epsilon^2,(\epsilon'')^2} \leq \frac{8}{\eta_-\min_{1\leq k \leq K}w_k^*}\epsilon'\paren*{\text{Diam}^2(\Theta)K^{1/2} + \eta_+\epsilon'},
    \]
    which becomes
    \[
    (\epsilon'')^2 \leq \frac{8}{\eta_-\min_{1\leq k \leq K}w_k^*}\epsilon'\paren*{\text{Diam}^2(\Theta)K^{1/2} + \eta_+\epsilon'}
    \]
    when $\epsilon'$ is small enough.
    Combine with the case that $\epsilon'' \leq 2\epsilon'$, we conclude that
    \[
        \epsilon'' \leq \max\braces*{2\epsilon', \paren*{\frac{8}{\eta_-\min_{1\leq k \leq K}w_k^*}\epsilon'\paren*{\text{Diam}^2(\Theta)K^{1/2} + \eta_+\epsilon'}}^{1/2}}.
    \]    
    In other words, $\epsilon'' \to 0$ as $\epsilon' \to 0$.
    
    \textit{\underline{Step 3: Completing the proof.}}
    Let
    \[
    R = \max_{1 \leq k < l \leq K}\norm{\theta_k^* - \theta_l^*}.
    \]
    First pick $\epsilon' > 0$ small enough such that $\epsilon', \epsilon'' < \min\braces{R/4,\epsilon/2, w_k^*/2}$, so that in particular, the assumptions and conclusions of Steps 2 and 3 hold.
    Using the indexing strategy \eqref{eq:aligned_component_def} for the components of $\widehat G^{\mr}$, we then have
    \[
        \norm{\theta_k^{\mr} - \theta_k^*} < \epsilon'', \quad\quad \max_{l \neq k}\norm{\theta_k^{\mr} - \theta_l^*} > 3R/4
    \]
    for $k=1,2,\ldots,K$.
    Furthermore,
    \begin{equation}
    \label{eq:barycentre_step4_helper}
    \begin{split}
        \norm{\theta_k^{\mr} - \theta_{ik}} & \leq \norm{\theta_k^{\mr} - \theta_{k}^*} + \norm{\theta_k^* - \theta_{ik}} \\
        & \leq \norm{\theta_k^{\mr} - \theta_{k}^*} + \norm{\mG^* - \mG_i}  < \epsilon'' + \epsilon'
    \end{split}
    \end{equation}
    for $i=1,2,\ldots,m$.
    Similarly, we have
    \begin{equation}
    \label{eq:barycentre_step4_helper2}
        \max_{l \neq k}\norm{\theta_k^{\mr} - \theta_{il}} > R/2
    \end{equation}
    for $i=1,2,\ldots,m$.
    Combining \eqref{eq:barycentre_step4_helper} with Lemma \ref{lem:Bregman_upper_lower}(a), we get
    \be
    \label{eq:same_k_upper_bound}
            c(\theta_{ik},\theta_k^{\mr}) \leq \eta_+(\epsilon' + \epsilon'')^2
    \ee
    Combining \eqref{eq:barycentre_step4_helper2} with Lemma \ref{lem:Bregman_upper_lower}(b), we get
    \[
            \max_{l \neq k}c(\theta_{il},\theta_k^{\mr}) \geq \frac{\eta_-}{4}\min\braces{R^2,\epsilon^2}.
    \]
    Choosing $\epsilon'$ to be smaller if necessary, we get
    \be
    \label{eq:different_k_upper_bound}
            c(\theta_{ik},\theta_k^{\mr}) < \max_{l \neq k}c(\theta_{il},\theta_k^{\mr}).
    \ee
    Recall that 
    \[
    \begin{split}
    T_{c}\left(\widebar G, \widehat{G}^{\mr}\right) & = \min_{G\in \sG_{K}} T_{c}\left(\widebar G, G\right)\\
    &=\min_{G\in \sG_{K}} \min\left\{\sum_{i,k,j}\pi_{ikj}c(\theta_{ik}, \theta_{j}): \sum_{i,k}\pi_{ikj} = w_j, \sum_{j}\pi_{ikj} = \frac{{w}_{ik}}{m}\right\}\\
    &=\min_{\theta_j} \min\left\{\sum_{i,k,j}\pi_{ikj}c(\theta_{ik}, \theta_{j}): \sum_{j}\pi_{ikj} = \frac{{w}_{ik}}{m}\right\}
    \end{split}
    \]
    with 
    \[
        w^{\mr}_k = \sum_{i,k} \pi^*_{ikj} = \frac{1}{m} \sum_{i=1}^m w_{ik}.
    \]
    where $\pi^*_{ikj} = w_{ik}\delta_{jk}/m$ and $\delta_{jk}=\mathbbm{1}(j=k)$ is the Dirac delta function is the solution to the inner minimisation problem based on~\eqref{eq:same_k_upper_bound} and~\eqref{eq:different_k_upper_bound}.
    
    This also means that $\theta^{\mr}_k$ is the solution to
    \[
            \argmin_{\theta}\sum_{i=1}^m w_{ik}c(\theta_{ik},\theta).
    \]
    These statements hold for $k=1,2,\ldots,K$, giving the desired conclusion.
    Finally, by Lemma \ref{lem:continuity_of_parameter_distance}, there exists $r_0$ small enough so when $D(G_i,G^*) < r_0$ for $i=1,2,\ldots,m$, $\epsilon'$ is small enough for the calculations in this step to hold.
\end{proof}

\begin{lemma}
\label{lem:Bregman_upper_lower}
Consider a reverse Bregman divergence
\[
c(\theta',\theta) = A(\theta) - A(\theta') - (\theta - \theta')^T\nabla A(\theta')
\]
computed over $\theta, \theta' \in \R^d$.
Suppose there exists $\theta^* \in \R^d$ and $\eta_+ \geq \eta_- > 0$ such that $\eta_- \mI \leq \nabla^2 A(\theta) \leq \eta_+\mI$ for all $\theta$ satisfying $\norm{\theta - \theta^*} \leq \epsilon$.
Then
\begin{enumerate}[label=(\alph*)]
\item For all $\theta, \theta'$ satisfying $\norm{\theta - \theta^*}, \norm{\theta' - \theta^*} \leq \epsilon$, we have
\begin{equation}
\label{eq:bregman_upper}
    c(\theta',\theta) \leq \eta_+ \norm{\theta'-\theta}^2.
\end{equation}
\item For all $\theta'$ satisfying $\norm{\theta' - \theta^*} \leq \epsilon/2$ and for all $\theta \in \R^d$, we have
\begin{equation}
\label{eq:bregman_lower}
    c(\theta',\theta) \geq \eta_- \min\braces{\norm{\theta'-\theta}^2, \epsilon^2/4}.
\end{equation}
\end{enumerate}
\end{lemma}

\begin{proof}
Using the Taylor expansion
\begin{equation}
\label{eq:taylor_expansion_bregman}
    A(\theta) = A(\theta') + \langle \nabla A(\theta'), \theta-\theta'\rangle + \int_{0}^{1}(\theta-\theta')^{\top}\nabla^2 A(\rho\theta+(1-\rho)\theta')(\theta-\theta')\,d\rho,
\end{equation}
we may write
\begin{equation}
\label{eq:bregman_cost_integral}
    c(\theta',\theta) = (\theta-\theta')^{\top} M(\theta) (\theta-\theta'),
\end{equation}
where
\[
    M(\theta) = \int_{0}^{1}\nabla^2 A(\rho\theta+(1-\rho)\theta')\,d\rho.
\]
Under the assumptions of (a), $\rho\theta^*+(1-\rho)\theta' \in B_\epsilon(\theta^*)$ for any $\rho \in [0,1]$, so that Assumption~\ref{assump:cost_bregman_divergence} implies $M(\theta) \preceq \eta_+\mI$.
Plugging this into~\eqref{eq:bregman_cost_integral} gives~\eqref{eq:bregman_upper}.

To establish~\eqref{eq:bregman_lower}, first consider $\theta', \theta$ such that $\norm{\theta'-\theta^*} \leq \epsilon/2$ and $\norm{\theta-\theta^*} \leq \epsilon$.
Then Assumption~\ref{assump:cost_bregman_divergence} implies $M(\theta) \succeq \eta_{-}\mI$ in~\eqref{eq:bregman_cost_integral} and we get
\[
c(\theta',\theta) \geq \eta_-\norm{\theta'-\theta}^2.
\]
Next, suppose $\norm{\theta-\theta^*} > \epsilon$.
Then in particular, $\norm{\theta-\theta'} > \epsilon/2$.
Set
\[
\theta'' = \theta' + (\epsilon/2)\frac{\theta - \theta'}{\norm{\theta - \theta'}}.
\]
We see that
\[
    \norm{\theta'' - \theta'} \leq \epsilon/2,
\]
which implies $\norm{\theta'' - \theta^*} \leq \epsilon$, giving
\[
    c(\theta',\theta'') \geq \eta_-\norm{\theta'-\theta''}^2 \geq \eta_- \epsilon^2/4.
\]
Meanwhile, since $\theta''$ is a convex combination of $\theta$ and $\theta'$, the convexity of $A$ implies that 
\[
c(\theta',\theta) \geq \frac{\|\theta-\theta'\|^2}{\epsilon/2} c(\theta',\theta'') > c(\theta',\theta''). \qedhere
\]
\end{proof}

\begin{lemma}
\label{lem:continuity_of_parameter_distance}
    For any $\epsilon > 0$, there exists $\delta > 0$ such that whenever $D(G,G^*) < \delta$, we have $\norm{\mG-\mG^*} < \epsilon$.
\end{lemma}

\begin{proof}
    Fix $\epsilon > 0$, and consider the set
    \[
        E = \sG_K^{\text{vec}} \big\backslash \bigcup_{\sigma \in S_K} B(\mG_{\sigma}^*, \epsilon),
    \]
    which was also defined in the proof of Lemma~\ref{lemma:local_mle_property} (see Section~\ref{sec:local_mle_properties}).
    By Assumption~\ref{assumption:parameter-space}, this set is compact.

    Consider the function $\mG \mapsto D(G,G^*)$.
    This function is continuous and by Assumption \ref{assumption:identifiability}, its set of minimizers comprise $\braces*{ \mG^*_{\sigma} \colon \sigma \in S_K}$.
    As such, $\inf_{\mG \in E} D(G,G^*)$ is achieved by an element in $A$ and we have
    \[
        \inf_{\mG \in E} D(G,G^*) > 0.
    \]
    Setting $\delta$ to be this value, we therefore see that if $D(G,G^*) < \delta$, we need $G \notin E$, which implies that $\norm{\mG - \mG^*} < \epsilon$.
\end{proof}

\subsection{Rate of convergence of MR}
\label{app:mr_rate_proof}
\begin{proof}[\textbf{Proof of Theorem \ref{thm:oracle_rate}}]
\textit{\underline{Step 1: Equivalence to barycentre estimator.}}
We first observe that
\[
\sP\paren*{\max_{1 \leq i \leq m} D(\widehat G_i,G^*) \geq m^{1/4}n^{-1/2}} \leq m \sP\paren{D(\widehat G_i,G^*) \geq m^{1/4}n^{-1/2}} = O(m^{-1}),
\]
where the second inequality follows from Lemma~\ref{lemma:distance_concentration} (b).
In particular, since $m = O(n)$, we have
\begin{equation}
\label{eq:omr_helper}
    \max_{1 \leq i \leq m} D(\widehat G_i,G^*) = o_P(1).
\end{equation}
This means that Lemma \ref{lem:barycentre_equivalence} can be applied for $n$ and $m$ large enough, which shows that that the mixing weights and subpopulation parameters of $\widehat G^{\mr}$ can be written respectively as
\[
w_k^{\dagger} 
= 
\frac{1}{m}\sum_{i=1}^m \widehat{w}_{ik},\quad\quad
\theta_k^{\dagger} 
= 
\argmin_{\theta} \sum_{i=1}^m \widehat{w}_{ik} c(\widehat{\theta}_{ik},\theta),
\]
for $k=1,2,\ldots,K$.

\textit{\underline{Step 2: Convergence of mixing weight estimates.}}
Focusing on some fixed $k$, we next decompose the MSE of each mixing weight estimate into bias and variance terms:
\[
\sE\braces*{\paren{w_k^\dagger - w_k^*}^2}  = \paren*{\sE\braces{w_k^\dagger} - w_k^*}^2 + \Var\paren{w_k^\dagger} = \paren*{\sE\braces{\widehat w_{ik}} - w_k^*}^2 + \frac{1}{m}\Var\paren{\widehat w_{ik}}.
\]
Using Lemma~\ref{lemma:local_mle_property} (b), the bias term satisfies
\[
\paren*{\sE\braces{\widehat w_{ik}} - w_k^*}^2 = O(n^{-2}) = O(N^{-1}).
\]
Meanwhile, using Lemma~\ref{lemma:local_mle_property} (c), the variance term satisfies
\[
m^{-1}\Var\paren{\widehat w_{ik}} \leq m^{-1}\sE\braces*{\paren{\widehat w_{ik}-w_k^*}^2} \leq m^{-1}\sE\braces*{\norm{\widehat\mG - \mG^*}^2} = O(N^{-1}).
\]
By Markov's inequality, we thus have
\[
w_k^\dagger - w_k^* = O_P(N^{-1/2}),
\]
which implies also that
\be
\label{eq:barycentre_weight_helper}
\frac{1}{w_k^\dagger} = \frac{1}{w_k^* + (w_k^\dagger - w_k^*)} = \frac{1}{w_k^*} + O_P(N^{-1/2}).
\ee

\textit{\underline{Step 3: Asymptotic expansion for subpopulation parameter estimates.}}
Denote  
\[
h_k(\theta) = \sum_{i=1}^m \widehat{w}_{ik} c(\widehat{\theta}_{ik}, \theta) 
= \sum_{i=1}^m \widehat{w}_{ik} \big[A(\theta) - A(\widehat{\theta}_{ik}) - (\theta - \widehat{\theta}_{ik})^\top \nabla A(\widehat{\theta}_{ik})\big].
\]  
Since $\theta_k^\dagger$ is the minimizer of $h_k(\theta)$ it satisfies the zero gradient condition $\nabla h_k(\theta_k^\dagger) = 0$.
We may expand
\[
\begin{split}
    \nabla h_k(\theta_k^\dagger) & = \sum_{i=1}^m \widehat{w}_{ik} \paren*{ \nabla A(\theta_k^\dagger) - \nabla A(\widehat{\theta}_{ik})} \\
    & = \sum_{i=1}^m \widehat{w}_{ik} \paren*{\nabla A(\theta_k^\dagger) - \nabla A(\theta_k^*)} - \sum_{i=1}^m \widehat{w}_{ik} \paren*{\nabla A(\widehat\theta_{ik}) - \nabla A(\theta_k^*)}
\end{split}
\]
Setting this equal to zero and rearranging, we get
\begin{equation}
\label{eq:omr_helper2}
    \nabla A(\theta_k^\dagger) - \nabla A(\theta_k^*) = \paren*{m w_k^\dagger}^{-1}\sum_{i=1}^m \widehat{w}_{ik} \paren*{\nabla A(\widehat\theta_{ik}) - \nabla A(\theta_k^*)}.
\end{equation}
Using~\eqref{eq:omr_helper} and Lemma~\ref{lem:continuity_of_parameter_distance}, we see that 
\begin{equation}
\label{eq:omr_helper3}
    \sP\paren*{\max_{1\leq i \leq m}\norm*{\widehat\theta_{ik} - \theta_k^*} \geq \epsilon} = o_P(1).
\end{equation}
On the complement of this event, we may show using Assumption~\ref{assump:cost_bregman_divergence} and by differentiating~\eqref{eq:taylor_expansion_bregman} that
\[
    \norm*{\nabla A(\widehat\theta_{ik}) - \nabla A(\theta_k^*)} \leq \eta_+\norm*{\widehat\theta_{ik} - \theta_k^*}
\]
for $i=1,2,\ldots,m$.
In particular, the right hand side of~\eqref{eq:omr_helper2} is $o_P(1)$.
Note also that by Assumption~\ref{assump:cost_bregman_divergence}, the function $\theta \mapsto \nabla A(\theta)$ has a non-zero Jacobian at $\theta_k^*$ and hence has a differentiable inverse in a neighbourhood of that point.
Therefore, based on~\eqref{eq:omr_helper2}, there is an event with probability converging to 1 on which we may write
\begin{equation}
\label{eq:omr_helper4}
    \theta_k^\dagger =  \paren{\nabla A}^{-1}\paren*{\nabla A(\theta_k^*) + \paren*{mw_k^\dagger}^{-1}\sum_{i=1}^m \widehat{w}_{ik} \paren*{\nabla A(\widehat\theta_{ik}) - \nabla A(\theta_k^*)}}
\end{equation}
Applying the definition of differentiability, we may then rewrite the right hand side in \eqref{eq:omr_helper4} as
\begin{equation}
\label{eq:omr_helper5}
\begin{split}
    & \theta_k^* + \paren*{mw_k^\dagger}^{-1} \nabla^2 A(\theta^*_k)^{-1}\sum_{i=1}^m \widehat{w}_{ik} \paren*{\nabla A(\widehat\theta_{ik}) - \nabla A(\theta_k^*)} + o(\norm{\theta_k^\dagger - \theta_k^*}) \\
    & = \theta_k^* + \paren*{mw_k^*}^{-1} \nabla^2 A(\theta^*_k)^{-1}\sum_{i=1}^m \widehat{w}_{ik} \paren*{\nabla A(\widehat\theta_{ik}) - \nabla A(\theta_k^*)} + o(\norm{\theta_k^\dagger - \theta_k^*}),
\end{split}
\end{equation}
where the equality above follows from~\eqref{eq:barycentre_weight_helper}.

Next, for $i=1,2,\ldots, m$, we expand
\[
\begin{split}
    \nabla A(\widehat\theta_{ik}) - \nabla A(\theta_k^*) & = \int_{0}^{1}\nabla^2 A(\rho\theta_k^* + (1-\rho)\widehat{\theta}_{ik})\,d\rho(\widehat{\theta}_{ik}-\theta_k^*) \\
    & = \nabla^2 A(\theta_k^*)\paren{\widehat \theta_{ik}-\theta_k^*} + U_{i}
\end{split}
\]
where for the ease of notation, we denote 
\[
U_{i} = \int_{0}^{1}\braces{\nabla^2 A(\rho\theta_k^* + (1-\rho)\widehat{\theta}_{ik})-\nabla^2 A(\theta_k^*)}\,d\rho (\widehat{\theta}_{ik}-\theta_k^*).
\]
This allows us to get
\[
\begin{split}
    & \widehat{w}_{ik} \paren*{\nabla A(\widehat\theta_{ik}) - \nabla A(\theta_k^*)} \\
    & = w_k^* \paren*{\nabla A(\widehat{\theta}_{ik}) - \nabla A(\theta^*_k)} + \paren{\widehat{w}_{ik} - w_k^*}\paren*{\nabla A(\widehat{\theta}_{ik}) - \nabla A(\theta^*_k)} \\
    & = w_k^*\nabla^2 A(\theta_k^*)\paren{\widehat \theta_{ik}-\theta_k^*} + \paren{\widehat{w}_{ik} - w_k^*}\nabla^2 A(\theta_k^*)\paren{\widehat \theta_{ik}-\theta_k^*} + w_k^*U_{i} + \paren{\widehat{w}_{ik} - w_k^*}U_{i}
\end{split}
\]

Plugging these into~\eqref{eq:omr_helper5} gives
\be
\label{eq:barycentre_expansion}
\theta_{k}^{\dagger} -\theta_k^* = \frac{1}{m}\sum_{i=1}^{m}(\widehat\theta_{ik}-\sE\braces{\widehat\theta_{ik}}) + (\sE\braces{\widehat\theta_{ik}} - \theta_k^*) + T_1+ T_2 + T_3 + o(\|\theta_{k}^{\dagger} - \theta_k^*\|),
\ee
where 
\[
\begin{split}
T_1 & = \paren*{m w_k^*}^{-1} \sum_{i=1}^m \paren{\widehat{w}_{ik}-w_k^*}\paren{\widehat\theta_{ik}-\theta_k^*}, \\
T_2 & = \paren*{m \nabla^2 A(\theta_k^*)}^{-1}\sum_{i=1}^{m} U_{i},\\
T_3 & = \paren*{m w_k^*\nabla^2 A(\theta_k^*)}^{-1}\sum_{i=1}^{m}(\widehat{w}_{ik}-w_k^*)U_{i}.
\end{split}
\]

\textit{\underline{Step 4: Bounding $T_1$, $T_2$, and $T_3$.}}
Using Cauchy-Schwarz and Lemma \ref{lemma:local_mle_property} (c), we get
\[
\sE\braces*{\norm*{\paren{\widehat{w}_{ik}-w_k^*}\paren{\widehat\theta_{ik}-\theta_k^*}}^2} \leq \sE\braces*{\paren{\widehat{w}_{ik}-w_k^*}^4}^{1/2}\sE\braces*{\norm*{\widehat\theta_{ik}-\theta_k^*}^4}^{1/2} = O(n^{-2}).
\]
Similarly, using Lemma~\ref{lemma:local_mle_property} (b), we get
\[
\begin{split}
\norm*{\sE\braces*{\paren{\widehat{w}_{ik}-w_k^*}\paren{\widehat\theta_{ik}-\theta_k^*}}} & \leq
    \sE\braces*{\norm*{\paren{\widehat{w}_{ik}-w_k^*}\paren{\widehat\theta_{ik}-\theta_k^*}}} \\
    & \leq \sE\braces*{\paren{\widehat{w}_{ik}-w_k^*}^2}^{1/2}\sE\braces*{\norm*{\widehat\theta_{ik}-\theta_k^*}^2}^{1/2} \\
    & = O(n^{-1}).
\end{split}
\]
Decomposing $T_1$ into bias and variance components and arguing similar to Step 2, we get $\sE\braces{\norm{T_1}} = O(n^{-1})$, which then implies $\norm{T_1} = O_P(n^{-1}) = O_P(N^{-1/2})$.

For $\epsilon$ small enough, we may condition on the complement of the event given by~\eqref{eq:omr_helper3} and use Assumption~\ref{assump:cost_bregman_divergence} (b) to get
\[\matrixnorm{\nabla^2 A(\rho \theta_k^* + (1-\rho)\widehat{\theta}_{ik}) - \nabla^2 A(\theta_k^*)}\leq \tilde{A}\rho\|\widehat{\theta}_{ik} - \theta_k^*\|\]
for $0 \leq \rho \leq 1$.
Hence,
\[
\begin{split}
\|U_i\| & = \left\|\int_{0}^{1}\braces{\nabla^2 A(\rho\theta_k^* + (1-\rho)\widehat{\theta}_{ik})-\nabla^2 A(\theta_k^*)}\,d\rho (\widehat{\theta}_{ik}-\theta_k^*)\right\|\\
& \leq\left\{\int_{0}^{1}\tilde{A} \rho\,d\rho\right\}\|\widehat{\theta}_{ik} - \theta_k^*\|^2 =(\tilde{A}/2)\|\widehat{\theta}_{ik} - \theta_k^*\|^2.    
\end{split}
\] 
This implies
\[
\sE\braces{\norm{U_i}^2} \leq (\tilde A/2)^2\sE\braces{\|\widehat{\theta}_{ik} - \theta_k^*\|^4} = O(n^{-2}),
\]
and similarly, $\norm{\sE\braces{U_i}} = O(n^{-1})$.
Decomposing $T_2$ into bias and variance components and arguing similar to Step 2, we get $\sE\braces{\norm{T_2}} = O(n^{-1})$, which then implies $\norm{T_2} = O_P(n^{-1}) = O_P(N^{-1/2})$.

Since $\norm{\paren{\widehat w_{ik} - w_k^*}U_i} \leq 2\norm{U_i}$, we may argue similarly to get $\norm{T_3} = O_P(n^{-1}) = O_P(N^{-1/2})$.

\textit{\underline{Step 5: Completing the proof.}}

Using Lemma~\ref{lemma:local_mle_property}, we have
\[
    \norm*{\sE\braces*{\widehat\theta_{ik}-\theta_k^*}} = O(n^{-1}).
\]
Furthermore, we have
\[
\sE\braces*{\norm{\widehat\theta_{ik}- \sE\braces{\widehat\theta_k^*}}^2} \leq \sE\braces*{\norm{\widehat\theta_{ik}-\theta_k^*}^2} = O(n^{-1}),
\]
so that
\[
\frac{1}{m}\sum_{i=1}^{m}(\widehat\theta_{ik}-\sE\braces{\widehat\theta_{ik}}) = O_P(N^{-1/2}).
\]
Having bounded all the terms in \eqref{eq:barycentre_expansion}, we conclude that $\theta_k^\dagger - \theta_k^* = O_P(N^{-1/2})$ as desired.
Combining this with the conclusion from Step 2 completes the proof of Theorem~\ref{thm:oracle_rate} (a).

Finally, notice that apart from $m^{-1}\sum_{i=1}^{m}(\widehat\theta_{ik}-\sE\braces{\widehat\theta_{ik}})$, all terms in \eqref{eq:barycentre_expansion} satisfy $O_{P}(n^{-1})$.
Assuming $m = o(n)$, they thus satisfy $o(N^{-1/2})$.
Under this scaling regime, we therefore get
\[
\sqrt{N}\paren*{\widehat \mG^{\operatorname{MR}}-\mG^*} = \sqrt{N}\paren*{\widebar \mG-\mG^*} + o_P(1),
\]
where $\widebar\mG$ is the simple average of the vectorized local estimates.
Note that
\[
\begin{split}
\sqrt{N}\paren*{\widebar \mG-\mG^*}
=&~\frac{1}{\sqrt{m}}\sum_{i=1}^{m}\left\{\sqrt{n}(\widehat{\mG}_{i} - \mG^*) - \sE\{\sqrt{n}(\widehat\mG_{i}-\mG^*)\}\right\} + \sqrt{nm}\sE\{\widehat\mG_{i}-\mG^*\}\\
\to&~\mathcal{N}(0,\mI^{-1}(\mG^*)) + o_{P}(1).
\end{split}
\]
In the last line above, the first term follows from the Lindeberg-Lévy Central Limit Theorem.
For the second term, the order is obtained by combining $\|\sE\{\widehat\mG_{i}-\mG^*\}\| = O(n^{-1})$ according to Lemma~\ref{lemma:local_mle_property}~(b) and $m=o(n)$.
\end{proof}

\subsection{Rate of convergence of DFMR}
\label{app:DFMR_rate}

\begin{proof}[\textbf{Proof of Proposition~\ref{prop:COAT_bound}}]
By definition, at most $m/2$ failure-free machines are not contained in the set $B(\widehat G^{\coat}; r^{\coat})$.
As such, there exists $j_{i} \leq m/2+1$ such that $\widehat{G}_{j_{i}} \in B(\widehat G^{\coat}; r^{\coat})$.  
This implies that
\[
D(\widehat{G}^{\coat},G^*)  \leq D(\widehat{G}^{\coat},\widehat{G}_{j_{i}}) + D(\widehat{G}_{j_{i}},G^*)  \leq r^{\coat} + D(\widehat{G}_{j_{m/2+1}},G^*).
\]
We complete the proof by applying Lemma~\ref{lem:coat_radius_conc} and \eqref{eq:distance_bulk_concentration} to bound these two terms respectively.
\end{proof}

\begin{lemma}
\label{lem:upper_lower_selected_set_distance}
Under assumptions~\ref{assumption:parameter-space}--\ref{assump:cost_bregman_divergence} and assumption~\ref{assumption:parametric_family}, for any sequence $\rho_{n} \to \infty$, as $n,m\to\infty$, we have
\[
\sP\braces*{B(G^*, 2^{-1}\rho_{n} n^{-1/2}) \subset B(\widehat G^{\coat},\rho_{n} r^{\coat}) \subset B(G^*,2\rho_{n} n^{-1/2})} \to 1
\]        
\end{lemma}

\begin{proof}
    This follows immediately from Lemma \ref{lem:coat_radius_conc} and Proposition \ref{prop:COAT_bound}.
\end{proof}

\begin{lemma}
\label{lem:upper_bound_failure_free_omitted}
Under assumptions~\ref{assumption:parameter-space}--\ref{assump:cost_bregman_divergence} and assumption~\ref{assumption:parametric_family}, there is a universal constant $C > 0$ such that
\[
\sP\braces*{\abs*{\sB^c \cap \selected_{\rho}^c} \leq Cm^{1/2}} \to 1
\]
as $n, m \to \infty$.
In particular, we have
\[
\abs{\selected_\rho} \geq (1 - o_P(1))\abs{\sB^c}.
\]
\end{lemma}

\begin{proof}
Using Lemma \ref{lem:upper_lower_selected_set_distance}, there is an event with probability tending 1 on which
\[
|\selected_\rho \cap\sB^c |  =  \sum_{i\in \sB^c}\mathbbm{1}(i \in \sS_\rho) \leq \sum_{i\in \sB^c}\mathbbm{1}(D(\widehat{G}_i, G^*) > C^{-1}\rho n^{-1/2}).
\]
Using Lemma~\ref{lemma:distance_concentration} and our definition of $\rho$, we get
\begin{equation}
    \sP\paren*{D(\widehat{G}_i, G^*) > C^{-1}\rho n^{-1/2}} = O(\rho^{-8}) = o(m^{-1/2}).
\end{equation}
Hence for $m$ large enough, we get
\[
\sP\paren*{|\selected_\rho \cap\sB^c | \geq m^{1/2}} \leq \sP\paren*{Z \geq 2\sE\braces{Z} },
\]
where
\[
Z = \sum_{i\in \sB^c}\mathbbm{1}(D(\widehat{G}_i, G^*) > C^{-1}\rho n^{-1/2}).
\]
Applying Chernoff's inequality completes the proof.
\end{proof}

\begin{proof}[\textbf{Proof of Theorem~\ref{thm:main_theorem}}]
Since the oracle mixture reduction $\widehat G^{\omr}$ is the solution to the optimisation problem~\eqref{eq:GMR} in which the aggregation is performed over the failure-free machines $\sB^c$.
Therefore, Lemma~\ref{lem:barycentre_equivalence} implies that
\[
w_k^{\omr} 
= 
|\sB^{c}|^{-1}\sum_{i\in \sB^{c}} \widetilde{w}_{ik},\quad\quad
\theta_k^{\omr} 
= 
\argmin_{\theta} \sum_{i\in \sB^{c}} \widetilde{w}_{ik} c(\widetilde{\theta}_{ik},\theta).
\]

Recall the definition of the DFMR estimator $\widehat G^{\dfmr}$ in~\eqref{eq:DFMR} as the mixture reduction over the selected set $\sS_{\rho}$ defined in~\eqref{eq:cred-set-estimate}, with $\rho = m^{\delta}$, $\delta > 1/14$.
First, condition on the high probability events guaranteed by Lemma~\ref{lem:upper_lower_selected_set_distance}.
This implies, together with Lemma~\ref{lem:barycentre_equivalence}, that that the mixing weights and subpopulation parameters of $\widehat G^{\dfmr}$ can be written respectively as
\begin{equation}
\label{eq:dfmr_true_align}
w_k^{\dfmr} 
= 
|\selected_\rho|^{-1}\sum_{i\in \selected_\rho} \widetilde{w}_{ik},\quad\quad
\theta_k^{\dfmr} 
= 
\argmin_{\theta} \sum_{i\in \selected_\rho} \widetilde{w}_{ik} c(\widetilde{\theta}_{ik},\theta),
\end{equation}
for $k=1,2,\ldots,K$.
In other words, minimizing the composite transportation distance defined in~\eqref{eq:CTD_def} fully aligns all mixture components (c.f.~\eqref{eq:aligned_component_def}).

We now use the formulas in~\eqref{eq:dfmr_true_align} to express the error of the mixing weight and subpopulation parameter estimates.
The mixing weight errors can be decomposed as
\[
w_k^{\dfmr} - w_k^* 
=|\selected_\rho|^{-1}\sum_{i \in \sB^c}(\widetilde{w}_{ik} - w_k^*)  - |\selected_\rho|^{-1}\sum_{i \in \sB^c\cap\selected_\rho^c}(\widetilde{w}_{ik} - w_k^*)  + |\selected_\rho|^{-1}\sum_{i \in \sB\cap \selected_\rho}(\widetilde{w}_{ik} - w_k^*).
\]
The first term satisfies
\[
|\selected_\rho|^{-1}\sum_{i\in \sB^c}(\widetilde{w}_{ik} - w_k^*) = \frac{|\sB^c|}{|\selected_\rho|}\paren*{w^{\omr}_k - w_k^*},
\]
with $\frac{|\sB^c|}{|\selected_\rho|} = \frac{|\sB^c|}{|\sB^c| + |\selected_\rho\cap\sB|} + o_P(1)$.
To bound the second term, we first compute
\begin{equation}
\label{eq:main_theorem_helper_weight}
\begin{split}
\sE\sqbracket*{\left|\sum_{i \in \sB^c\cap\selected_\rho^c}(\widetilde{w}_{ik} - w_k^*)\right|} & \leq \sum_{i \in \sB^c} \sE\sqbracket*{\left|\widetilde{w}_{ik} - w_k^*\right|\mathbbm{1}\paren*{i \notin \selected_\rho}}\\
& \leq m \sE\sqbracket*{\left|\widetilde{w}_{ik} - w_k^*\right|\mathbbm{1}\paren*{D(\widehat G_i, G^*) > 2^{-1}\rho n^{-1/2}}} \\
& \leq m\sE\sqbracket*{\left|\widetilde{w}_{ik} - w_k^*\right|^8}^{1/8} \left\{\sP\paren*{D(\widehat G_i, G^*) > 2^{-1}\rho n^{-1/2}}\right\}^{7/8},
\end{split}
\end{equation}
where the second inequality follows from Lemma~\ref{lem:upper_lower_selected_set_distance} and the last inequality follows from H\"older's inequality.
Because $\left|\widetilde{w}_{ik} - w_k^*\right| \leq \norm{\widehat\mG - \mG^*}$, Lemma~\ref{lemma:local_mle_property}~(b) implies $\sE\sqbracket*{\left|\widetilde{w}_{ik} - w_k^*\right|^8}^{1/8} = O(n^{-1/2})$.
Meanwhile, using Lemma~\ref{lemma:distance_concentration}~(b), we get
\[
\left\{\sP\paren*{D(\widehat{G}_i, G^*) > C^{-1}\rho n^{-1/2}}\right\}^{7/8} = O(\rho^{-7}) = o(m^{-1/2})
\]
where the last equality holds since $\rho = \Omega(m^{\delta})$ with $\delta>1/14$.
Plugging these into the right hand side~\eqref{eq:main_theorem_helper_weight} and using Markov's inequality gives
\begin{equation}
\label{eq:main_thm_weight_decomposition}
    \sum_{i \in \sB^c\cap\selected_\rho^c}(\widetilde{w}_{ik} - w_k^*) = o_P(mN^{-1/2}),
\end{equation}
which we combine with Lemma~\ref{lem:upper_bound_failure_free_omitted} to get
\[
|\selected_\rho|^{-1}\sum_{i \in \sB^c\cap\selected_\rho^c}(\widetilde{w}_{ik} - w_k^*) = o_P(N^{-1/2}).
\]
To bound the third term, we write
\[
\begin{split}
    \abs*{|\selected_\rho|^{-1}\sum_{i \in \sB\cap \selected_\rho}(\widetilde{w}_{ik} - w_k^*)} & \leq |\selected_\rho|^{-1}\sum_{i \in \sB\cap \selected_\rho}\abs*{\widetilde{w}_{ik} - w_k^*} \\
    & = |\selected_\rho|^{-1}\sum_{i \in \sB, D(\widehat G_i,G^*) \leq 2\rho n^{-1/2}}\abs*{\widetilde{w}_{ik} - w_k^*} \\
    & \leq |\selected_\rho|^{-1}\sum_{i \in \sB, D(\widehat G_i,G^*) \leq 2\rho n^{-1/2}}\norm*{\widetilde{\mG}_i - \mG^*} \\
    & = O_P(\widetilde\alpha_m(2\rho n^{-1/2})\cdot \phi(2\rho n^{-1/2})),
\end{split}
\]
where 
\begin{equation}
\label{eq:phi_definition}
    \phi(r) = \min\braces*{\norm{\mG - \mG^*}/r \colon G \in \sG_K, D(G,G^*) \leq r}.
\end{equation}
Together with Lemma \ref{lem:strong_identifiability}, this proves~\eqref{eq:main_theorem} for the coordinates corresponding to the mixing weights.

To do the same for the subpopulation parameter estimates, we first make use of~\eqref{eq:omr_helper5} to get
\[
\theta_{k}^{\dfmr} -\theta_k^* = \paren*{\abs{\selected_\rho}w_k^{\dfmr}}^{-1} \nabla^2 A(\theta^*_k)^{-1}\sum_{i \in \selected_\rho} \widetilde{w}_{ik} \paren*{\nabla A(\widetilde\theta_{ik}) - \nabla A(\theta_k^*)} + o(\norm{\theta_k^{\dfmr} - \theta_k^*}).
\]
We similarly decompose the summation above as 
\[
\begin{split}
\sum_{i \in \selected_\rho} \widetilde{w}_{ik} \paren*{\nabla A(\widetilde\theta_{ik}) - \nabla A(\theta_k^*)} 
& = \sum_{i \in \in \sB^c} \widetilde{w}_{ik} \paren*{\nabla A(\widetilde\theta_{ik}) - \nabla A(\theta_k^*)} - \sum_{i \in \sB^c\cap\selected_\rho^c} \widetilde{w}_{ik} \paren*{\nabla A(\widetilde\theta_{ik}) - \nabla A(\theta_k^*)} \\
& \quad\quad + \sum_{i \in \sB\cap \selected_\rho} \widetilde{w}_{ik} \paren*{\nabla A(\widetilde\theta_{ik}) - \nabla A(\theta_k^*)}
\end{split}
\]
and consider each term in this decomposition separately.
The second term can be bounded as
\[
\begin{split}
&~\sE\sqbracket*{\left\|\sum_{i \in \sB^c\cap\selected_\rho^c} \widetilde{w}_{ik} \paren*{\nabla A(\widetilde\theta_{ik}) - \nabla A(\theta_k^*)}\right\|} \\ 
\leq &~\sum_{i \in \sB^c} \sE\sqbracket*{\left\|\widetilde{w}_{ik} \paren*{\nabla A(\widetilde\theta_{ik}) - \nabla A(\theta_k^*)}\right\|\mathbbm{1}\paren*{i \notin \selected_\rho}}\\
\leq &~m \sE\sqbracket*{\left\|\nabla A(\widetilde\theta_{ik}) - \nabla A(\theta_k^*)\right\|\mathbbm{1}\paren*{D(\widetilde G_i, G^*) > 2^{-1}\rho n^{-1/2}}} \\
\leq &~m\sE\sqbracket*{\left\|\nabla A(\widetilde\theta_{ik}) - \nabla A(\theta_k^*)\right\|^8}^{1/8} \sP\paren*{D(\widetilde G_i, G^*) > 2^{-1}\rho n^{-1/2}}^{7/8}.
\end{split}
\]
Reasoning as in the calculations for mixing weights gives
\[
\paren*{\abs{\selected_\rho}w_k^{\dfmr}}^{-1} \nabla^2 A(\theta^*_k)^{-1}\sum_{i \in \sB^c\cap\selected_\rho^c} \widetilde{w}_{ik} \paren*{\nabla A(\widetilde\theta_{ik}) - \nabla A(\theta_k^*)} = o_P(N^{-1/2}).
\]
The third term can be bounded similarly as before to get
\begin{equation} \nonumber
    \begin{split}
    \paren*{\abs{\selected_\rho}w_k^{\dfmr}}^{-1} \nabla^2 A(\theta^*_k)^{-1}\sum_{i \in \sB\cap \selected_\rho} \widetilde{w}_{ik} \paren*{\nabla A(\widetilde\theta_{ik}) - \nabla A(\theta_k^*)} & = O_P(\widetilde\alpha_m(2\rho)\cdot \phi(2\rho n^{-1/2})) \\
    & = O_P(\widetilde\alpha_m(2\rho n^{-1/2})\cdot \phi(2\rho n^{-1/2})),
    \end{split}
\end{equation}
where we again used Lemma \ref{lem:strong_identifiability}.
Finally, using the expansion~\eqref{eq:omr_helper5}, the first term can be written as
\[
\paren*{\abs{\selected_\rho}w_k^{\dfmr}}^{-1}\sum_{i \in \selected_\rho} \widetilde{w}_{ik} \paren*{\nabla A(\widetilde\theta_{ik}) - \nabla A(\theta_k^*)} = \frac{\abs{\sB^c}w_k^{\omr}}{\abs{\selected_\rho}w_k^{\dfmr}}\paren{\theta_k^{\omr} - \theta_k^* + o_P\paren{\norm{\theta_k^{\omr} - \theta_k^*}}}.
\]
We use~\eqref{eq:main_thm_weight_decomposition} to get
\[
\frac{\abs{\sB^c}w_k^{\omr}}{\abs{\selected_\rho}w_k^{\dfmr}} = \frac{\abs{\sB^c}}{\abs{\sB^c} + \abs{\selected_\rho\cap\sB}} + o_P(1),
\]
which completes the proof of~\eqref{eq:main_theorem}.
\end{proof}

\begin{proof}[\textbf{Proof of Corollary~\ref{cor:failure_density}}]
Based on the assumption that $\sP(D(\xi_i,G^*)\leq r) = O(r^{3})$ as $r\to 0$, for each failure machine, we have
\begin{equation} \nonumber
    \sP\paren{D(\widetilde{G}_i,G^*) \leq 2\rho n^{-1/2}} = O((\rho n^{-1/2})^3).
\end{equation}
Using the union bound, we have
\[
\begin{split}
\sP\paren{\min_{i \in \sB}D(\widetilde{G}_i,G^*) \leq 2\rho n^{-1/2} } \leq&~m \sP\paren{D(\widetilde{G}_i,G^*) \leq 2\rho n^{-1/2}}\\
=&~O(m\rho^{3} n^{-3/2}) = O(\rho^{3}n^{-1/2})
\end{split}
\]
since $m\leq n$.
Under our assumption that $\rho = o(n^{1/6})$, the quantity on the right hand side tends to zero.
This implies that
\[
\sP\left(\{ D(\xi_i,G^*) > 2\rho n^{-1/2},~\forall i \in \sB  \}\right)\to 1,
\]
which completes the proof that
\be
\label{eq:helper_corollary}
\sP(\widetilde\alpha_m(2\rho n^{-1/2})>0) = o(1).
\ee

For the last statement of the corollary, notice from the proof of Theorem~\ref{thm:main_theorem} that
\[
\widehat{\mG}^{\dfmr} - \mG^* = \gamma \paren*{\widehat{\mG}^{\omr} - \mG^*} + O_P(\widetilde\alpha_m(2\rho n^{-1/2})) + o_P(N^{-1/2})
\]
where
\[
\gamma = \frac{\abs{\sB^c}}{\abs{\sB^c} + \abs{\selected_\rho\cap\sB}} + o_P(1).
\]
Under the complement of the event considered above, we have $\abs{\selected_\rho \cap \sB} = 0$, so that $\gamma = 1 + o_P(1)$.
This combines with~\eqref{eq:helper_corollary} implies that 
\[
\widehat{\mG}^{\dfmr} = \widehat{\mG}^{\omr}  + o_P(N^{-1/2}),
\]
which completes the proof of the corollary.
\end{proof}

\begin{lemma}
\label{lem:strong_identifiability}
    Let $\phi$ be defined as in \eqref{eq:phi_definition}.
    If $\sG_K$ is strongly identifiable, then $\limsup_{r \to 0} \phi(r) < \infty$.
\end{lemma}

\begin{proof}
    Under Assumption \ref{assumption:parameter-space}, the true subpopulation parameters $\theta_1^*,\theta_2^*,\ldots,\theta_K^*$ are distinct. 
    Under strong identifiability, we observe that the collection
    \[
    \braces{f(x,\theta_k^*), \partial_1f(x,\theta_k^*),\ldots,\partial_pf(x,\theta_k^*) \colon k=1,2,\ldots, K}
    \]
    are linearly independent as functions on $\R^d$ where $\partial_j f(x, \theta_k^*)$ is the partial derivative of the density function with respect to the $j$th component of parameter $\theta$.
    Since this collection is also contained in $L^2(\R^d)$, they must also be linearly independent within this subspace.
    This statement extends also to the collection
    \[
    \braces{f(x,\theta_k^*) - f(x,\theta_K^*), k=1,2,\ldots,K-1}\cup\braces{\partial_1f(x,\theta_k^*),\ldots,\partial_pf(x,\theta_k^*) \colon k = 1,2,\ldots, K}
    \]
    Now, notice that the matrix $\mH^*$ in Lemma \ref{lemma:L2_distance_densities} is the Gram matrix of this collection of functions.
    By linear independence, $\mH^*$ must therefore be invertible, implying that $\lambda_{\min}(\mH^*) > 0$.
    We may then write
    \[
        \begin{split}
            \norm{\mG-\mG^*}^2 & \leq \lambda_{\min}(\mH^*)^{-1}(\mG-\mG^*)^T\mH^*(\mG-\mG^*) \\
            & = D(G,G^*) + O(\norm{\mG-\mG^*}^3).
        \end{split}
    \]    
    Since $\norm{\mG-\mG^*} \to 0$ as $D(G,G^*) \to 0$, there exists $r$ small enough so that the $O(\norm{\mG-\mG^*})$ term is much smaller than $\norm{\mG-\mG^*}^2$ for $D(G,G^*) \leq r$, in which case we have
    \[
        \phi(r') \leq 2\lambda_{\min}(\mH^*)^{-1}
    \]
    for any $r' \leq r$.
\end{proof}
We have established the rate of convergence of the DFMR estimator for general $\gF$ that is not necessarily first-order strongly identifiable.
This Lemma implies that for first-order strongly identifiable $\gF$, $\phi(r) = O(1)$.
Hence, we can remove $\phi(r)$ in the rate results above for first-order strongly identifiable mixtures.

\subsection{Technical lemmas}
\label{app:technical_lemmas}
\begin{lemma}[Concentration of order statistics]
\label{lem:quantile_conc}
Let $q$ be the quantile function of a random variable $X$.
Given $\lbrace X_1, \ldots,X_m\rbrace$ be $m$ IID observations of $X$.
Let $j_1,\ldots, j_m$ be a permutation of $1,\ldots, m$ such that 
\[X_{j_1}\leq X_{j_2}\leq \cdots \leq X_{j_m}\]
are order statistics.
Then for any $\epsilon > 0$, with probability at least $1-2\exp(-2m\epsilon^2)$, we have
\begin{equation}
    \label{eq:quantile_conc}
    q\left(\alpha - \epsilon\right) \leq X_{j_{\alpha m}} \leq q\left(\alpha + \epsilon\right),
\end{equation}
uniformly for all $\alpha$ satisfying $\epsilon \leq \alpha \leq 1 - \epsilon$.
\end{lemma}

\begin{proof}
Let $F$ denote the CDF of $Z$ and let $\widehat F_m$ denote the empirical CDF for the random sample.
The DKW inequality~\citep{massart1990tight} states that with probability at least $1-2\exp(-2m\epsilon^2)$, we have $\| F_m - F\|_\infty \leq \epsilon$.
Next note that $X_{j_{\alpha m}} = \widehat q_m(\alpha)$, where $\widehat q_m$ is the empirical quantile function.
Since $q$ and $\widehat q_m$ are the almost sure inverses of $F$ and $\widehat F_m$ respectively, we may reflect this inequality around the $y=x$ line to obtain \eqref{eq:quantile_conc}.
\end{proof}

\begin{lemma}[Maximum probability of convex sublevel sets]
\label{lem:convex_function_CDF}
Let $\psi\colon \R^d \to \R$ be any even convex function and let $t \geq 0$ be any threshold.
Let $Z$ be a standard Gaussian random vector in $\R^d$.
Then for any vector $\vv \in \R^d$, we have
\[
\sP\paren*{\psi(Z+\vv) \leq t} \leq \sP\paren*{\psi(Z) \leq t}.
\]
\end{lemma}

\begin{proof}
Let
\[
f(\vx)= \mathbbm{1}\paren*{\psi(\vx-\vv) \leq t}\exp\paren*{-\frac{\norm{\vx}^2}{2}},
\]
\[
g(\vx) =\mathbbm{1}\paren*{\psi(\vx+\vv) \leq t}\exp\paren*{-\frac{\norm{\vx}^2}{2}},
\]
\[
h(\vx)=\mathbbm{1}\paren*{\psi(\vx) \leq t}\exp\paren*{-\frac{\norm{\vx}^2}{2}}.
\]
Since $\psi$ is convex, for any $\vu,\vw \in \R^d$, we have
\[
\psi\paren*{\frac{\vu+\vw}{2}} \leq \max\braces*{\psi(\vu), \psi(\vw)},
\]
which implies
\[
\mathbbm{1}\paren*{\psi((\vu+\vw)/2)\leq t} \geq \mathbbm{1}\paren*{\psi(\vu) \leq t}\mathbbm{1}\paren*{\psi(\vw) \leq t}.
\]
Furthermore, the simple inequality $\norm{\vu + \vw}^2 \leq 2\vu^2 + 2 \vw^2$ implies
\[\exp\paren*{-\frac{\norm{(\vu+\vw)/2}^2}{2}} \geq \exp\paren*{-\frac{\norm{\vu}^2}{2}}^{1/2}\exp\paren*{-\frac{\norm{\vw}^2}{2}}^{1/2}.\]    
Combining the above two inequalities, we obtain
\[h((\vu + \vw)/2) \geq f^{1/2}(\vu)g^{1/2}(\vw).\]
We then apply the Pr{\'e}kopa-Leindler inequality~\citep[Theorem 3.15]{wainwright2019high} to get
\[
\begin{split}
\sP\paren*{\psi(Z) \leq t} & = \sE\{\mathbbm{1}(\psi(Z)\leq t)\} = \paren*{2\pi}^{-d/2}\int h(\vx)d\vx \\
& \geq \paren*{\paren*{2\pi}^{-d/2}\int f(\vx)d\vx}^{1/2}\paren*{\paren*{2\pi}^{-d/2}\int g(\vx)d\vx}^{1/2} \\
& = \sP\paren*{\psi(Z+\vv) \leq t}^{1/2} \sP\paren*{\psi(Z-\vv) \leq t}^{1/2}.
\end{split}
\]
Finally, using the fact that $\psi$ is even and the symmetry of the standard Gaussian distribution, we have
\[
\sP\paren*{\psi(Z-\vv) \leq t}  = \sP\paren*{-\psi(Z+\vv) \leq t}  = \sP\paren*{\psi(Z+\vv) \leq t},
\]
which completes the proof.
\end{proof}

\begin{lemma}[Rate of MSE simple average estimator]
\label{lemma:mse_rate_simple_average}
Let $\theta^*$ be some interior point of a convex and compact subset $\Theta \subset \sR^{d}$.
Let $\widehat\theta_1,\ldots, \widehat\theta_m$ be $m$ independent estimates for $\theta^*$, each based on $n$ samples.
Suppose that $\sE\|\hat\theta_i - \theta^*\|^{2} = O(n^{-1})$ and $\|\sE\{\hat\theta-\theta^*\}\|  = O(n^{-1})$ for $i=1,\ldots, m$.
Let 
\[\bar\theta = \frac{1}{m}\sum_{i=1}^{m} \hat\theta_i,\]
and $N= mn$, then when $m = O(n)$, we have 
\[\sE\|\bar\theta-\theta^*\|^2 = O(N^{-1}).\]
\end{lemma}
The proof of the Lemma can be found in~\citet[Lemma 24]{huang2019distributed}.


\subsection{Recap on exponential family and Bregman divergences}
\label{app:recap-bregman}
In this section, we briefly review exponential family and Bregman divergence.
We recommend~\citet{banerjee2005clustering} and~\citet{kunstner2021homeomorphic} for a more detailed review of Bregman divergences in clustering algorithms and their relation with exponential families.

\subsubsection{Exponential family}
The density function of a distribution from exponential family can be written as
\[
f(x;\theta) = \exp(\theta^{\top}T(x) - A(\theta))
\]
with respect to some reference measure $\nu(\cdot)$.
We call $\theta = (\theta_1, \theta_2,\ldots,\theta_m)^{\top}$ the natural parameter and $T(x)=(T_1(x), T_2(x),\ldots, T_m(x))^{\top}$ the natural sufficient statistics. 
The natural parameters and natural sufficient statistics of some widely used distributions in the exponential family are given in Table~\ref{tab:exp_family}. 
An exponential family is \emph{regular} if the parameter space $\Theta$ is an open set.
The exponential family is \emph{minimal} if the functions in $T(x)$ are linearly independent.
Without loss of generality, we assume the exponential families in our discussion are minimal.
\begin{table}[!htpb]
\centering
\caption{Parameter specification and statistics of widely used density functions in full exponential family. In all cases, the base measure $\nu$ is Lebesgue measure that is modulated by a factor $h(\cdot)$. The $\psi(\cdot)$ is the digamma function.}
\label{tab:exp_family}
\resizebox{0.9\textwidth}{!}{
\begin{tabular}{lll}
\toprule
$\gF$ & $T(x)$ & $A(\theta)$\\
\midrule
\multicolumn{3}{c}{Univariate distribution}\\
Exponential &  $x$ & $-\log(-\theta)$\\
Weibull (known $k$)& $x^k$ & $-\log(-\theta)$ \\
Laplace (known $\mu$)& $|x-\mu|$ &  $\log(-2/\theta)$\\
Rayleigh& $x^2$ & $-\log(-2\theta)$\\
Log-normal& $(\log x, \log^2 x)^{\top}$ & $-\theta_1^2/\theta_2-1/\sqrt{2\theta_2}$ \\
Gamma& $(\log x,x)^{\top}$ & $\log\Gamma(\theta_1+1)-(\theta_1+1)\log(-\theta_2)$ \\
Inverse Gamma& $(\log x, 1/x)^{\top}$ & $\log\Gamma(-\theta_1-1))+(\theta_1+1)\log(-\theta_2)$\\
\midrule
\multicolumn{3}{c}{Multivariate distribution}\\
Gaussian Gamma& $(\log \tau, \tau, \tau x, \tau x^2)^{\top}$ & $\log\Gamma(\theta_1+\frac{1}{2}) - \frac{1}{2}\log(-2\theta_4) - (\theta_1+\frac{1}{2})\log(-\theta_2+\frac{\theta_3^2}{4\theta_4})$\\
Dirichlet& $\log \vx$ & $\{\mathbbm{1}_K^{\top}\log \Gamma(\boldsymbol\theta+1) - \log \{\mathbbm{1}_K^{\top}\Gamma(\boldsymbol\theta+1)\}$\\
\bottomrule
\end{tabular}}
\end{table}

Based on the parameterisation of the exponential family, it is easy to see that 
\begin{equation}
\label{eq:Atheta_exponential_family}
A(\theta) = \log \int \exp\{\theta^{\top}T(x)\}\,\nu(dx).
\end{equation}
This function is called the log partition function.
Moreover, $A(\cdot)$ is a convex function of $\theta$. 
When exponential family is minimal, then $A(\cdot)$ is strictly convex~\citet{wainwright2008graphical}.

\subsubsection{Bregman Divergence}
\label{sec:bregman_divergence_appendix}
Bregman divergence is a generalisation of squared Euclidean distance based on convex functions. 
For a strongly convex function $A$, the Bregman divergence induced by $A$ is denoted as $D_{A}(\theta, \theta')$. 
It represents the difference between the function value at $\theta$ and its linear approximation constructed at $\theta'$:
\[
D_{A}(\theta, \theta') = A(\theta) - A(\theta') - (\theta - \theta')^{\top}  \nabla A(\theta').
\]
Below are some examples of Bregman divergence:

\begin{itemize}
    \item \textbf{Euclidean distance}: For $x \in \mathbb{R}^d$ and $A(\theta) = \|\theta\|^2/2$, the Bregman divergence is:
    $D_{A}(\theta, \theta') = \|\theta - \theta'\|^2/2.$
    
    \item \textbf{Mahalanobis distance}: For $A(\theta) = \theta^{\top}A\theta/2$ with some positive definite matrix $A$, the Bregman divergence is:
    $D_A(\theta, \theta') = \|\theta - \theta'\|_A^2/2 = (\theta-\theta')^{\top}A(\theta-\theta')/2$.
    
    \item \textbf{KL divergence under exponential family}: For an exponential family $f(x;\theta)$, let the log-partition function be $A(\theta)$.
    Then the Bregman divergence induced by the log-partition function is the KL divergence between the densities:
    \[
    \begin{split}
    \KL(f(x;\theta')\|f(x;\theta)) &= \int \log \left\{\frac{f(x;\theta')}{f(x;\theta)}\right\}f(x;\theta')\,dx \\
    &= \int \{S^{\top}(x)\theta' - A(\theta') - S^{\top}(x)\theta + A(\theta)\}f(x;\theta')\,dx \\
    &= A(\theta) - A(\theta') + (\theta' - \theta)^{\top}\mathbb{E}_{f(x;\theta')}[S(X)]\\
    &= A(\theta) - A(\theta') + (\theta' - \theta)^{\top}\nabla A(\theta') = D_{A}(\theta, \theta').
    \end{split}
    \]
\end{itemize}
We now show that the function $A$ in the above examples satisfies Assumption~\ref{assump:cost_bregman_divergence}.
\begin{itemize}
    \item For Euclidean distance, we have $\nabla^2 A(\theta) = \mI$ which clearly satisfies both assumptions.
    \item For Mahalanobis distance, we have $\nabla^2 A(\theta) = A$.
    Since $A>0$, we have $\lambda_{\min}(A)\mI \leq \nabla^2 A(\theta) \leq \lambda_{\max}(A)\mI$ for all $\theta$, satisfying both assumptions.
    \item For KL divergence under exponential family, we have 
    \[\nabla^2 A(\theta) = \sE\{T(X)T^{\top}(X)\} - \sE\{T(X)\}\sE\{T^{\top}(X)\},\]
    which is the variance matrix of the sufficient statistics. 
    Thus, the Hessian also satisfies the assumption in a small neighbourhood of the true parameter value.

    Note also that for exponential family $A(\theta)$ is the cumulative generating function, it is well known that the function is infinitely differentiable.
    Hence, its third order derivative must be bounded in the compact parameter space $\Theta$ and the $\nabla^2 A(\theta)$ must be locally Lipschitz.
\end{itemize}

%% file: sections/appendix_exp.tex
\section{Experiments}

\subsection{Closed-form of $L^2$ distance}
\label{app:L2_closed_form}
Under many families $\gF$, the $L^2$ distance we used in finding the COAT has a closed form.
\[
D^2(G',G)=\int \{f_{G'}(x) - f_{G}(x)\}^2 \,dx =w^{\top}S_{OO}w - 2w^{\top}S_{OR} w'+ (w')^{\top} S_{RR}(w')
\]
where $w = (w_i)$, $w' = (w'_j)$, and $S_{OO}$, $S_{OR}$, and $S_{RR}$ 
are matrices with their $(i, j)$-th elements respectively being the integral of the product of two densities: $\int f(x;\theta_i) f(x;\theta_j)\,dx$, 
$\int f(x;\theta_i) f(x;\theta_j')\,dx$, and $\int f(x;\theta_i') f(x;\theta_j')\,dx$.

For the Gaussian, Gamma, and Poisson distributions considered in our simulations, the integral admits the following analytical form.
\begin{itemize}
\item Under finite \emph{Gaussian} mixtures, we have
\[
\int \phi(x;\mu_1,\Sigma_1)\phi(x;\mu_2,\Sigma_2)\,dx 
= \phi(\mu_1;\mu_2,\Sigma_1 + \Sigma_2),
\]
for any two Gaussian densities.

\item Under finite \emph{Gamma} mixtures, we have
\[
\int f(x;r_1,\theta_1)f(x;r_2,\theta_2)\,dx 
= \frac{\theta_1^{r_2-1}\theta_2^{r_1-1}\Gamma(r_1+r_2-1)}{
(\theta_1+\theta_2)^{r_1+r_2-1}\Gamma(r_1)\Gamma(r_2)},
\]
provided that $r_1+r_2>1$.

\item Under finite \emph{Poisson} mixtures, we have
\[
\sum_{x=0}^{\infty} f(x;\theta_1)f(x;\theta_2) 
= e^{-(\theta_1+\theta_2)} I_0\!\left(2\sqrt{\theta_1\theta_2}\right),
\]
where $I_0$ is the modified Bessel function of order $0$.

\end{itemize}

\subsection{Details about the performance metrics}
\label{app:exp_performance_metric}
We give the  details of the performance metrics we used in the experiments.
Let $G^{(r)}$ be the true mixing distribution in $r$th repetition and 
$\widehat G^{(r)} = \sum_{k=1}^{K} \widehat{w}_k^{(r)} \delta_{\widehat\theta_k^{(r)}}$ be an estimate. 
We assess the performance using the following metrics:
\begin{enumerate}
\item \textbf{Adjusted Rand Index} (ARI). 
Mixture models are often used for model-based clustering. For a given $G$, we can assign a unit with observed value $x$ to a cluster via: 
\begin{equation*} 
k(x) = \argmax_{k} \{w_k \phi(x; \mu_k, \Sigma_k)\}. 
\end{equation*} 
We measure the performance of the clustering results of the full dataset based on $\widehat G^{(r)}$ and $G^{(r)}$ by calculating the adjusted Rand index (ARI).
The ARI is a well-known measure of the similarity between two clustering outcomes. 
Suppose that the observations in a dataset are divided into $K$ clusters $A_1, A_2, \ldots, A_K$ by one method and $K'$ clusters $B_1, B_2, \ldots, B_{K'}$ by another. 
Let $N_i = \text{Card} (A_i), ~ M_j = \text{Card}(B_j), ~~ N_{ij} = \text{Card}(A_i\cap B_j)$ for $i \in [K]$ and $j \in [K']$. 
The ARI of these two clustering results is calculated as follows: 
\[ \text{ARI} = \dfrac{ \sum_{i, j } \binom{N_{ij}}{2} - \binom{N}{2}^{-1} \sum_{i, j}\binom{N_{i}}{2}\binom{M_{j}}{2}} {\frac{1}{2} \sum_{i}\binom{N_{i}}{2} + \frac{1}{2} \sum_{j} \binom{M_{j}}{2} - \binom{N}{2}^{-1} \sum_{i, j}\binom{N_{i}}{2}\binom{M_{j}}{2}} \] 
where $\binom{n}{k}$ is the number of combinations of $k$ from a given set of $n$ elements. 
A value close to 1 for the ARI indicates a high degree of agreement between the two clustering methods, and it is equal to 1 when the two methods completely agree.

\item \textbf{Transportation distance} ($W_1$).
The transportation distance between two mixing distributions is defined as 
\begin{equation*}
W_1(\widehat{G}, G)=\argmin\left\{\sum_{ij}\pi_{ij} D_1(\widehat\theta_i,\theta_j): \sum_{j}\pi_{ij} = \widehat{w}_i^{(r)},\sum_{i}\pi_{ij} = w_j\right\}
\end{equation*}
The following forms of $D_1(\widehat\theta_i,\theta_j)$ are used under different models:
\begin{itemize}
\item Under Gaussian case, we let
$
D_1(\widehat\theta_{i}, \theta_{j}) 
=
  \|\widehat\mu_{i}-\mu_{j}\|_{2} + \|(\widehat\Sigma_i)^{1/2}-(\Sigma_{j})^{1/2}\|_{F}$
where $\Sigma^{1/2}\geq 0$ is the square root of $\Sigma\geq 0$ and $\|\Sigma\|_F = \sqrt{\text{tr}(\Sigma^{\top}\Sigma)}$ is the Frobenius norm of a matrix.

\item Under Gamma case, we let $
D_1(\widehat\theta_{i}, \theta_{j}) 
=|\widehat{r}_i - \widehat{r}_j| + |\widehat{\theta}_i - \widehat{\theta}_j|$.
\end{itemize}
\end{enumerate}

\subsection{More experiment results under Gaussian mixture}
\label{app:more_results}
In this section, we present the performance of various methods in terms of ARI. 
The relative performance of the estimators in terms of ARI mirrors that of $W_1$.
Figure~\ref{fig:ari_fixoverlap} corresponds to the settings in Figure~\ref{fig:w1_fixoverlap}. 
\begin{figure}[!ht]
  \centering
  \includegraphics[width=0.9\textwidth]{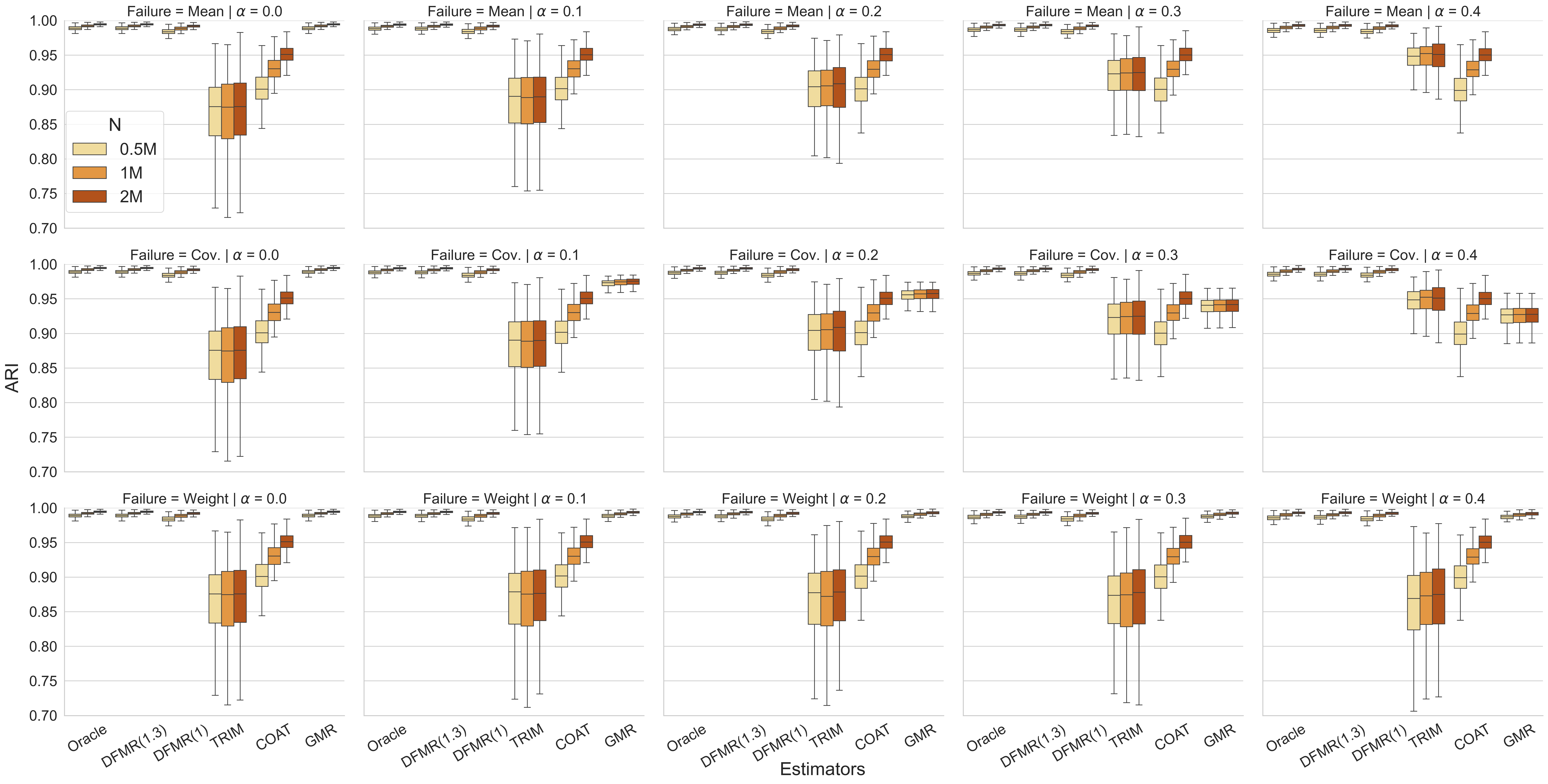}
\caption{\footnotesize{The ARI of different clustering methods based on the full dataset as $N$ and $\alpha$ varies, with $m=100$ and \texttt{MaxOmega}$=0.3$.}}
\label{fig:ari_fixoverlap}
\end{figure}

Figures~\ref{fig:ari_fixn_varym} and~\ref{fig:ari_fixN_varym} correspond to the settings in Figures~\ref{fig:w1_fixn_varym} and~\ref{fig:w1_fixN_varym}, respectively.
The y-axis is constrained to $[0.8,1.0]$ to better illustrate differences among estimators. 
The GMR estimator has a low ARI under mean and covariance attacks, and is omitted in the first two panels in both.
\begin{figure}[!ht]
  \centering
  \includegraphics[width=0.9\textwidth]{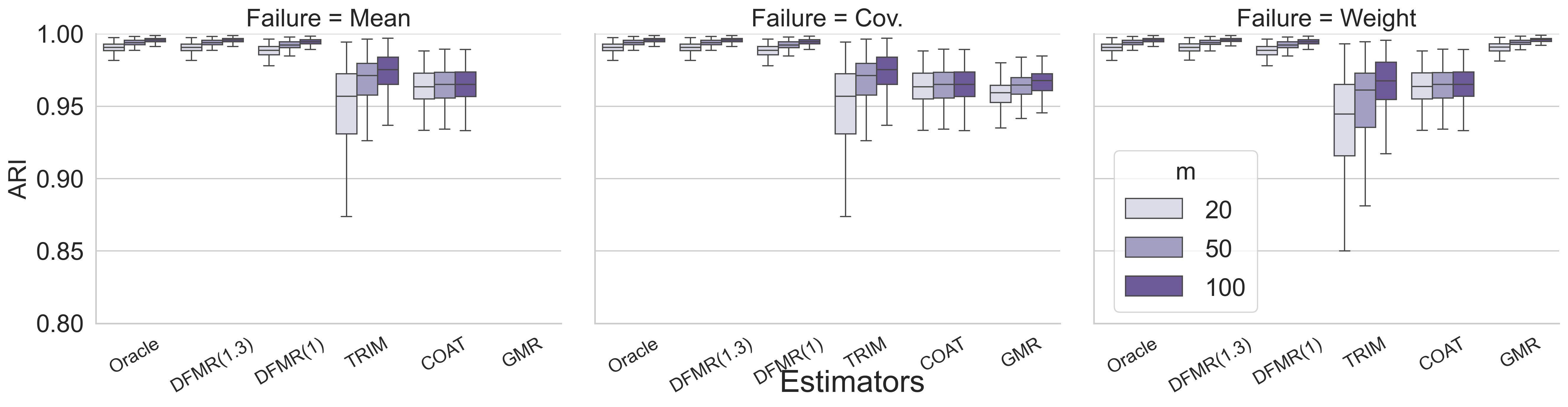}
\caption{\footnotesize{The ARI of different clustering methods based on the full dataset as $m$ varies, with $\alpha=0.1$, and \texttt{MaxOmega}$=0.3$. 
The local sample size $n=5000$, the total sample $N=nm$ increasing with $m$.}
}
\label{fig:ari_fixn_varym}
\end{figure}
\begin{figure}[!ht]
  \centering
\includegraphics[width=0.9\textwidth]{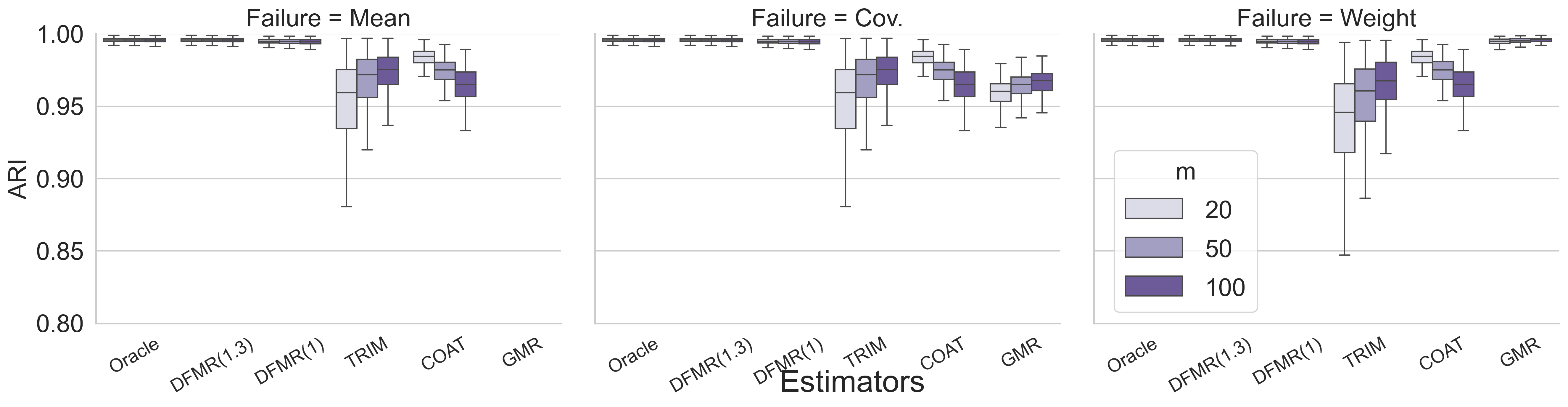}
\caption{\footnotesize{The ARI of different clustering methods based on the full dataset as $m$ varies, with $\alpha=0.1$, and \texttt{MaxOmega}$=0.3$. The total sample size $N=500K$, the local sample $n=N/m$ decreasing with $m$.}
}
\label{fig:ari_fixN_varym}
\end{figure}

Figure~\ref{fig:ari_fixalpha_fixn} corresponds to the settings in Figure~\ref{fig:w1_fixalpha_fixn}.
The y-axis is constrained to $[0.5,1.0]$ to emphasize the differences among estimators. 
The GMR estimator under mean attack has a low ARI and is therefore not shown in the first panel.
\begin{figure}[!ht]
\centering
\includegraphics[width=0.9\textwidth]{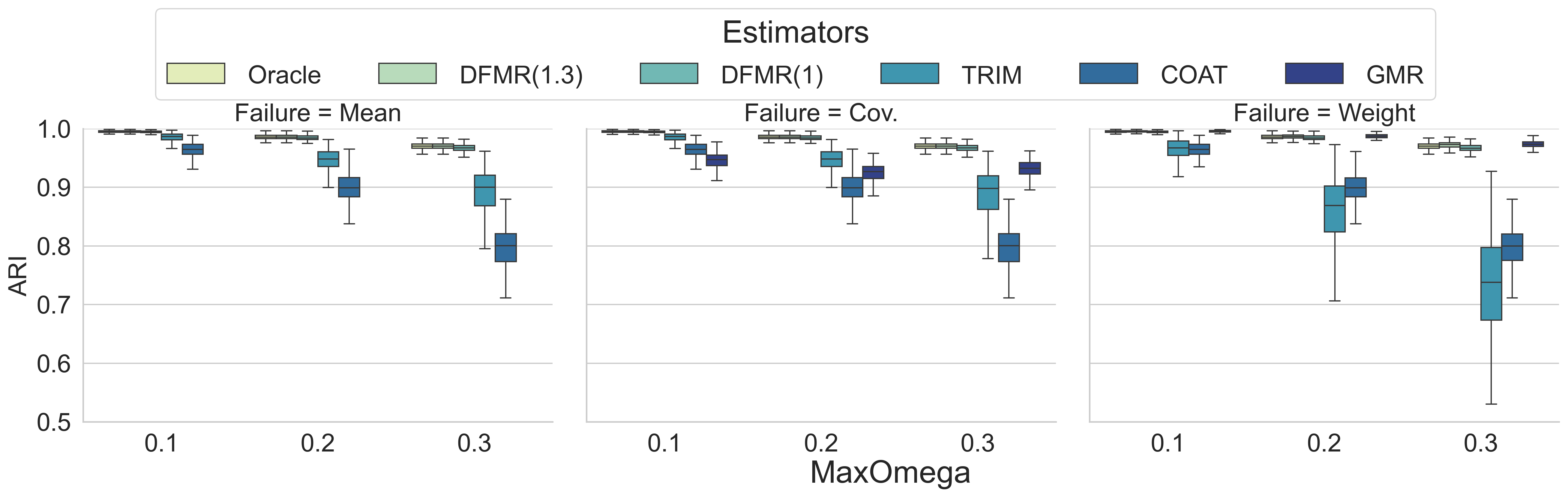}
\caption{\footnotesize{The ARI of different clustering methods based on the full dataset as the degree of overlap \texttt{MaxOmega} varies, with $\alpha=0.2$, $m=100$, and local sample size $n=5000$, under various failure scenarios.}}
\label{fig:ari_fixalpha_fixn}
\end{figure}

\section{Experiments under Gamma mixture}
\label{app:gamma_mixture}
The proposed method is applicable to a wide range of mixture models, extending far beyond Gaussian mixtures. In this section, we present additional empirical results using a gamma mixture.

\subsection{Experimental setup}
Let $f_{\gamma}(x; r, \theta)$ denote the density function of a gamma distribution with shape parameter $r$ and scale parameter $\theta$. 
We consider a finite Gamma mixture model with parameter space $\{(r,\theta): r>0.5, \theta>0\}$.
We generate $R = 300$ independent datasets, each consisting of $ N = 2^{18} $ samples drawn from the following gamma mixture:
\begin{equation}  
\label{eq:gamma_mixture}  
f_{G}(x) = 0.32 f_{\gamma}(x;1,1) + 0.35 f_{\gamma}(x;8,6) + 0.33 f_{\gamma}(x;30,10).
\end{equation}  
The datasets are randomly partitioned into $K = 2^k$ subsets of equal size, where $k \in \{3, \ldots, 8\} $.

To assess robustness against Byzantine failures, we introduce adversarial corruption by varying the proportion of Byzantine machines, denoted by $ \alpha $, from $0.1$ to $0.4$ in increments of $0.1$. 
For each $\alpha $, Byzantine machines are selected at random, and we consider three distinct failure scenarios:
\begin{itemize}
    \item \textbf{Shape failure:} Each estimated shape parameter from the failure machine is replaced with a value drawn independently from $U(1, 10)$.
    \item \textbf{Scale failure:} Each estimated scale parameter from the failure machine is perturbed by additive noise sampled from $U(0,10)$.
    \item \textbf{Weight failure:} The weight corruption mechanism follows the same procedure as in the Gaussian case.
\end{itemize}

\subsection{Distance concentration in Gamma mixtures}
Following the analysis in Section~\ref{sec:exp_distance_concentration}, we generate a dataset of size $N=2^{18}$ from~\eqref{eq:gamma_mixture} and randomly distribute it across $128$ machines. 
The histogram of the generated samples is displayed in the left panel of Figure~\ref{fig:gamma_dist_concentration}, illustrating that the mixture components are well-separated.

For each local dataset, we compute the local pMLE and subsequently introduce Byzantine failures in 20\% of the machines, following the corruption model described earlier. We then compute the COAT estimator $\widehat{G}^{\coat}$ and evaluate the $L^2$ distances between the local estimates $\widetilde{G}_i $ and $\widehat{G}^{\coat}$. 
The right panel of Figure~\ref{fig:gamma_dist_concentration} presents the density functions of these distances, distinguishing between failure-free and corrupted estimates.

\begin{figure}  
\centering  
\includegraphics[width=0.45\textwidth]{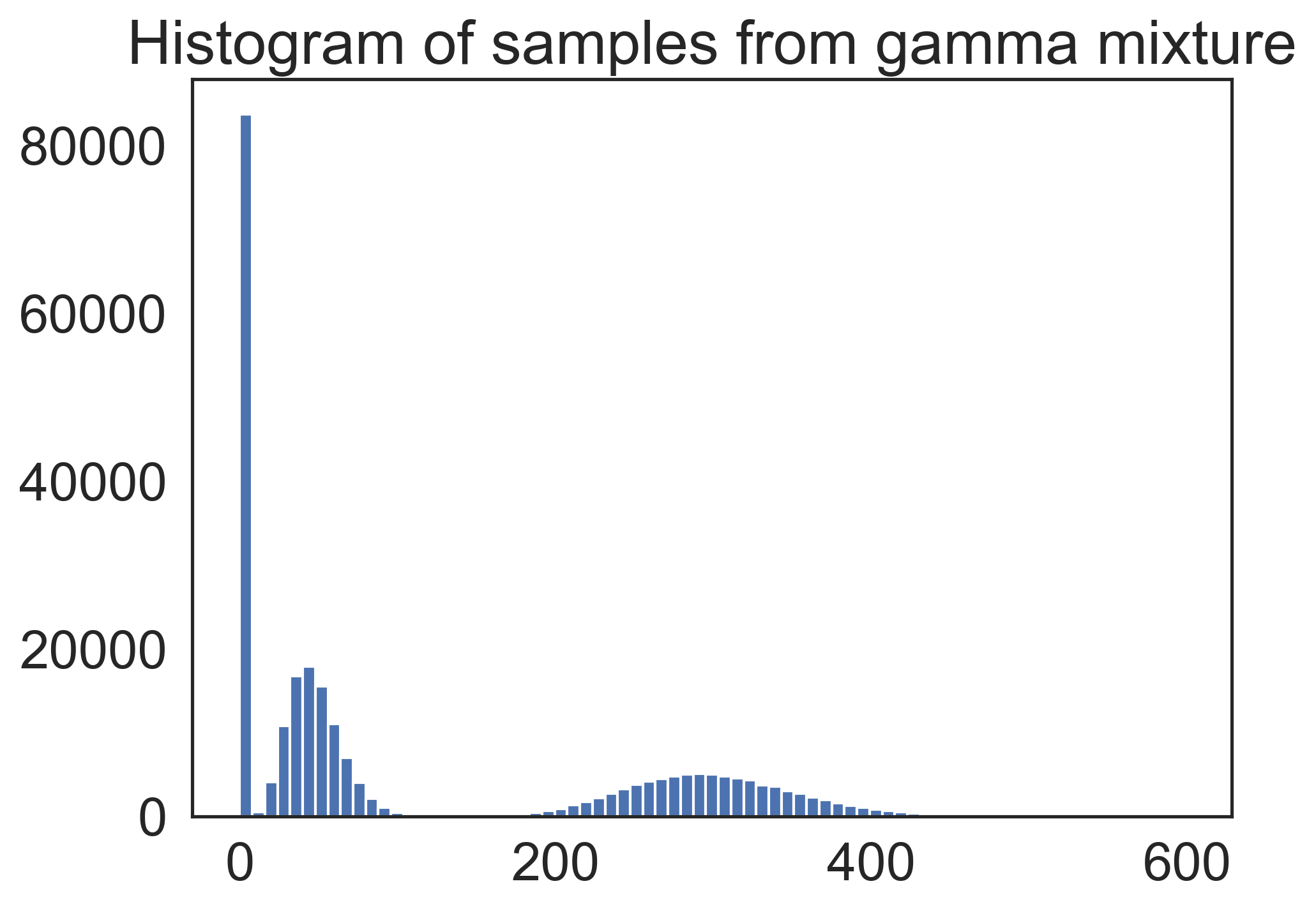}  
\includegraphics[width=0.45\textwidth]{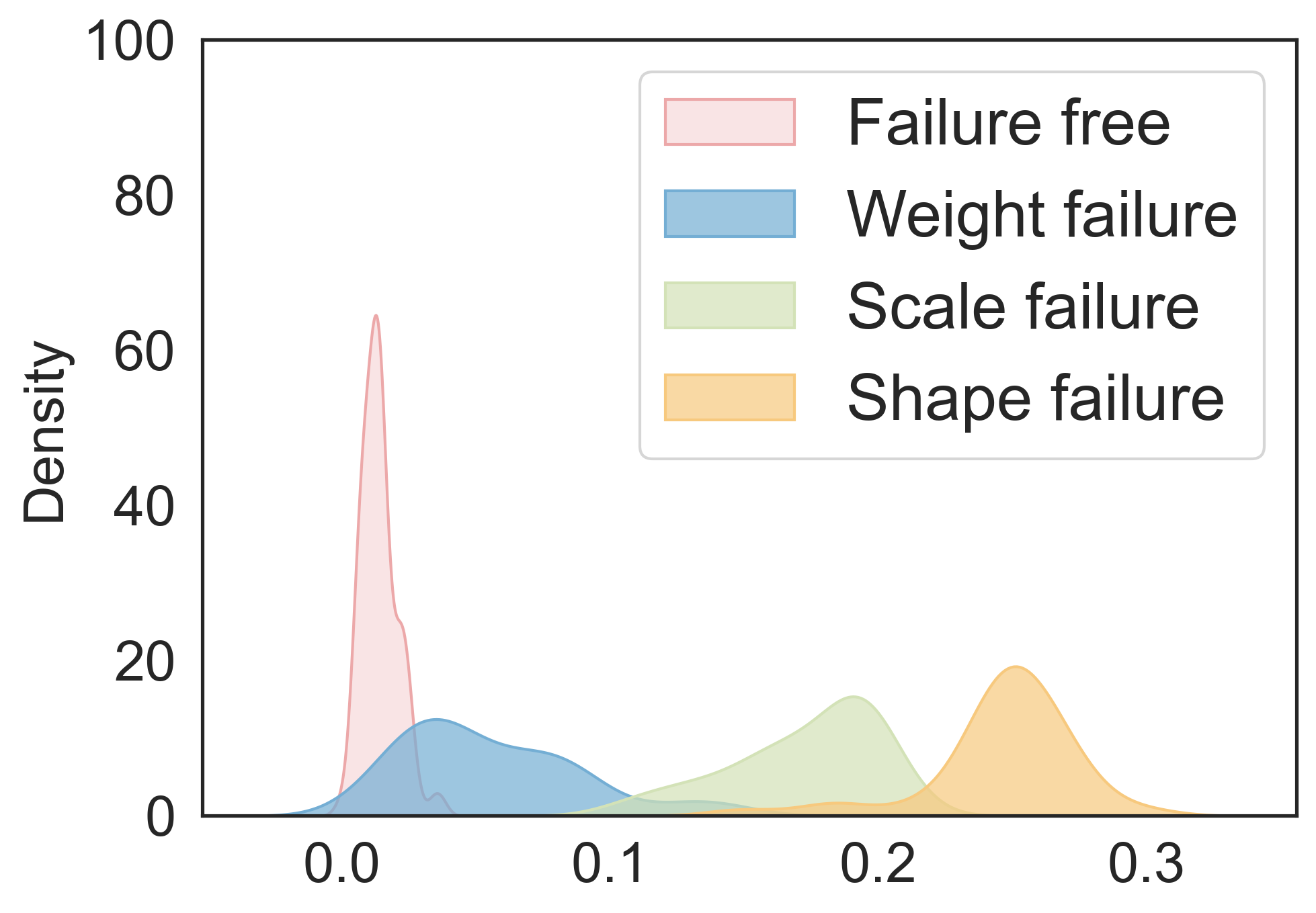}  
\caption{Left: Histogram of samples drawn from the gamma mixture. Right: Distribution of $D(\widetilde{G}_i, \widehat{G}^{\coat})$ under different failure types, represented by distinct colours.}  
\label{fig:gamma_dist_concentration}  
\end{figure}

As observed in the figure, the distance distributions of failure-free and corrupted estimates exhibit distinct modes. However, since the Gamma mixture involves only eight parameters, the concentration of distances is unsurprisingly less pronounced compared to that observed in our experiments with Gaussian mixture models in Section \ref{sec:exp_distance_concentration}, which had at least 29 parameters.

\subsection{Robustness of the inflation factor $\rho$ under Gamma mixtures}
We next evaluate the robustness of the proposed DFMR($\rho$) approach under the Gamma mixture model.
The inflation factor $\rho$ is varied from $1.0$ to $3.0$ in increments of $0.05$.
The top row of Figure~\ref{fig:gamma_EEI_threshold} reports the mean and standard error of the distance between the DFMR($\rho$) estimator and the true parameters across $R = 300$ repetitions, with total sample size $N = 2^{18}$ and failure rate fixed at $\alpha = 0.2$, while varying the number of local machines.
In the bottom row, we use the same $N$, fix the number of local machines at $m = 64$, and vary the failure rate.
\begin{figure}[!htbp]
\centering
\includegraphics[width=0.3\textwidth]{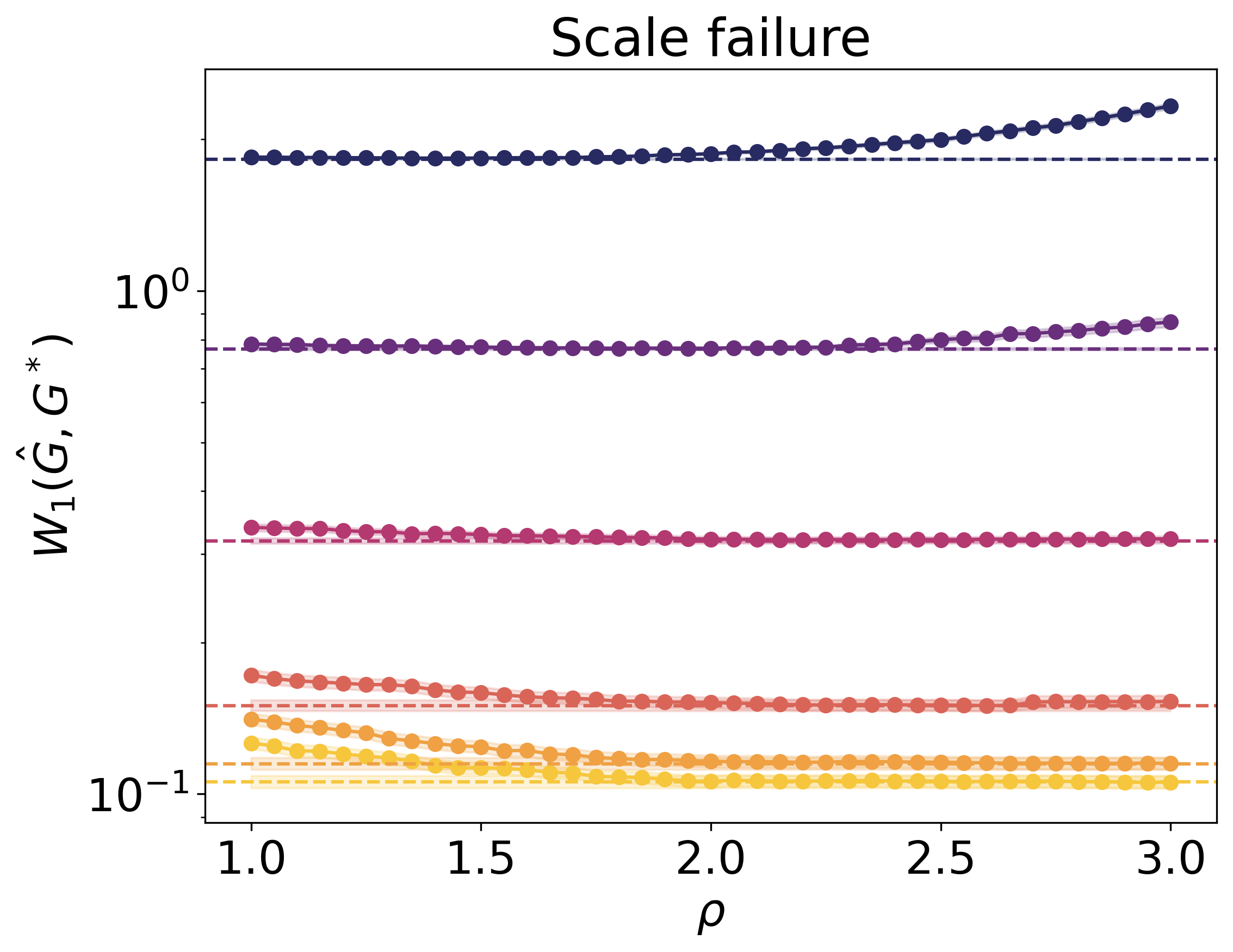}
\includegraphics[width=0.3\textwidth]{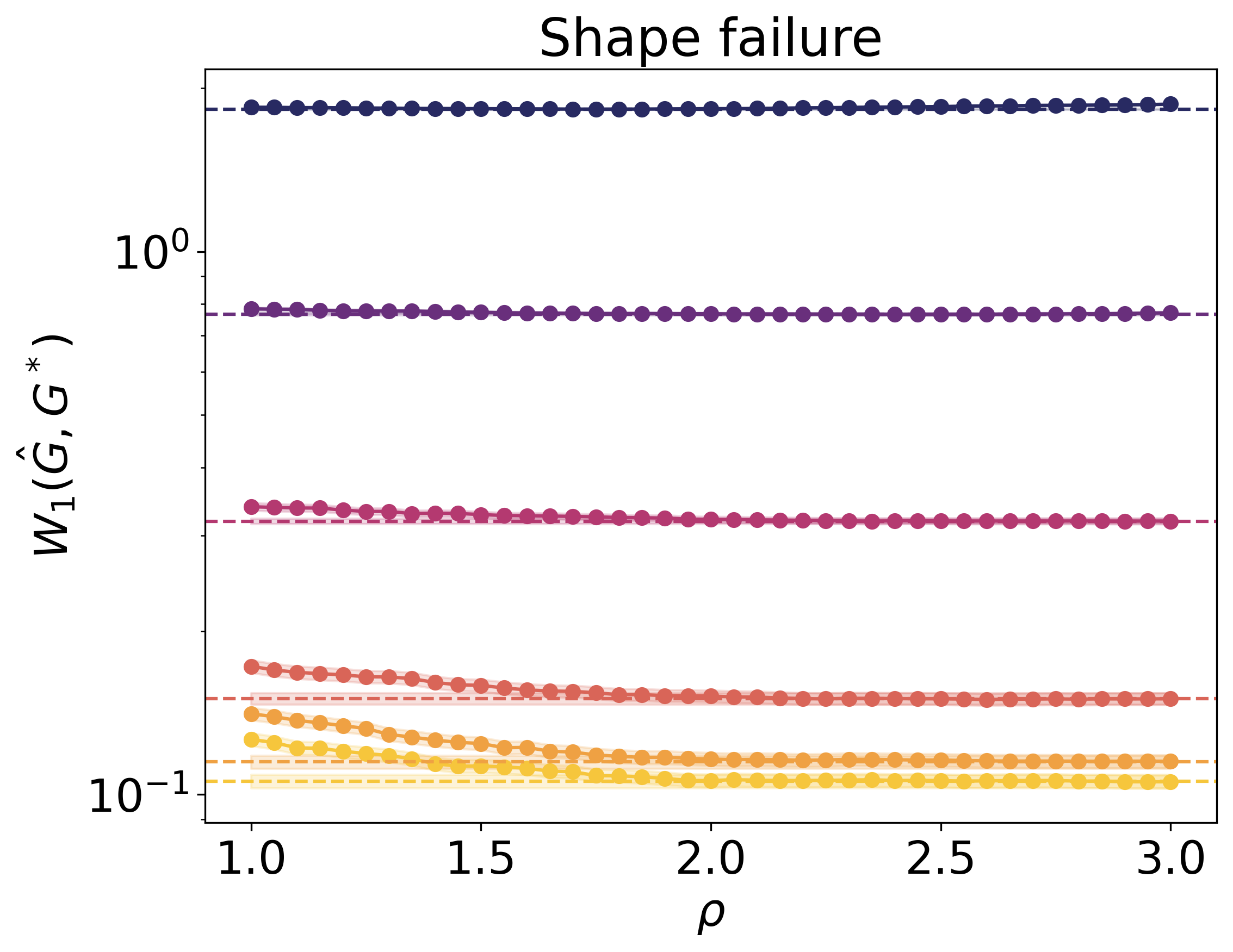}
\includegraphics[width=0.3\textwidth]{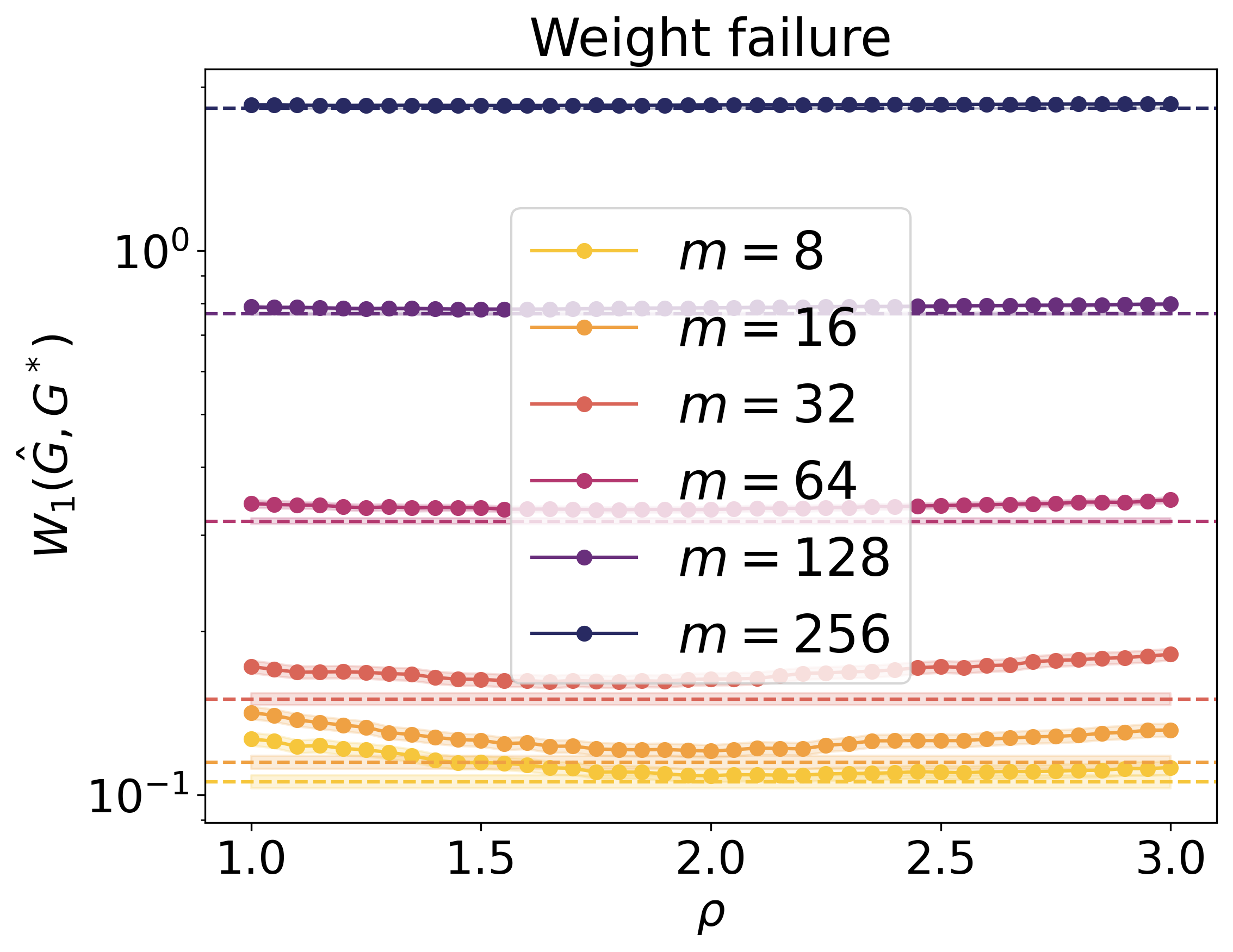}\\
\includegraphics[width=0.3\textwidth]{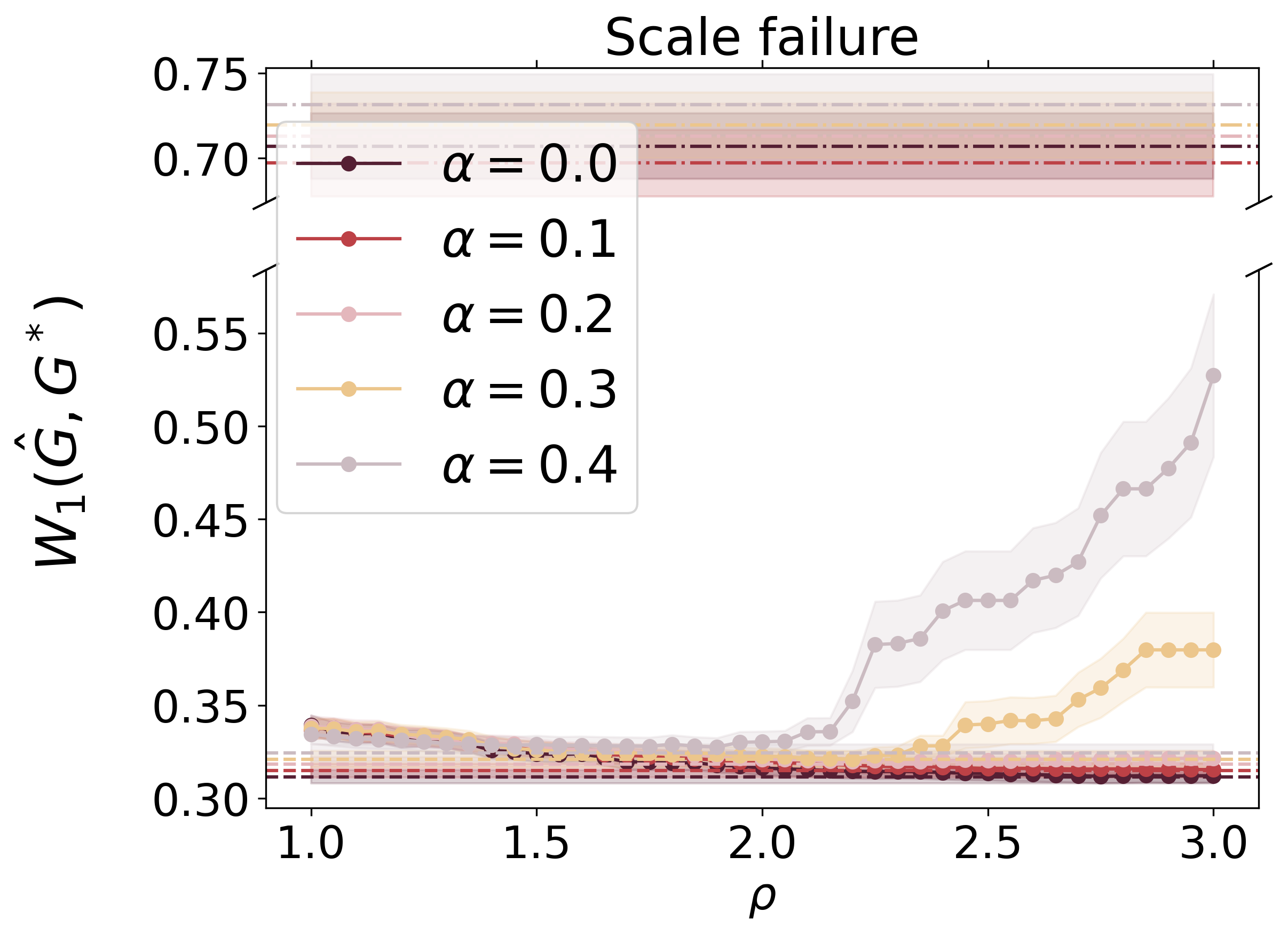}
\includegraphics[width=0.3\textwidth]{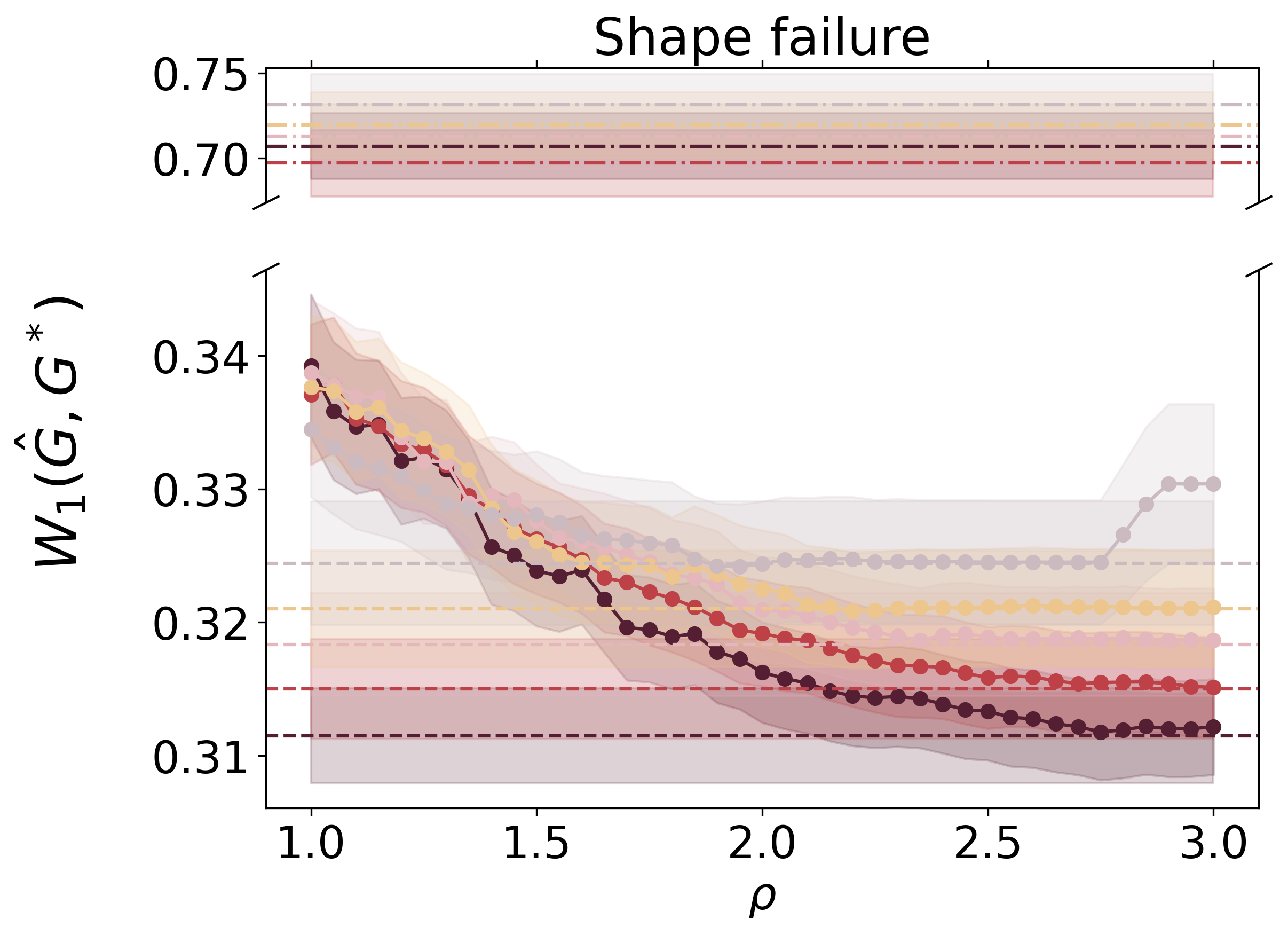}
\includegraphics[width=0.3\textwidth]{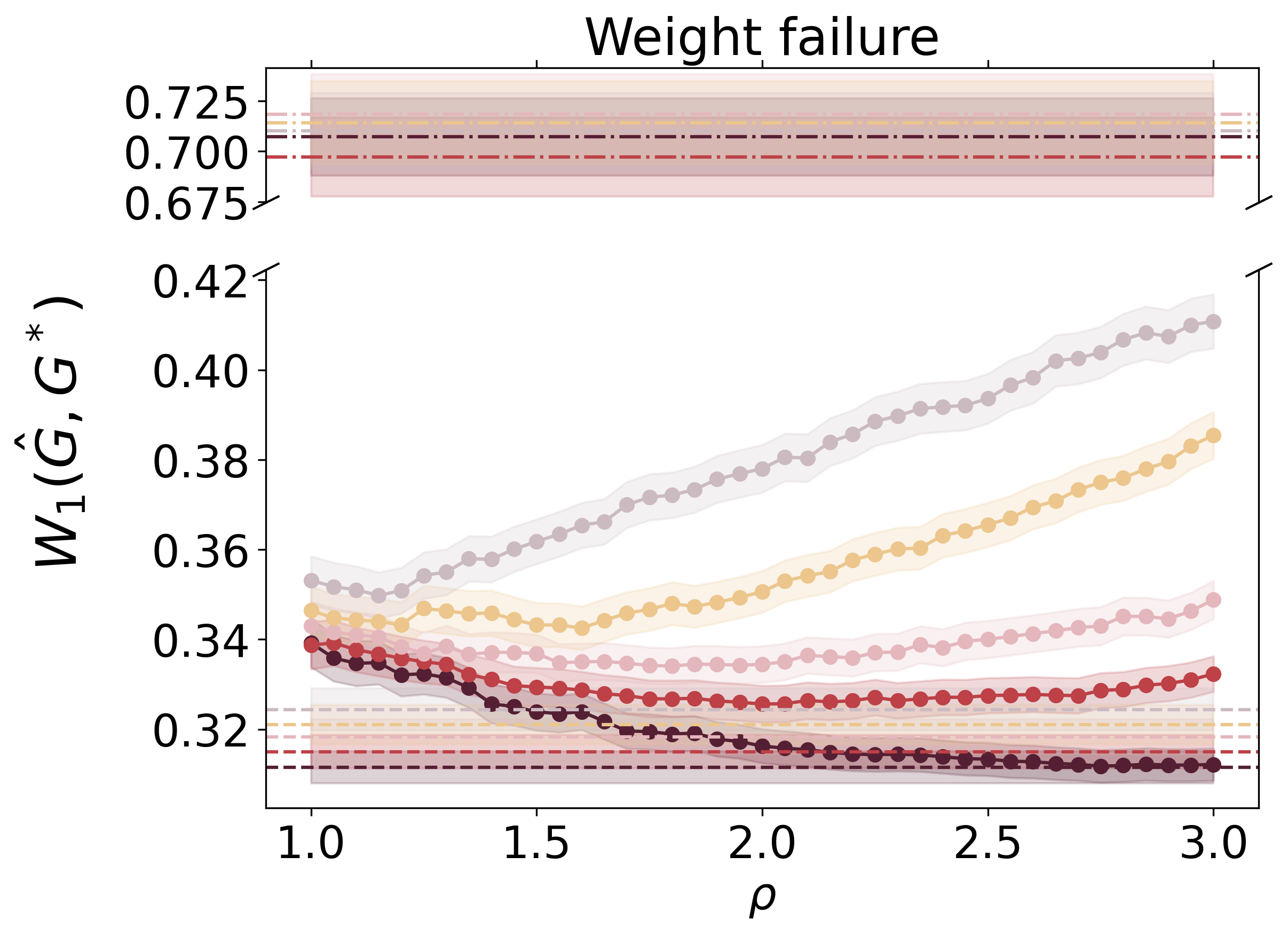}
\caption{The $W_1$ values of the DFMR($\rho$) approach as a function of the inflation factor $\rho$ under varying failure types, failure rates, and numbers of local machines. 
The dotted lines represent the DFMR($\rho$) approach, the dashed lines indicate the performance of the Oracle approach, and the dash-dotted lines correspond to the performance of the COAT approach.}
\label{fig:gamma_EEI_threshold}
\end{figure}
Compared to the Gaussian mixture case, the performance of the DFMR approach is more sensitive to the choice of $\rho$ under gamma mixture. 
Nevertheless, for both shape and scale failures, setting $\rho \approx 2.0$ yields an estimator nearly as efficient as the oracle. 
This empirical choice aligns with the theoretical rate derived in our analysis.

\begin{figure}  
\centering  
\includegraphics[width=0.9\textwidth]{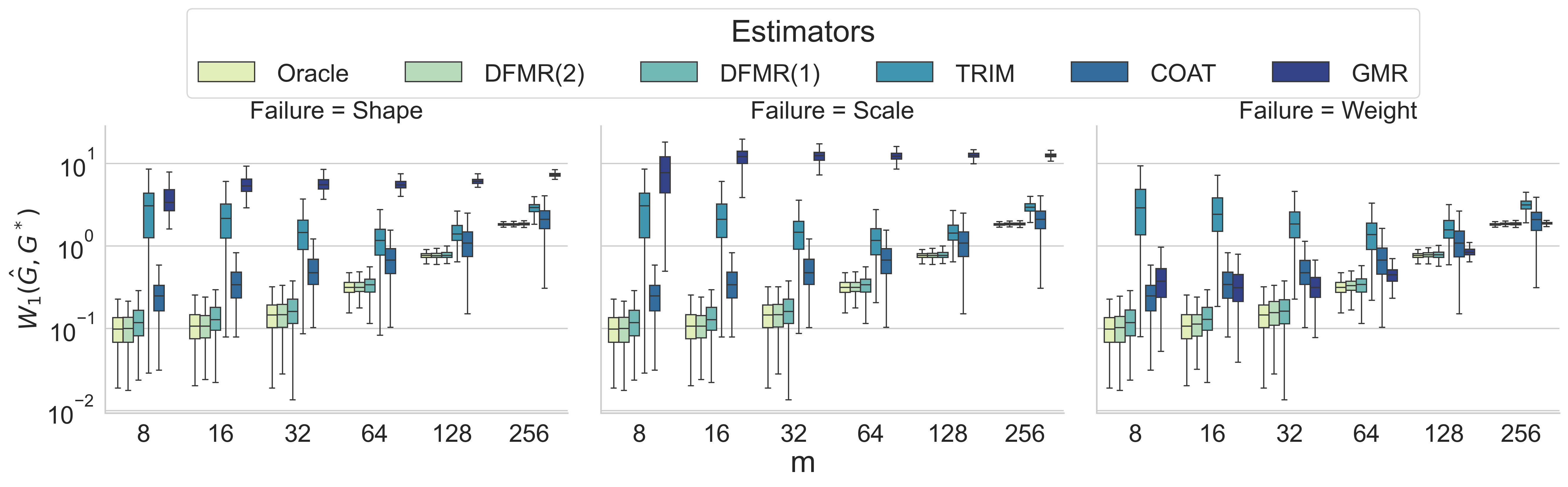}  
\caption{Comparison of different methods as the number of local machines $m$ varies, with $\alpha=0.2$ and total sample size $N=2^{18}$.}  
\label{fig:gamma_final_same_rho}  
\end{figure}
Our findings indicate that DFMR(2) consistently performs on par with the oracle estimator, regardless of failure type or number of local machines. 
Moreover, DFMR(1) alone is nearly as effective as the oracle. 
Both proposed estimators outperform existing baselines.

\subsubsection{Can we improve the performance?}

We observe in our simulation study under the gamma mixture that the DFMR($\rho$) may be better than the oracle estimator for a fixed repetition.  
Followed by this observation, for each repetition, we pick the inflation factor that gives the smallest $W_1$.
\ie let the inflation factor $\rho^{(r)}$ in the $r$th repetition be
\[\rho^{(r)} = \argmin_{\rho \in \{1+0.05t:t=0,1,\ldots,59\}} W_1(\widehat{G}_{\rho}^{\text{DFRM}},G^*)\]
and the corresponding estimator is denoted as $\widehat{G}^{\text{DFRM}}(\rho^{(r)})$.
As a result, the inflation factor may vary across repetitions.  
We then average over repetitions and compute the relative difference in the mean $W_1$ between DFMR($\rho$) and the oracle:
\[\frac{R^{-1}\sum_{r=1}^{R}W_1(\widehat{G}^{\text{DFRM}}(\rho^{(r)}), G^*) - R^{-1}\sum_{r=1}^{R}W_1(\widehat{G}^{\text{OMR}}, G^*)}{R^{-1}\sum_{r=1}^{R}W_1(\widehat{G}^{\text{OMR}}, G^*)}.\]
The values are reported in Table~\ref{tab:improvement}.
\begin{table}[!ht]
\centering
\caption{Relative improvement of DFMR($\rho$)  over oracle.}
\label{tab:improvement}
\begin{tabular}{c|cccccc}
\toprule
\multirow{2}{*}{Failure Rate} & \multicolumn{6}{c}{Number of Machines} \\
\cmidrule(lr){2-7}
 & 8 & 16 & 32 & 64 & 128 & 256 \\
\midrule
& \multicolumn{6}{c}{\textbf{Shape failure}} \\
\midrule
0.0 & -13.48\% & -15.99\% & -17.63\% & -9.85\% & -3.98\% & -1.49\% \\
0.1 & -13.48\% & -14.47\% & -16.83\% & -9.36\% & -3.96\% & -1.50\% \\
0.2 & -11.16\% & -12.83\% & -15.18\% & -8.95\% & -3.81\% & -1.34\% \\
0.3 & -8.22\% & -12.02\% & -12.25\% & -7.44\% & -3.34\% & -1.21\% \\
0.4 & 0.00\% & -6.15\% & -9.08\% & -5.25\% & -2.37\% & -0.89\% \\
\midrule
&\multicolumn{6}{c}{\textbf{Scale failure}} \\
\midrule
0.0 & -13.48\% & -15.99\% & -17.63\% & -9.85\% & -3.98\% & -1.49\% \\
0.1 & -13.48\% & -14.47\% & -16.83\% & -9.36\% & -3.93\% & -1.45\% \\
0.2 & -11.16\% & -12.83\% & -15.16\% & -8.95\% & -3.78\% & -1.21\% \\
0.3 & -8.22\% & -12.02\% & -12.25\% & -7.44\% & -3.25\% & -0.99\% \\
0.4 & 0.00\% & -6.15\% & -9.05\% & -5.21\% & -2.10\% & -0.13\% \\
\midrule
&\multicolumn{6}{c}{\textbf{Weight failure}} \\
\midrule
0.0 & -13.48\% & -15.99\% & -17.63\% & -9.85\% & -3.98\% & -1.49\% \\
0.1 & -13.48\% & -13.99\% & -17.14\% & -8.99\% & -3.82\% & -1.47\% \\
0.2 & -11.14\% & -12.42\% & -14.58\% & -9.04\% & -3.91\% & -0.95\% \\
0.3 & -8.10\% & -11.18\% & -11.98\% & -7.08\% & -3.50\% & -0.77\% \\
0.4 & 1.13\% & -5.52\% & -7.59\% & -5.70\% & -1.43\% & -0.41\% \\
\bottomrule
\end{tabular}
\end{table}
It can be seen that other than weight failure with $\alpha = 0.4$, the DFMR($\rho$) approach improves significantly over the oracle estimator, even when there is no Byzantine failure.  
We do not observe such an improvement for the Gaussian mixture.  
Our conjecture is that the distribution of $\sqrt{n} (\widehat{\mG}_n - \mG^*)$ under the gamma mixture has a heavier tail than the Gaussian mixture. 
Therefore, removing those local estimates far from COAT can help improve efficiency.

\section{Experiments under Poisson mixture}
\label{app:poisson_mixture}
In this section, we present additional empirical results for a widely used discrete mixture model, namely the Poisson mixture.

\subsection{Experimental setup}
Let $f(x;\theta)$ be the probability mass function of a Poisson distribution with mean $\theta$. 
We generate $R=300$ independent datasets, each of size $N=2^{18}$, from the mixture
\begin{equation}
\label{eq:poisson_mixture}
f_G(x) = 0.32 f(x;1) + 0.35 f(x;10) + 0.33 f(x;15).
\end{equation}
Each dataset is randomly partitioned into $K=2^k$ subsets of equal size, for $k = 3,\ldots,8$.

To assess robustness to Byzantine failures, we vary the proportion of Byzantine machines, denoted by $\alpha$, from $0.1$ to $0.4$ in increments of $0.1$. 
For each $\alpha$, the Byzantine machines are selected at random, and we consider two failure scenarios:
\begin{itemize}
    \item \textbf{Rate failure:} Each estimated parameter from a Byzantine machine is replaced by a value drawn independently from $U(1,100)$.
    \item \textbf{Weight failure:} The weights are corrupted using the same mechanism as in the Gaussian case.
\end{itemize}

\subsection{Distance concentration in Poisson mixtures}
Following Section~\ref{sec:exp_distance_concentration}, we generate a dataset of size $N=2^{18}$ from~\eqref{eq:poisson_mixture} and randomly distribute it across $128$ machines. 
The histogram of the samples (left panel of Figure~\ref{fig:poisson_dist_concentration}) shows that one component is well separated, while the other two overlap substantially.

For each local dataset, we compute the local MLE and introduce Byzantine failures in $20\%$ of the machines, following the corruption model described earlier. 
We then compute the COAT estimator $\widehat{G}^{\coat}$ and evaluate the $L^2$ distances between the local estimates $\widetilde{G}_i$ and $\widehat{G}^{\coat}$. 
The right panel of Figure~\ref{fig:poisson_dist_concentration} shows the kernel density estimate of the squared $L^2$ distances, distinguishing between failure-free and corrupted estimates.

\begin{figure}  
\centering  
\includegraphics[width=0.45\textwidth]{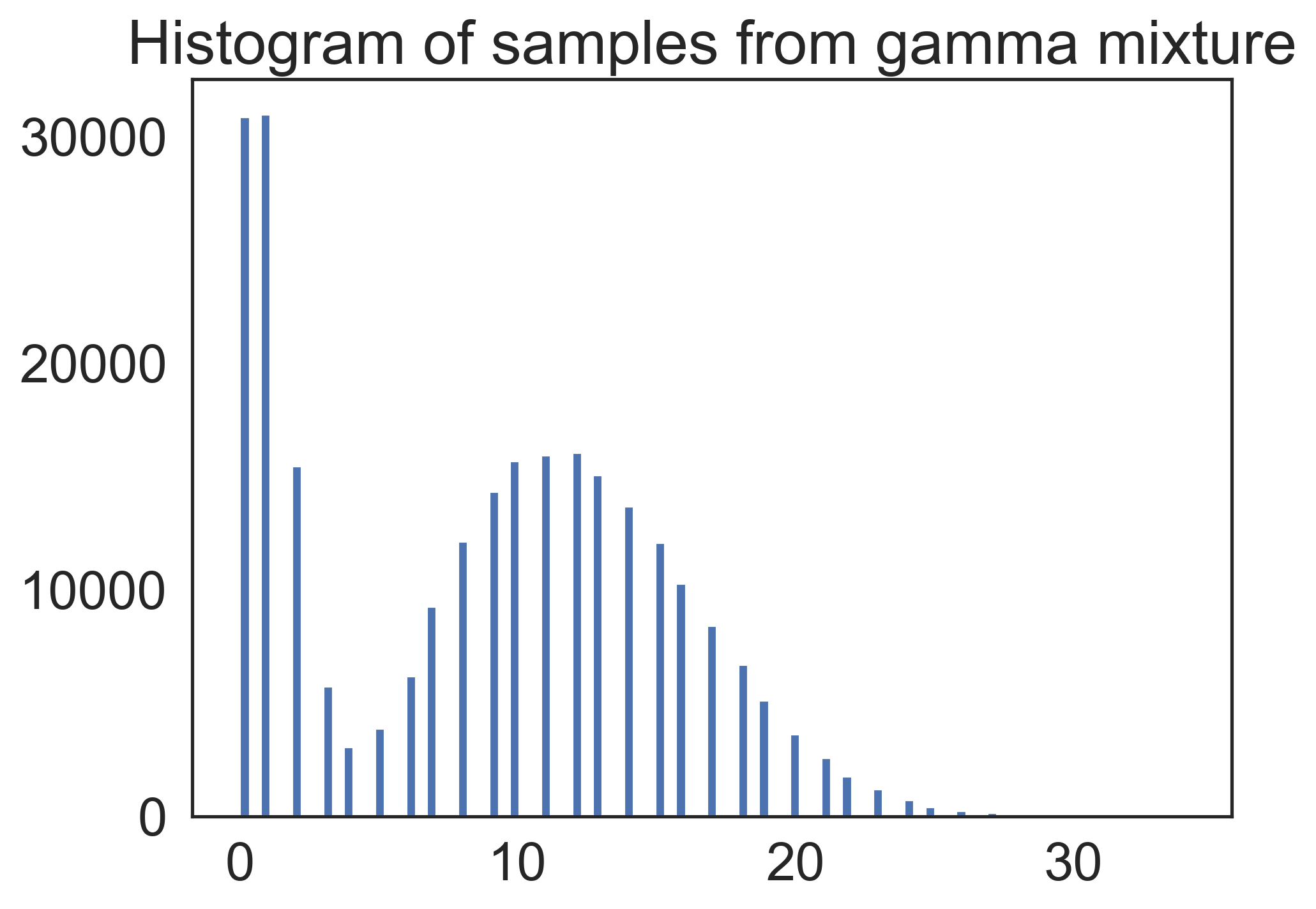}  
\includegraphics[width=0.45\textwidth]{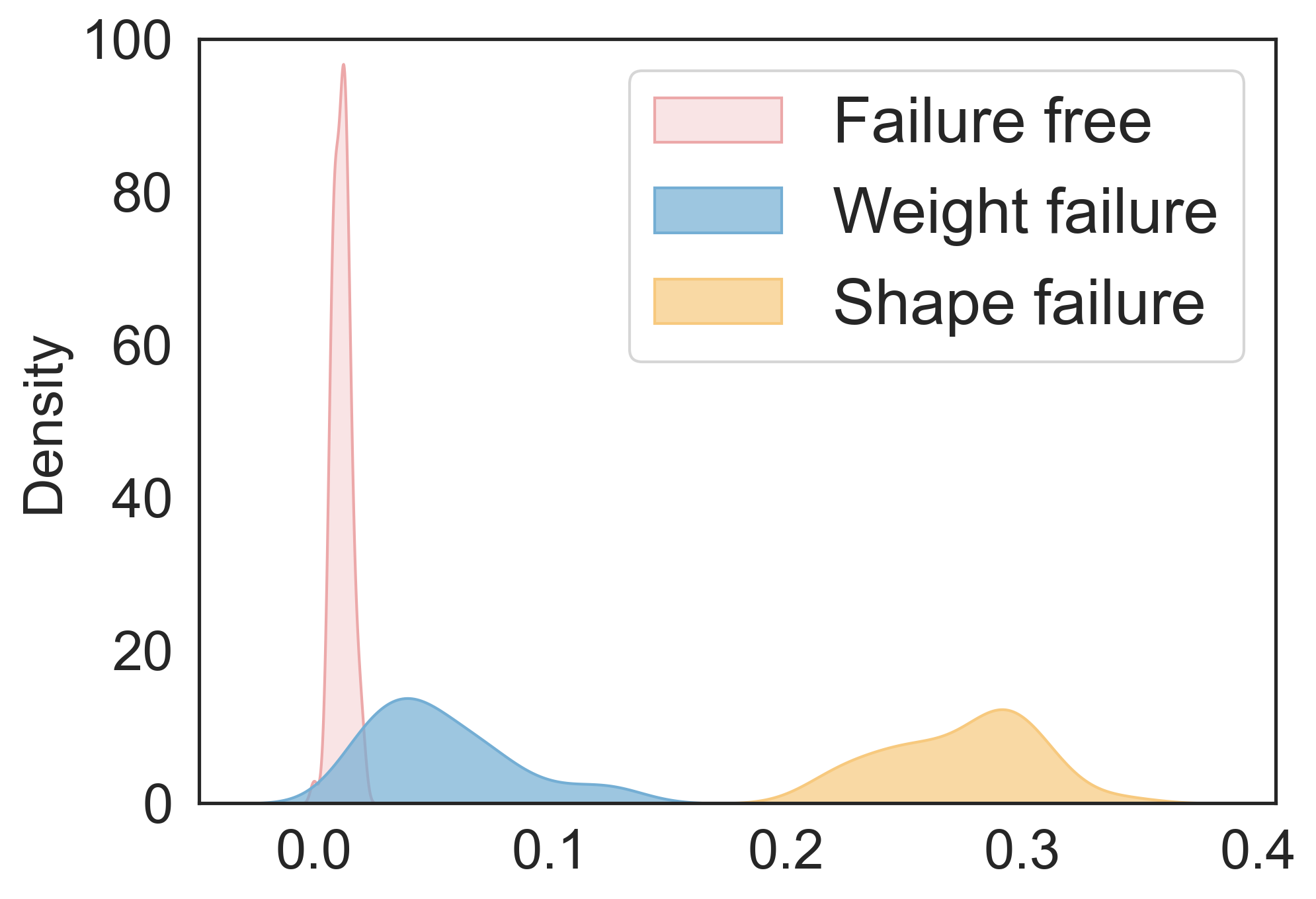} 
\caption{Left: Histogram of samples drawn from the Poisson mixture. Right: Distribution of $D(\widetilde{G}_i, \widehat{G}^{\coat})$ under different failure types, represented by distinct colours.}  
\label{fig:poisson_dist_concentration}  
\end{figure}
As seen in the figure, the distance distributions for failure-free and corrupted estimates have distinct modes. 
However, since the Poisson mixture involves only five parameters, the concentration is less pronounced than in the Gaussian mixture experiments in Section~\ref{sec:exp_distance_concentration}, which involved at least 29 parameters.

\subsection{Robustness of the inflation factor $\rho$ under Poisson mixtures}
We next evaluate the robustness of the DFMR($\rho$) approach under the Poisson mixture.
The inflation factor $\rho$ ranges from $1.0$ to $3.0$ in increments of $0.05$.
The top row of Figure~\ref{fig:poisson_EEI_threshold} shows the mean and standard error of the distance between the DFMR($\rho$) estimator and the true parameters over $R=300$ repetitions, with total sample size $N=2^{18}$ and failure rate fixed at $\alpha=0.2$, while varying the number of machines.
In the bottom row, we keep $N$ fixed, set the number of machines to $m=64$, and vary the failure rate.

We next evaluate the robustness of the proposed DFMR($\rho$) approach under the Poisson mixture.
The inflation factor $\rho$ ranges from $1.0$ to $3.0$ in increments of $0.05$.
The top row of Figure~\ref{fig:poisson_EEI_threshold} shows the mean and standard error of the distance between the DFMR($\rho$) estimator and the true parameters over $R=300$ repetitions, with total sample size $N=2^{18}$ and failure rate fixed at $\alpha=0.2$, while varying the number of machines.
In the bottom row, we keep $N$ fixed, set the number of machines to $m=64$, and vary the failure rate.

\begin{figure}[!htbp]
\centering
\includegraphics[width=0.3\textwidth]{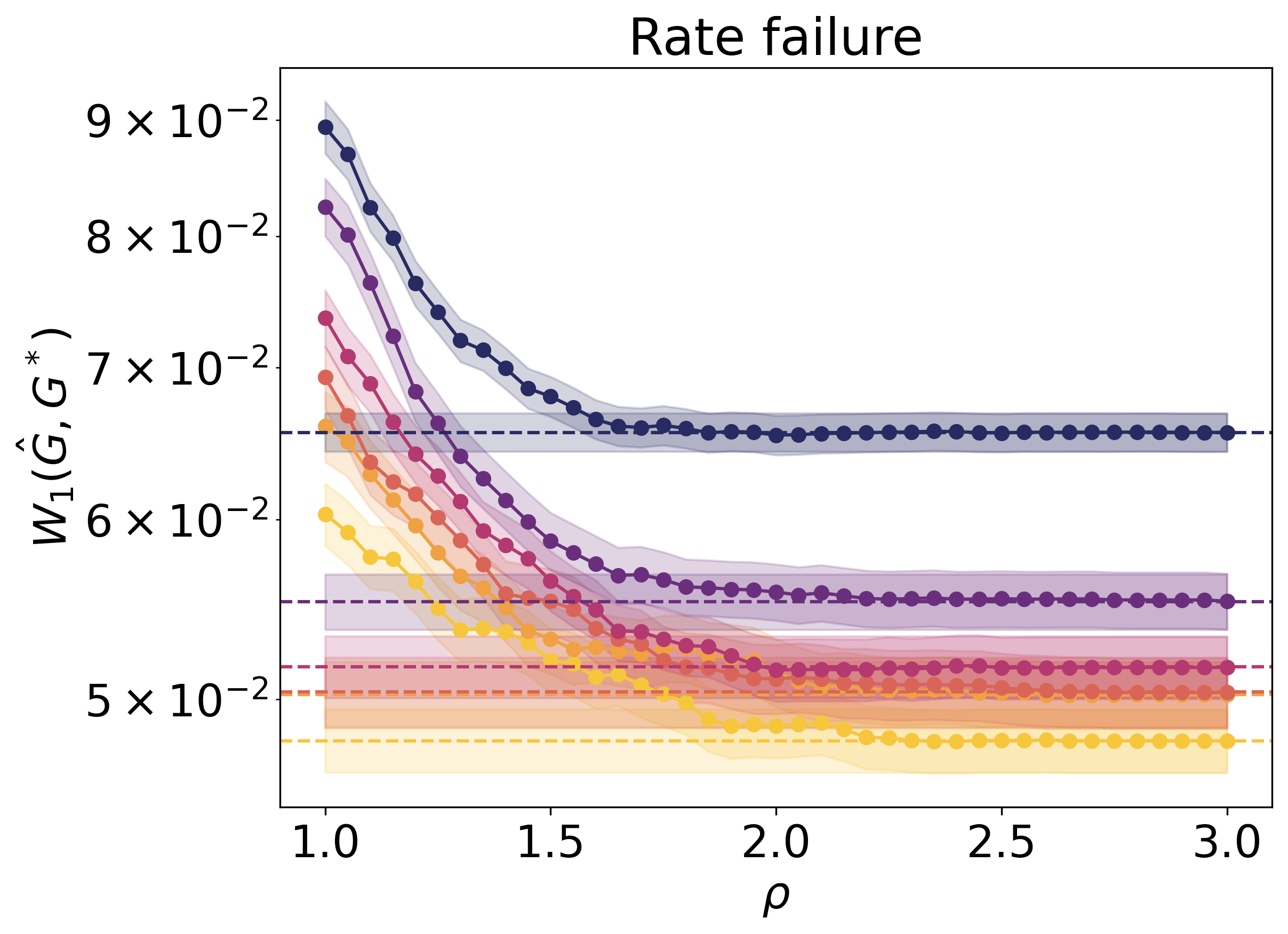}
\includegraphics[width=0.3\textwidth]{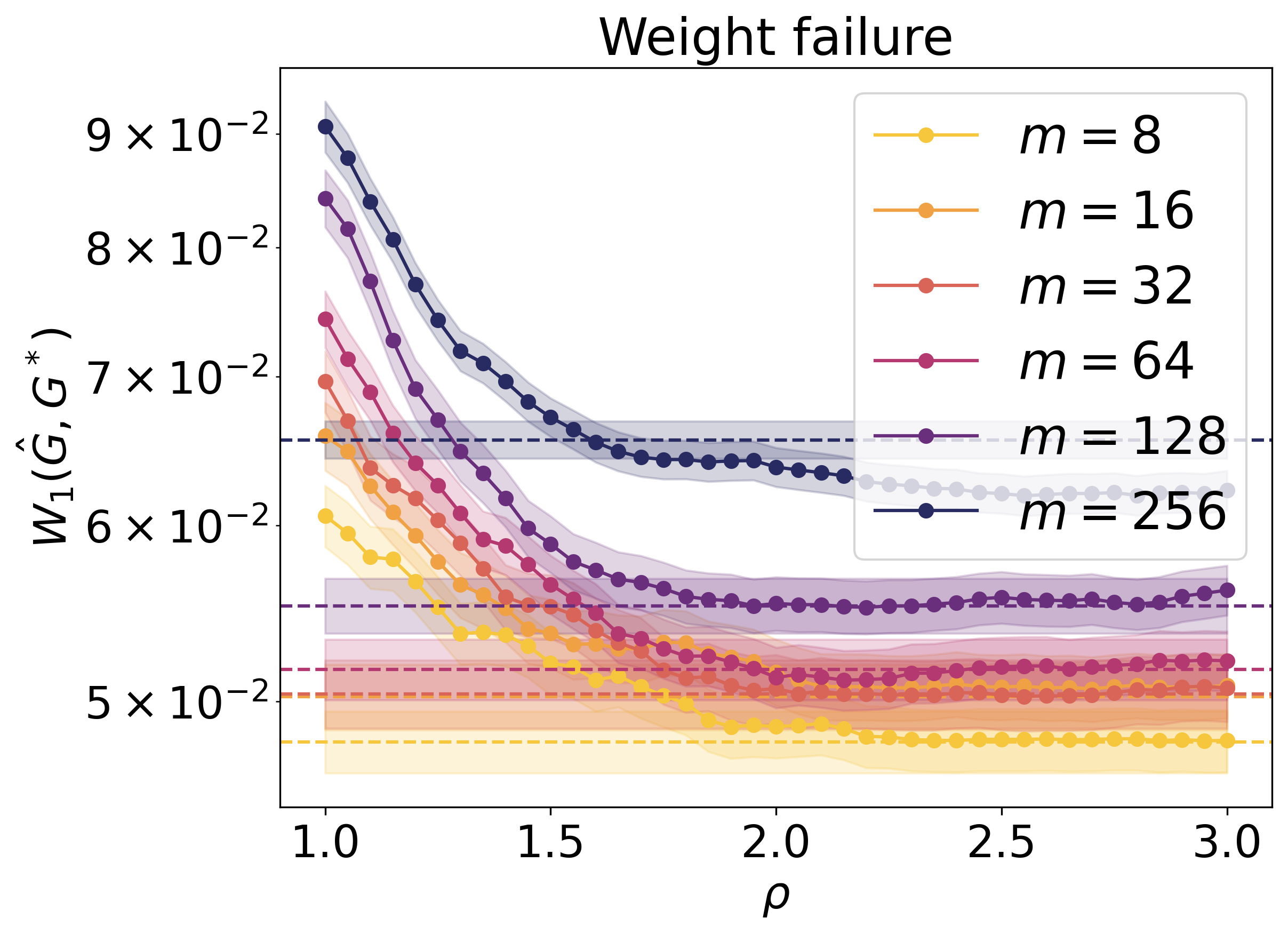}\\
\includegraphics[width=0.3\textwidth]{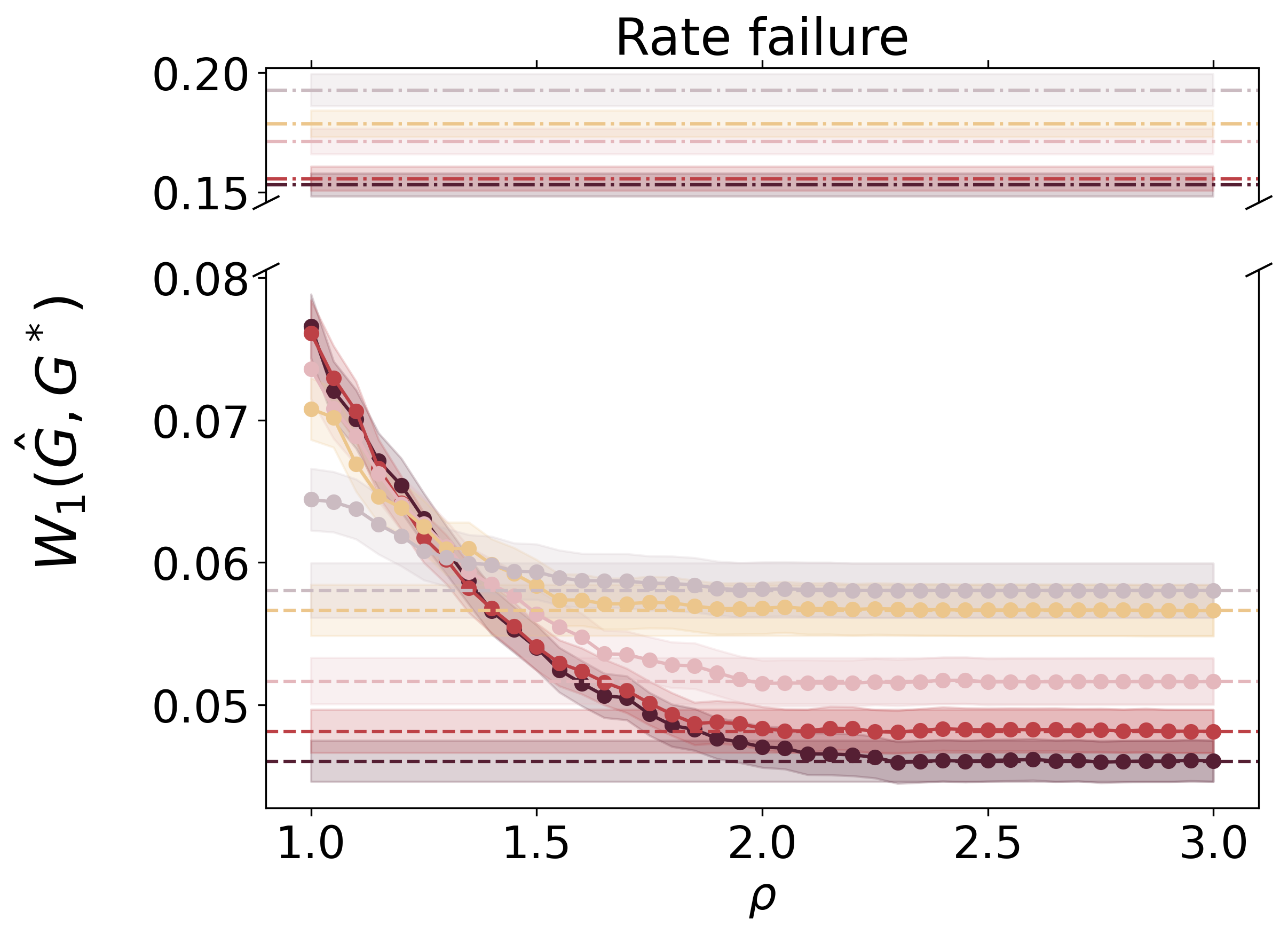}
\includegraphics[width=0.3\textwidth]{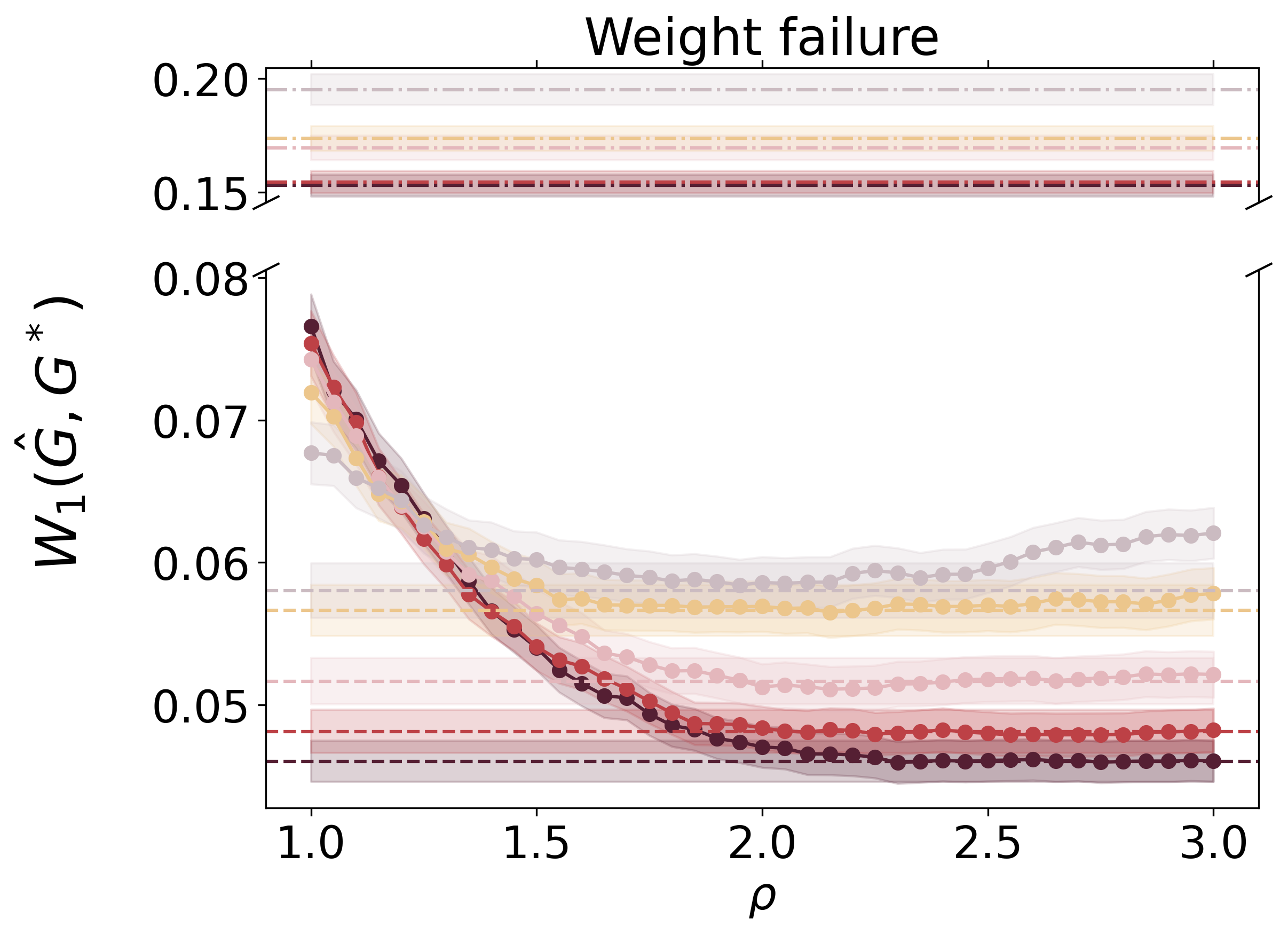}
\caption{
The $W_1$ values of the DFMR($\rho$) approach as a function of the inflation factor $\rho$ under varying failure types, failure rates, and numbers of machines. 
Dotted, dashed, and dash-dotted lines correspond to DFMR($\rho$), Oracle, and COAT, respectively.}
\label{fig:poisson_EEI_threshold}
\end{figure}

As in the Gaussian mixture setting, the performance of DFMR is not highly sensitive to the choice of $\rho$ under Poisson mixtures.
Nevertheless, for both rate and weight failures, setting $\rho \approx 2.0$ yields performance nearly as efficient as the oracle estimator.
This observation is consistent with the theoretical rate established in our analysis.

\begin{figure}  
\centering  
\includegraphics[width=0.9\textwidth]{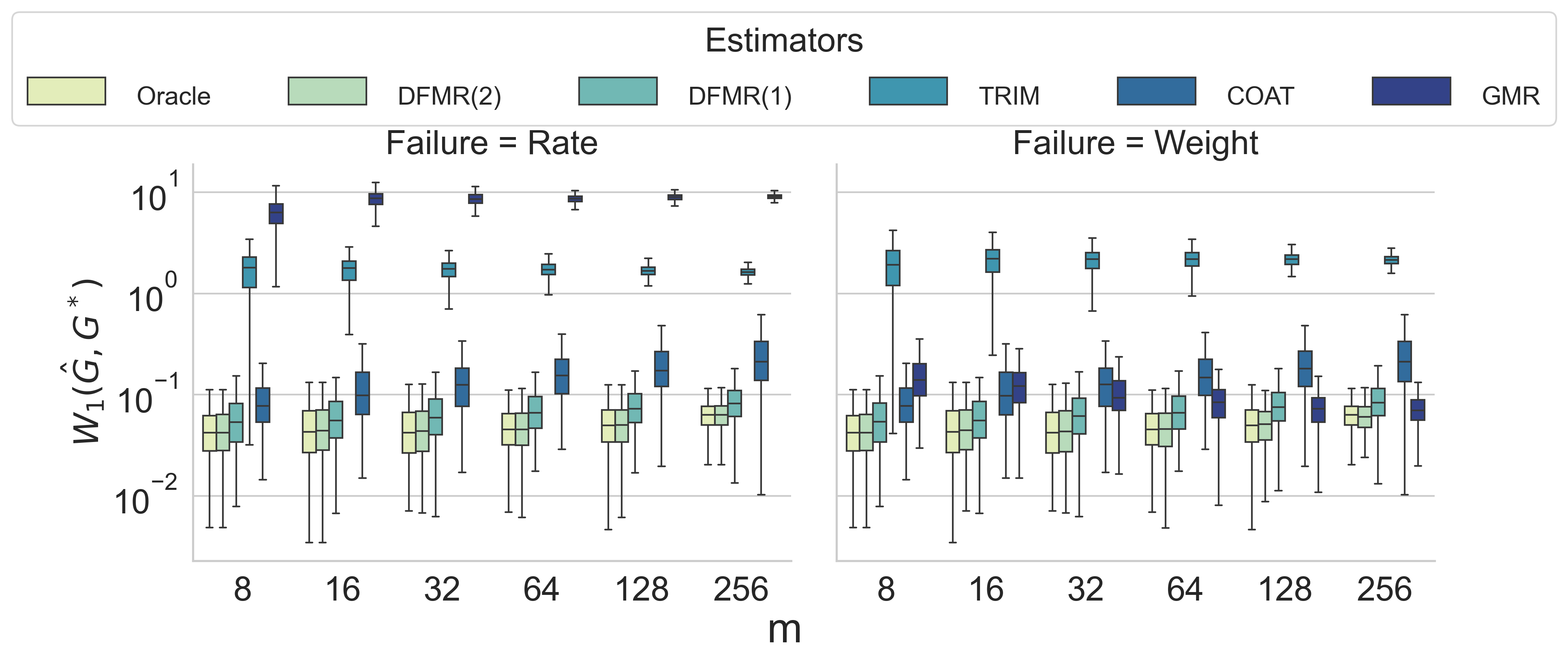}  
\caption{Comparison of methods as the number of machines $m$ varies, with $\alpha=0.2$ and total sample size fixed at $N=2^{18}$.} 
\label{fig:poisson_final_same_rho}  
\end{figure}

Our results show that DFMR(2) performs comparably to the oracle estimator, while DFMR(1) is slightly less accurate.
Both variants outperform existing baselines.

\section{More details on the real data}
\label{app:real_data}

We use the second edition of the NIST dataset.\footnote{Downloaded from \url{https://www.nist.gov/srd/nist-special-database-19}.}
It consists of approximately 4 million images of handwritten digits and letters (0--9, A--Z, and a--z) by different writers.
Each digit or letter has its own directory, with the \emph{by\_class$\slash$hsf\_4} subdirectory designated as the test set, and a \emph{train} subdirectory containing the training data.

Following common practice, we train a $5$-layer convolutional neural network (CNN) to reduce each image to a $d=50$ feature vector of real values. 
The CNN architecture used for dimension reduction in the NIST experiment is described in Table~\ref{tab:architecture}.
The final softmax layer in a CNN can be interpreted as fitting a multinomial logistic regression model on reduced feature space.

\begin{table}[!ht]
\caption{Architecture and layer specifications of the CNN for dimension reduction in the NIST example. $C_{\text{in}}$, $C_{\text{out}}$, $H$, and $W$ denote the input channel size, output channel size, height, and width of the convolutional kernel, respectively.}
\label{tab:architecture}
\centering
\begin{tabular}{c|c|c}
\toprule
{\bf Layer} & {\bf Specification} & {\bf Activation Function} \\
\midrule
Conv2d & $C_{\text{in}}=1$, $C_{\text{out}}=20$, $H=W=5$ & ReLU\\
MaxPool2d & $k=2$ & --\\
Conv2d & $C_{\text{in}}=20$, $C_{\text{out}}=50$, $H=W=5$ & ReLU\\
MaxPool2d & $k=2$ & --\\
Flatten & -- & --\\
Linear & $H_{\text{in}}=800$, $H_{\text{out}}=50$ & ReLU \\
Linear & $H_{\text{in}}=50$, $H_{\text{out}}=10$ & Softmax \\
\bottomrule
\end{tabular}
\end{table}

We implemented the CNN in \emph{PyTorch 2.2.1}~\citep{pytorch2019} and trained it for $10$ epochs on the NIST training dataset. 
We used the SGD optimizer with a learning rate of $0.01$, momentum of $0.9$, and a batch size of $64$. 
After training, we discarded the final layer and employed the resulting CNN to reduce the dimension to $50$ for both the training and test sets. 
These $50$-dimensional numerical features are then utilized to fit Gaussian mixture models on the training set and evaluate performance on the test set.